\newcommand{\notes}{0}
\newcommand{\stoc}{0}
\newcommand{\full}{\text{full version}}
\newcommand{\Full}{\text{Full version}}
\newcommand{\fullcite}{\text{full version}
\cite{LinderRSS25}}

\ifnum\stoc=0
\documentclass[11pt,dvipsnames,table]{article}

\date{}
\author{Ephraim Linder\thanks{
        Boston University. \texttt{\{ejlinder,sofya,ads22\}@bu.edu}. E.L.\ and A.S.\ were supported in part by NSF awards CCF-1763786 and CNS-2120667. S.R.\ was supported in part by the NSF award DMS-2022446.
    } 
    \and Sofya Raskhodnikova\footnotemark[1] 
    \and Adam Smith\footnotemark[1]\ \footnotemark[2] 
    \and Thomas Steinke\thanks{
        Google DeepMind. \texttt{\{adamdsmith,steinke\}@google.com}}}

\else
\documentclass[sigconf,screen,dvipsnames,table,pbalance]{acmart}

\usepackage{microtype}

\copyrightyear{2025}
\acmYear{2025}
\setcopyright{cc}
\setcctype{by}

\acmConference[STOC '25]{Proceedings of the 57th Annual ACM Symposium on Theory of Computing}{June 23--27, 2025}{Prague, Czechia}
\acmBooktitle{Proceedings of the 57th Annual ACM Symposium on Theory of Computing (STOC '25), June 23--27, 2025, Prague, Czechia}

\acmDOI{10.1145/3717823.3718247}
\acmISBN{979-8-4007-1510-5/25/06}

\author{Ephraim Linder}
\affiliation{%
  \institution{Boston University}
  \city{Boston, MA}
  \country{United States}}
\email{ejlinder@bu.edu}

\author{Sofya Raskhodnikova}
\affiliation{%
  \institution{Boston University}
  \city{Boston, MA}
  \country{United States}}
\email{sofya@bu.edu}

\author{Adam Smith}
\affiliation{%
  \institution{Boston University}
  \city{Boston, MA}
  \country{United States}}
\email{ads22@bu.edu}

\author{Thomas Steinke}
\affiliation{%
  \institution{Google DeepMind}
  \city{Mountain View, CA}
  \country{United States}}
\email{steinke@google.com}

\fi

\title{Privately Evaluating Untrusted Black-Box Functions}
\ifnum\stoc=1
\titlenote{A full version of this article is available on \href{https://arxiv.org/abs/2503.19268}{arXiv} \cite{LinderRSS25}.}
\titlenote{E.L.\ and A.S.\ were supported in part by NSF awards CCF-1763786 and CNS-2120667. S.R.\ was supported in part by the NSF award DMS-2022446.}
\fi

\usepackage{graphicx}
\usepackage{macros}

\begin{document}

\ifnum\stoc=0
\maketitle
\fi

\pagenumbering{roman}
\begin{abstract}
We provide tools for sharing sensitive data in situations when the data curator does not know in advance what questions an (untrusted) analyst might want to ask about the data.  
The analyst can specify a program that they want the curator to run on the dataset. We model the program as a black-box function~$f$. 
We study differentially private algorithms, called {\em privacy wrappers}, that, given black-box access to a real-valued function $f$ and a sensitive dataset $x$, output an accurate approximation to $f(x)$. The dataset $x$ is modeled as a 
finite subset of 
a possibly infinite set $\univ$, in which each entry $x$ represents data of one individual. 
A privacy wrapper calls $f$ on the dataset $x$ and on some subsets of $x$ and returns either an approximation to $f(x)$ or a nonresponse symbol $\perp$.  
The wrapper may also use additional information (that is, parameters) provided by the analyst,
but differential privacy is required for {\em all} values of these parameters. Correct setting of these parameters will ensure better accuracy of the privacy wrapper. The bottleneck in the running time of our privacy wrappers is the number of calls to $f$, which we refer to as {\em queries}. Our goal is to design privacy wrappers with high accuracy and small query complexity. 

We introduce a novel setting, called the {\em automated sensitivity detection} setting, where the analyst supplies only the black-box function $f$ and the intended (finite) range of $f$. In contrast, in the previously considered setting, which we refer to as the {\em claimed sensitivity bound} setting, the analyst also supplies additional parameters that describe {\em the sensitivity of} $f$. We design privacy wrappers for both settings and show that our wrappers are nearly optimal in terms of accuracy, locality (i.e., the depth of the local neighborhood of the dataset $x$ they explore), and query complexity. In the {\em claimed sensitivity bound} setting, we  provide the first accuracy guarantees that have no dependence on the size of the universe $\univ$. We also re-interpret and analyze previous constructions in our framework, and use them as comparison points. In addition to addressing the black-box privacy problem, our private mechanisms provide feasibility results for differentially private release of general classes of functions.
\end{abstract}

\ifnum\stoc=1



\maketitle

\fi

\ifnum\stoc=0
\newpage
{\small \tableofcontents}
\newpage
\pagenumbering{arabic}
\fi

\section{Introduction}
The goal of this work is to provide tools for sharing sensitive data in situations when the data curator does not know in advance what questions the (untrusted) analyst might want to ask about the data. Instead of putting the analyst through background checks and monitoring their access to data, we would like to provide an automated way for the analyst to interact with the data. We  allow the analyst to specify a program, modeled  as a black-box function~$f$, that they want the curator to run on the dataset. Black-box specification allows analysts to construct arbitrarily complicated programs; it also enables them to obfuscate their programs and thus hide the analysis they intend to perform from their competitors. 

We study differentially private algorithms that, given a sensitive dataset $x$ and black-box access to a real-valued function $f$, output an accurate approximation to $f(x)$.
Each entry in the dataset $x$ comes from some large (finite or countably infinite) universe $\univ$ and represents data of one individual. Let $\dom$ denote the set of finite subsets of $\univ.$
The dataset $x$
is modeled as %
a member of $\dom.$ Two datasets in $\dom$ are neighbors if they differ by the insertion or removal of one element. 
(Our constructions also apply to a more general setting that covers multisets and node-level privacy in graphs; we expand on this in 
\ifnum\stoc=0 \Cref{sec:posets}
\else the full version~\cite{LinderRSS25}
\fi.)
The algorithm run by the data curator calls $f$ on the dataset $x$ and on some {\em subsets} of $x$ and returns either an approximation to $f(x)$ or a nonresponse symbol $\perp$.  Following Kohli and Laskowski \cite{KohliL23}, we call such an algorithm a \emph{privacy wrapper}
if it is differentially private for {\em all} functions $f$.
A privacy wrapper may also use additional information (that is, parameters) provided by the analyst, %
but differential privacy is required for {\em all} values of these parameters. Correct setting of these parameters will ensure better accuracy of the privacy wrapper. The bottleneck in the running time of our privacy wrappers is the number of calls to $f$, which we refer to as {\em queries}. Our goal is to design privacy wrappers with high accuracy and small query complexity.

We introduce a novel setting, called the {\em automated sensitivity detection} setting, where the analyst supplies only the black-box function $f$ and the intended (finite) range of $f$. In contrast, in the previously considered setting \cite{JhaR13,KohliL23,LangeLRV25}, which we refer to as the {\em claimed sensitivity bound} setting, the analyst also supplies additional parameters that describe {\em the sensitivity of} $f$. We design 
the first privacy wrappers in the {\em automated sensitivity detection} setting and new privacy wrappers in the {\em claimed sensitivity bound} setting. We show that our wrappers are nearly optimal in terms of accuracy and locality (i.e., the depth of the local neighborhood of the dataset $x$ they explore, which is a proxy for query complexity). In the {\em claimed sensitivity bound} setting, we  provide the first accuracy guarantees that have no dependence on the size of the universe $\cU$. We also re-interpret and analyze previous constructions in our framework, tailoring the analysis to our setting, and show that our wrappers improve on all previous constructions.
Finally, we prove tight lower bounds for both settings.

In addition to addressing the black-box privacy problem, our private mechanisms provide feasibility results for differentially private release of general classes of functions. Most work on differential privacy considers the easier %
``white box'' setting, when a complete description of the input function $f$ is available. Our results have applications in this setting. We discuss this perspective on our results in \Cref{sec:applications_intro}.

\subsection{Our Contributions}\label{sec:contributions}
Releasing $f(x)$ privately in the black-box setting is especially challenging when the universe of potential data records is large, since it is not even feasible to query $f$ on all the datasets that differ from $x$ by the addition of one data record. Instead, we consider privacy wrappers that query the function $f$ only on large subsets of $x$, and whose accuracy guarantees depend only on the behavior of the function $f$ on such sets. 
Specifically, let $\lambda\in \N$ be the locality parameter. 
Our algorithms query $f$ only on subsets of the input dataset $x$, obtained from $x$ by removing data of at most $\lambda$ individuals.  
Let $\DN_\lambda(x)$, called the  \emph{$\lambda$-down neighborhood of $x$}, denote the collection of all subsets of $x$ of size at least $|x|-\lambda$:
\begin{align}\label{eq:neighborhood-def}
    \DN_\lambda(x) \defeq \{z\subseteq x: |x\setminus z|\leq \lambda\} \, .
\end{align}
Algorithms that query $f$ only on $\DN_\lambda(x)$ are called {\em $\lambda$-down local}. One of our goals in designing wrappers is to provide small locality $\lambda.$ Observe that the query complexity of a $\lambda$-down local algorithm is  
$O(|x|^{\lambda})$.
We focus on down-local algorithms for two reasons: first, such algorithms handle large or even infinite universe of potential data records. Second, it allows us to design privacy wrappers that are accurate for functions $f$ that are ``well behaved'' on the input dataset $x$ and its subsets, but potentially sensitive to the insertion of new data points (such as outliers).

Our accuracy guarantees (for both settings) are formulated in terms of the behavior of $f$ in the region $\DN_\lambda(x)$. 
Given an accuracy parameter $\alpha>0$ and failure probability $\beta\in(0,1)$, an algorithm $\cA$ is {\em $(\alpha,\beta)$-accurate} for function $f$ on input $x$ if $|\cA(x)-f(x)|\leq \alpha$ with probability at least $1-\beta.$
Let $f(S)$ denote the set $\{f(x): s\in S\}$. The {\em diameter} of a set is the difference between its maximum and minimum.
We achieve small $\alpha$ when the diameter of $f(\DN_\lambda(x))$ is small and when function $f$ is ``smooth'' (i.e., has a small Lipschtiz constant) on the domain $\DN_\lambda(x)$.

Unlike our accuracy guarantees, the privacy guarantees of privacy wrappers are unconditional. As formalized in \Cref{{def: privacy wrapper}}, an $(\eps,\delta)$-privacy wrapper is $(\eps,\delta)$-DP for all black-box functions $f$ and for all settings of parameters. %
See \Cref{def:dp} for the definition of $(\eps,\delta)$-differential privacy (DP). We use both pure DP (when $\delta=0$) and approximate DP (when $\delta>0$).  
 
\subsubsection{Automated Sensitivity Detection}\label{sec:intro-automated-sensitivity-detection}

We consider a novel setting when a bound on sensitivity has to be automatically inferred %
instead of being provided by the analyst.  
In contrast, previous work on the black-box privacy problem  (discussed in \Cref{sec:related}) required the analyst to provide a parameter that bounds the sensitivity of the black-box function. 
We circumvent the need for information about the sensitivity from the analyst 
by phrasing the accuracy guarantee 
of the privacy wrapper in terms of the \emph{down sensitivity}  
(called \emph{deletion sensitivity} in some works) of the function at $x$, defined next.

\begin{definition}[Down sensitivity at specified depth]
    \label{def:down-sensitivity} Let $\lambda\in\N$. The {\em down sensitivity at depth %
    $\lambda$} of 
    a function $f:\dom\to\R$  at point $x\in\dom$ is
    \[
    DS^f_\lambda(x):=\max_{\substack{z\in\DN_\lambda(x)}} |f(x)-f(z)|.
    \]
\end{definition}
The down sensitivity 
differs by at most a factor of 2 from the diameter of the set $f(\DN_\lambda(x)).$

As a simple example, consider the average function $\avg(x) = \frac 1{|x|}\sum_{i=1}^{|x|} x_i$, where $x_i \in \R$. The value of $\avg(x)$ can change arbitrarily with the insertion of a single additional value to $x$. However, the down sensitivity at depth $\lambda$ of $\avg$ at $x$ is finite: for every $x$, it is at most $diameter(x) \cdot \frac{\lambda}{|x|-\lambda}$%
\ifnum\stoc=0\ (see, e.g., \Cref{cor:average}). 
\else .
\fi
We discuss more sophisticated examples in \Cref{sec:applications_intro}.

\definecolor{verybad}{RGB}{255,230,230} 
\definecolor{optimal}{RGB}{230,250,255}
\definecolor{suboptimal}{RGB}{255,230,200}

\begin{table*}[ht]
    \caption{\rm Results in the setting with the automated sensitivity detection (for functions $f:\dom\to\cY$, where the range $\cY\subset\R$ has size $\rangesize$). Locality $\lambda$ is expressed in terms of the privacy parameters $\eps,\delta$, failure probability $\beta$, and range size $\rangesize$; it does not depend on $|\univ|$. Only the third
    row describes a privacy wrapper because previous rows require an assumption on the function $f$ for privacy. }
    \label{table:results-automated-sensitivity-detection}
    \centering
    \begin{tabular}{|c|c||c|c|c|}
        \hline
         Algorithm & Reference& Privacy &  Accuracy $\alpha$& Down Locality $\lambda$  \\
         \hline
         \hline
         \ShI& \cite{FangDY22}&
         \cellcolor{verybad} 
         only for  & %
         $DS^f_\lambda(x)$ & 
         $\lambda_{(\eps,0)}:=O\paren{\dfrac 1 \eps \log \dfrac{\rangesize}{\beta}}$ \\
         & & 
         \cellcolor{verybad}
         monotone $f$ &  & \\   \hline
         {\sf \small Modified}  &\Cref{lem:shifted-inverse-both} &
         \cellcolor{verybad}
         only for &
         $DS^f_\lambda(x)$ & 
         $\lambda_{(\eps,\delta)}:=\dfrac 1 \eps \log\dfrac 1 {\delta} \cdot 2^{O(\log^* \rangesize)}$ \\
         \ShI & & 
         \cellcolor{verybad} 
         monotone $f$ & & \\  
        \hline 
          \hline
         \ASDalgshort & \Cref{thm:generalized-shifted-inverse-mechanism} & for all $f$   & $DS^f_\lambda(x)$ & 
         $\min(\lambda_{(\eps,0)},\lambda_{(\eps,\delta)})$
         \\  
          \hline
          \hline
          Lower Bound & 
          \Cref{cor: locality lower bound}
          & for all
          $f$ &$DS^f_\lambda(x)$&$\Omega\left(\dfrac{1}{\eps}\log\min\left(\dfrac{\rangesize}{\beta},\dfrac1\delta\right)\right)$\\
          \hline
    \end{tabular}
\end{table*}

We provide the first privacy wrapper in the automated sensitivity detection setting. The guarantees of the wrapper, which we dub \ASDalgfull, are stated in \Cref{thm:generalized-shifted-inverse-mechanism}. %
\ASDalgfull{}
works for all functions $f:\dom\to\cY$, where the range $\cY\subseteq \R$ is finite.\footnote{The analyst does need to provide a description of the range $\cY$, but that description need not be trusted—it can be enforced for each query by replacing outputs outside the range. Similarly, the analyst needs to provide an upper bound on the running time of $f$, which can be enforced by terminating the program after the given time.}
For privacy parameters $\eps,\delta$ and failure probability $\beta$, the accuracy of our wrapper is $\alpha=DS^f_\lambda(x)$, where locality $\lambda$ is  $O\big({\frac 1 \eps \log \frac{|\cY|}{\beta}}\big)$ for pure DP and  $O\big(\frac 1 \eps \log\frac 1 {\delta} \log^{*}|\cY|\big)$ for approximate DP. 
In fact, \ASDalgfull{}
returns a value between the minimum and the maximum of $f(z)$ for $z$ in the $\lambda$-down neighborhood of dataset $x$. 
The locality $\lambda$ of our privacy wrapper does not depend on the size of the universe~$\univ.$ 

The starting point for the design of 
\ASDalgfull{}
is the Shifted Inverse Sensitivity Mechanism (\ShI) of Fang, Dong, and Yi \cite{FangDY22}. This is a privacy mechanism for releasing a value of a \emph{monotone} function evaluated on dataset $x$. A function $g:\dom\to\R$ is monotone if $g(x)\leq g(y)$ for all $x,y\in \dom$ such that $x\subset y$.
The \ShI mechanism is $(\eps,0)$-DP for all monotone functions $g$ and has locality that depends logarithmically on the size of the range.  We 
generalize their construction and 
present an approximate-DP variant\footnote{\label{footnote:blog}This variant appeared in a blog post by one of the authors \cite{DPorg-down-sensitivity}.} of the \ShI mechanism that has only $2^{\log^*}$ dependence on the size of the range (\Cref{lem:shifted-inverse-both}). 
Both versions of \ShI release the value of a monotone function $g$ at the data set $x$ privately, with error bounded by $DS^g_\lambda(x)$.
However, they are {\em not} privacy wrappers, because 
they can violate differential privacy for general, black-box functions---the privacy proof relies crucially on the promise that $g$ is monotone.

The wrapper we design, \ASDalgfull, is private for all functions $f: \dom\to\cY$ with finite range $\cY$ and extends the accuracy guarantees of (both versions of) \ShI from monotone functions to all functions.
It works by running \ShI (or our modification thereof) on a carefully selected ``monotonization'' of function $f$---see \Cref{sec:techniques} for more detail.

To complete the picture for automated sensitivity detection setting, we  provide a lower bound \ifnum\stoc=0 (in \Cref{cor: locality lower bound})\fi~ on the locality $\lambda$ of any privacy wrapper that achieves accuracy equal to down sensitivity at depth $\lambda$. Our lower bound matches the first term (corresponding to the pure DP) in the guarantee for \ASDalgshort{} and nearly matches the second term (corresponding to approximate DP). 
We also provide a query complexity lower bound, \Cref{thm: query lower bound} 
(and \Cref{rem: query lower bound}), 
showing that the query complexity we achieve cannot be significantly improved, even if the locality requirement is relaxed. 
\ifnum\stoc=1
Proofs of our lower bounds appear in the \fullcite.
\fi
Our guarantees for \ASDalgshort{} are  compared to guarantees for both variants of \ShI and the locality lower bound in \Cref{table:results-automated-sensitivity-detection}.

\subsubsection{Privacy Wrappers with Claimed Sensitivity Bound}\label{sec:wrapper-with-provided-c}
We also investigate the claimed sensitivity bound setting, which has been addressed in previous work, where  the analyst provides a sensitivity bound $c$  along with a black-box function $f$.

Sensitivity and the related notion of Lipschitz functions play a fundamental role in private data analysis. Intuitively, sensitivity measures how small modifications of the dataset affect the value of the function. 
Given a constant $c>0$ and a domain $D\subseteq \dom$, a function $f:\dom\to\R$ is \emph{$c$-Lipschitz over $D$} if $|f(x)-f(y)|\leq c$ for all neighbors $x,y\in D$.
The smallest constant $c$ for which function $f$ is $c$-Lipschitz on $\dom$ is called the {\em Lipschitz constant} or the {\em (global) sensitivity} of $f$. 
A basic result is that the 
Laplace mechanism, which
releases $f(x)$ with Laplace noise scaled to $\frac c \eps$, is ($\eps,0)$-differentially private \cite{DworkMNS06}.
Most work on differential privacy considers the ``white box'' case, where a complete description of $f$ is available, and one analyzes sensitivity analytically.

The problem of privately evaluating black-box functions was first considered by Jha and Raskhodnikova~\cite{JhaR13}. In their setting, 
in addition to the black-box function $f$, the analyst provides\footnote{Equivalently, the analyst provides a rescaled function $f/c$ instead of $f$, and the curator presumes a sensitivity bound of 1.}  a sensitivity parameter $c$.
The data curator must guarantee differential privacy whether or not the provided bound is correct, but the mechanism's accuracy depends on the bound being correct and as tight as possible. 
Jha and Raskhodnikova \cite{JhaR13} and follow-up work~\cite{AwasthiJMR14,LangeLRV25} handle the case where the data universe $\univ$ is finite, and they investigate the function's behavior in a ball around the input that includes both insertions and deletions. This means the data universe has to be small to achieve reasonable query complexity, and also that the function must be robust to both insertions and deletions of data points.

\bgroup
\def\arraystretch{1}
\begin{table*}[ht]
\caption{\rm 
    Results for the {\em claimed sensitivity bound} setting %
    (for functions $f:\dom\to\R$). Privacy guarantees hold for all settings of parameters. When $f$ is $c$-Lipschitz on $\DN_\lambda(x)$, each wrapper %
    returns $f(x)+Z$ where $\Ex [Z] = 0$ and $Z$ is a Laplace distribution. %
    The table lists the scale parameter (roughly, standard deviation) of the noise variable. This noise scale and locality $\lambda$ are expressed in terms of the privacy parameters $\eps$ and $\delta$. %
    The lower bound on down locality is stated for mechanisms with (optimal) error scale $O(c/\eps)$. 
    }
    \label{table:results}
    \centering
    \begin{tabular}{|c|c||c|c|c|c|}
    \hline
         \multirow{2}{*}{Algorithm} &\multirow{2}{*}{Reference} &Privacy& Accuracy  & Subexponential  & Down Locality   \\
         &&Guarantee& Assumption& Error Scale & $\lambda$\\
         \hline
         \hline
         \multirow{2}{*}{Cummings-Durfee} & \multirow{2}{*}{\cite{CummingsD20}}& \multirow{2}{*}{$(\eps,0)$-DP} & $c$-Lipschitz %
         &\multirow{2}{*}{$\dfrac{c}{\eps}$} & 
          \cellcolor{suboptimal}
          $|x|$ \\
           &&&  on %
          $\cP(x)$ && \cellcolor{suboptimal}   \\           
          \hline 
          (Our analysis of)  &
          \ifnum\stoc=0 Prop.~\ref{prop:KL-analysis}
          \else Full version 
          \fi
          & $(\eps,\delta)$-DP& $c$-Lipschitz %
          &
          \cellcolor{suboptimal}
          $\Theta\Bparen{\lambda \cdot \dfrac{1}{\eps}}$ 
          &
          $O\paren{\dfrac {1} \eps \log \dfrac {1}{\delta }}$ \\
          TAHOE \cite{KohliL23} & \ifnum\stoc=1 of our work
\cite{LinderRSS25}\fi&& on $\DN_\lambda(x)$& 
          \cellcolor{suboptimal}
          $=\Theta\Bparen{ \dfrac{1}{\eps^2}\log\dfrac1\delta}$ &
          \\
          \hline
          \hline 
        \algSE\ & \ifnum\stoc=1 Theorem\else Thm.\fi
         ~\ref{thm: subset extension}
        &$(\eps,\delta)$-DP& $c$-Lipschitz
          &
          $O\Bparen{\dfrac{c}{\eps}}$ & 
          $O\paren{\dfrac {1} \eps \log \dfrac {1}{\delta }}$ \\
           && &on $\DN_\lambda(x)$ & & \\ 
          \hline
          \hline
         Lower Bounds&  \cite{GhoshRS09} and& $(\eps,\delta)$-DP&
          $c$-Lipschitz %
          & $\dfrac{c}{\eps}$ & 
        $\Omega\Bparen{\dfrac{1}{\eps}\log\dfrac1\delta}$
          \\
          &\ifnum\stoc=1 Theorems\else Thms.\fi ~\ref{thm: locality lower bound},\ref{thm: query lower bound}
          & & & 
           \cite{GhoshRS09} & 
           \ifnum\stoc=1 Theorems\else Thms.\fi~\ref{thm: locality lower bound},\ref{thm: query lower bound}
          \\   
          \hline
    \end{tabular}
    
\end{table*}
\egroup

We consider instead the setting introduced by Kohli and Laskowski \cite{KohliL23}, where the curator is limited to evaluating the black-box at subsets of the actual input $x$.  The privacy wrapper for this setting in  \cite{KohliL23}, called TAHOE, has locality $\lambda$ independent of the data set size~$|x|$. In  \cite{KohliL23}, the accuracy of TAHOE is analyzed for some special cases, under distributional assumptions, and empirically. We aim to get a privacy wrapper in this setting that is private for all $f$ and $c$, and satisfies the following type of accuracy: For every $f$ and $x$ such that $f$ is $c$-Lipschitz on $\DN_\lambda(x)$ (where $\lambda$ varies by mechanism), the wrapper outputs $f(x) +Z$ where $Z$ is drawn from a Laplace distribution. Our goal is to simultaneously minimize $\lambda$ and the scale of $Z$.

We provide a novel privacy wrapper, \algSE, for this setting with accuracy and locality guarantees that are each optimal up to constant factors. \Cref{table:results} summarizes our results and compares them to the guarantees of existing privacy mechanisms that query $f$ only on the subsets of dataset $x$, which we summarize briefly here.
In particular, we provide a self-contained accuracy analysis of TAHOE in \ifnum\stoc=0 \Cref{prop:KL-analysis}\else the \fullcite\fi. It adds Laplace noise that is larger by a factor of $\log(1/\delta)/\eps$ than the scale $c/\eps$ that can be used in the white-box setting when the function $f$ is promised to be $c$-Lipschitz on its entire domain. (This latter scale is known to be optimal, even for the special case where $f$ is a simple counting query~\cite{GhoshRS09}.)

We observe that the literature on \emph{Lipschitz extensions}---generally thought of as ``white box'' objects---also provides a different, incomparable privacy wrapper for this black-box setting. Specifically, given $c$ and $f$, Cummings and Durfee~\cite{CummingsD20} construct a function $f_c$ that is $c$-Lipschitz (everywhere), and for which $f_c(x)$ is computable based only on the values of $f$ on the subsets of $x$. Furthermore, $f_c(x)$ equals $f(x)$ whenever $f$ is $c$-Lipschitz on the entire power set $\cP(x)$ of $x$, that is, on the neighborhood $\DN_{|x|}(x)$. This construction implies a privacy wrapper: given $c$, $f$ and $x$, compute $f_c(x)$ and release it with noise scale $c/\eps$. The resulting wrapper has optimal error but very high locality $|x|$ and query complexity $2^{|x|}$.

Our \algSE\ privacy wrapper, \ifnum\stoc=0 whose performance is stated in \Cref{thm: subset extension}, \fi achieves the best of both worlds: the accuracy of Cummings-Durfee and the locality of TAHOE. We describe its construction, which departs significantly from the existing approaches, in \Cref{sec:techniques}.

Both the error scale and locality 
of our privacy wrappers
are essentially optimal. The optimality of error follows the work of Ghosh et al.~\cite{GhoshRS09}, mentioned above. To prove optimality of locality, we give two lower bounds: one which bounds locality directly \ifnum\stoc=0 (\Cref{thm: locality lower bound})\fi, and the other which bounds query complexity \ifnum\stoc=0 (\Cref{thm: query lower bound})\fi, and thus locality by implication. \ifnum\stoc=1 The proofs of our lower bounds appear in the \fullcite.\fi

\subsubsection{Privacy Wrappers with Claimed Sensitivity Bound for Bounded-Range Functions}

\ifnum\stoc=0 
In \Cref{sec: small diameter,sec:GenShI-with-nice-noise}, we 
\else 
In the \full, we 
\fi show that one can achieve improved guarantees in the claimed sensitivity bound setting for functions with small intended range. We present our  results for this setting in \Cref{tab:small-range}.

\begin{table*}[ht]
    \caption{\rm Results in the claimed sensitivity bound setting with a public and bounded range (for functions  $f:\dom\to [0,r]$ and for $(\eps,0)$-DP).}\label{tab:small-range}
    \centering
    \begin{tabular}{|c|c||c|c|c|}
    \hline
         \multirow{2}{*}{Algorithm}  & Reference&
         \multirow{2}{*}{Accuracy}  & Subexponential  & \multirow{2}{*}{Down Locality $\lambda$}  \\
         & Assumption && Error Scale  & \\
         \hline
         \hline
          Small Diameter   & \ifnum\stoc=0 \multirow{2}{*}{Thm. \ref{thm: small diameter}}\else \multirow{4}{*}{Full version of our work
\cite{LinderRSS25}}\fi &
          &
            \multirow{4}{*}{$O\Bparen{\dfrac{c}{\eps}}$} & 
          \multirow{2}{*}{$\dfrac{2r}{c}$} 
          \\
          Subset Extension  & &  $c$-Lipschitz  &  &\\  
          \ifnum\stoc=1\cline{1-1}\else\cline{1-2}\fi\cline{5-5}
          Double   Monotonization & \ifnum\stoc=0 Thm. \ref{thm:GenShi-with-nice-noise} \fi&  on $\DN_\lambda(x)$
           &
          & 
          $O\paren{\dfrac 1 \eps \log\dfrac{r}{c\beta}}$ 
          \\           
          \hline
          \hline
          \multirow{2}{*}{Lower Bounds} 
          &  \cite{GhoshRS09} and &\multirow{2}{*}{$c$-Lipschitz} 
          &  $\dfrac{c}{\eps}$ & 
        $\tilde \Omega\Bparen{\dfrac{1}{\eps}\log\dfrac{\eps r}{c\beta}}$
          \\
          &\ifnum\stoc=0 Thms. \ref{thm: locality lower bound} and \ref{thm: query lower bound}\else full version of our work
\cite{LinderRSS25}\fi & &
           \cite{GhoshRS09} &\ifnum\stoc=0
          Thms. \ref{thm: locality lower bound} and \ref{thm: query lower bound}\else\Full~ of our work
\cite{LinderRSS25}\fi \\   
          \hline
    \end{tabular}
\end{table*}

The Double Monotonization mechanism extends the guarantees of \algSE{} to the setting of bounded range, replacing the $\log(1/\delta)$ factor in the locality with $\log r$. At a technical level, it builds on \ASDalgshort{}, modifying it to take advantage of the (claimed) Lipschitz constant provided by the analyst. 

In contrast, the Small Diameter Subset Extension mechanism is based on a novel approach to \emph{local Lipschitz filters}. 
This approach follows the spirit of \cite{JhaR13}, who initiated an approach based on sublinear-time algorithms that was later studied in \cite{AwasthiJMR14,LangeLRV25}.
As a corollary of our techniques, we construct a local Lipschitz filter \ifnum\stoc=0 (\Cref{cor: lipschitz filter}) \fi that improves on \cite{LangeLRV25}. The resulting privacy wrapper achieves the best of both the algorithms of \cite{LangeLRV25}, and \cite{CummingsD20}.%

\subsubsection{Applications of  Privacy Wrappers to White-Box Setting}
\label{sec:applications_intro}

\newcommand{\mypar}[1]{\medskip\noindent \emph{#1}:}

\ifnum\stoc=0
In \Cref{sec:applications}, we 
\else 
Formal statements and proofs from this section appear in the \fullcite. We, 
\fi
give several applications of our privacy wrappers in the white-box privacy setting, both recovering known results, and obtaining immediate improvements to existing results. Our privacy wrappers offer a unified derivation of a range of results that, a priori, seem to require different techniques.

\mypar{Private mean estimation} As a simple illustration of our results, we show how they recover known bounds \cite{NissimRS07,DworkL09} on private mean estimation in one dimension. If the dataset $x\in\mathbb{R}^n$ is contained in an unknown interval of radius $\sigma$, then applying each of our wrappers to the average function leads to privately releasing $\mu$ with error roughly~$\frac \sigma {\epsilon n}$. \ifnum\stoc=0 See \Cref{cor:average} for details. \fi

\mypar{Empirical risk minimization} \ifnum\stoc=0 In \Cref{sec:ERM}, we 
\else
We 
\fi
improve upon the \emph{user-level} (also known as \emph{person-level}) private empirical risk minimization algorithm of \cite{GhaziKKMMZ23-user-sco} for the setting of one dimensional parameter spaces. In particular, \ifnum\stoc=0 in \Cref{thm:ERM-autosense,thm:ERM-SE}, \fi we show that in empirical risk minimization for a one-dimensional parameter space $\cY$, the dependence on privacy parameters $\eps$ and $\delta$ can be reduced from $O\big(\frac 1 {\eps^{5/2}}\log^2\frac 1\delta\big)$ to the minimum of \ifnum\stoc=1 the terms\fi~$O\big(\frac 1 {\eps^{3/2}}\log\frac 1\delta\big)$ and $\frac 1\eps\big(\log\frac 1 \delta\big)\exp(O\log^*|\cY|)$.
The first term in the minimum follows immediately by substituting the \algSE\ mechanism for the relevant subroutine in the private empirical risk minimization algorithm of \cite{GhaziKKMMZ23-user-sco}. The second term leverages the \ASDalgshort\ mechanism, as well as a straightforward extension of one of the key tools used in \cite{GhaziKKMMZ23-user-sco} to bound the Lipschitz constant on $\DN_\lambda(x)$ for $\lambda\geq |x|/2$.

\mypar{Estimating graph parameters} Our results give simpler derivations of the rate at which one can estimate the parameter $p$ of an Erdős–Rényi graph model $G(n,p)$ subject to node privacy guarantees. There is a node-private algorithm \cite{BorgsCSZ18} that, for all $p$, given a graph drawn from $G(n,p)$,  generates an estimate $\hat p$ with additive error 
$    \frac {1} n + O\bparen{\frac{\sqrt{\log n} }{ \eps n^{3/2}}}
    $  for $p$ bounded away from 0 and 1 (efficient algorithms achieving a similar bound were later given by \cite{SealfonU21,ChenDHS24erdosrenyi}). 
    This bound was surprising since the non-private estimator is a simple sum whose local sensitivity, which is $\Theta(\frac 1n)$ on typical graphs from $G(n,p)$, is too large to obtain the optimal rate by simply applying the Laplace mechanism. Such a strategy would lead to error $\frac 1{\eps n}$, instead of $\frac 1{\eps n^{3/2}}$.

We can rederive this bound using the observation that, for graphs generated from $G(n,p)$, the non-private estimator---which reports the density of edges in the input $G$---has down sensitivity   $\approx \frac{\lambda\sqrt{\log n}}{n^{3/2}}$ at depth $\lambda$, with high probability. Applying our results for automated sensitivity detection directly implies a similar feasibility result to that of \cite{BorgsCSZ18}. See \ifnum\stoc=0 \Cref{sec:ER-density}\else the \fullcite\fi~for details.

\subsection{Techniques}
\label{sec:techniques}

\paragraph{Monotonization and the Shifted Inverse Mechanism.}
Our wrappers with automatic sensitivity detection use the Shifted Inverse mechanism of \cite{FangDY22} as a starting point. That mechanism relies on the promise, for both privacy and accuracy, that its input function $f$ is monotone. We show 
how to transform an arbitrary function $f$ into a monotone function $g$ with two  additional locality properties: (a) the values of $g$ on $\DN_{\lambda}(x)$ depend only on the values of $f$ on a slightly larger down-neighborhood (say $\DN_{2\lambda}(x)$) and (b) the image $g(\DN_\lambda(x))$ is included in the image $f(\DN_{2\lambda}(x))$. Property (a) allows us to compute $g(x)$ locally (looking only at $\DN_{2\lambda}(x)$), and property (b) means that $DS_\lambda^g(x) \leq DS_{2\lambda}^f(x)$. 
We dub this transformation \emph{monotonization} (\Cref{def:monotonization}). It uses the input privately, by measuring its size $|x|$ (with Laplace noise) and setting a lower bound $\ell$ which is in the interval $[|x|-2\lambda, |x|-\lambda]$ with high probability. The monotization of $f$ is then 
$$\monfl(x)=\max\big(\{f(z): z\subseteq x, |z|\geq \ell\}\cup\{-\infty
\}\big),$$   
where ``$-\infty$'' is a lower bound on the analyst-specified range of $f$.

We combine this both with the original shifted inverse mechanism of \cite{FangDY22} as well as a new, $(\eps,\delta)$-private variant, described in \Cref{sec:shifted-inverse-monotone}, that achieves better dependency on the range size. This latter improvement comes from abstracting the original version as a reduction to the \emph{generalized interior point problem} \cite{BunDRS18,CohenLNSS23}. Monotonization and the resulting wrapper are described in  \Cref{sec:genShI}.

\paragraph{Subset Extension.} 
Our starting point for the claimed sensitivity bound setting is the TAHOE algorithm of \cite{KohliL23}. We first briefly describe (in a way that fits well with our adaptation).
As with monotonization, we first (privately) compute and release a lower bound $\ell$ on the size of $x$, which lies in the range $[|x|-3\lambda, |x|-2\lambda]$ (with high probability). This gives a ``floor" below which we do not need to read $f$ and bounds the locality by $3\lambda$. The next step is to attempt to find a subset of $\DN_{2\lambda}(x)$  on which $f$ is Lipschitz. Specifically, we say a subset $u$ of $x$ is \emph{$\ell$-stable} if $f$ is Lipschitz when restricted to the subsets of $u$ of size at least $\ell$. Kohli and Laskowksi show several structural properties of these sets. Let $\stabset{\ell}{h_0}{f}(x)$ denote the collection of  $\ell$-stable subsets of $x$ of size at least $h_0 \defeq \frac{\ell+|x|}{2}$. Then $f$ is Lipschitz on the domain $\stabset{\ell}{h_0}{f}(x)$. Furthermore, for every $h\geq h_0$, if $x'$ is a neighboring dataset of $x$ that is larger (by one insertion), then 
$$
\paren{
\begin{array}{c}
    \stabset{\ell}{h+1}{f}(x) \text{ not} \\
    \text{empty}
\end{array}
}
\implies 
\paren{
\begin{array}{c}
    \stabset{\ell}{h+1}{f}(x') \text{ not} \\
    \text{empty}
\end{array}
}
\implies
\paren{
\begin{array}{c}
    \stabsetlhf(x) \text{ not} \\
    \text{empty}
\end{array}
} \, .
$$
TAHOE can then be described as follows: first, select a publicly-released $\ell$ and unreleased $h$, both noisy (according to Laplace like distributions) and likely to satisfy $\frac{\ell+|x|}{2}\leq h\leq |x|$. Next, check if $\stabsetlhf(x)$ is empty, and release this bit. Finally, if $\stabsetlhf(x)$ is not empty, then pick an arbitrary largest set $u$ in $\stabsetlhf(x)$  and release $f(u) + Lap(\lambda c /\eps)$.  We can think of $f(u)$ as an approximation to $f(x)$ which is exact when $f$ is Lipschitz on all of $\DN_{3\lambda}(x)$ (since then $u=x$).
Privacy goes through because the bit indicating the emptiness of $\stabsetlhf(x)$ is randomized (by the randomness of $h$) and differentially private; and the diameter of $\stabset{\ell}{h_0}{f}(x')$ is $O(\lambda)$, so the sensitivity of $f(u)$ is $O(c\lambda)$ no matter how $u$ is chosen.  

We develop two different improvements over TAHOE, each of which modifies the structure above so that, roughly, the approximation to $f(x)$ has sensitivity $O(c)$ even when $f$ is not Lipschitz. 

Our first major departure from the TAHOE approach is to transform $f$ to get a new function $\hat f: x\mapsto  \frac 1 2 (\frac 1 c f(x) + |x|)$ that is monotone and Lipschitz whenever $f$ is Lipschitz. Unlike the monotonization described in the previous section, this transformation comes with no guarantees for arbitrary $f$. However, it allows us to choose a nearly-canonical representative in the set $\stabsetlhf(x)$, which is a point $u$ that maximizes $\hat f(u)$ over $\stabsetlhf(x)$. When $\hat f$ is monotone, maximizing the function value and maximizing the size of $u$ coincide.  

Our second major departure is to \emph{average} over the choices of $h$ rather than choosing $h$ randomly. In \algSE, we set $\ell$ using the truncated Laplace mechanism and average over the choice of $h$ in a range determined by $\ell$. For neighboring datasets $x$ and $x'$, we show a coupling between the function estimates computed using different $(\ell,h)$ pairs such that coupled $(\ell,h)$ values lead to similar estimates of $f(x)$ and $f(x')$ with high probability. Averaging over $h$ turns this high-probability Wasserstein distance bound into an $O(1)$ upper bound on the difference between estimates (when all checks for existence of stable sets pass), allowing us to release the estimate with little noise. 

This description hides a number of challenges that arise in the analysis. See \ifnum\stoc=0 \Cref{sec:subset-extension} \else the \fullcite\fi~for details.

\paragraph{Lower Bound Techniques.}
\ifnum\stoc=1
Proofs of our lower bounds appear in the \full.
\fi~
We provide two lower bounds: one on the locality of accurate privacy wrappers---via reduction from well-studied problem in the privacy literature---and another on the actual query complexity, via a new argument closer to the techniques in the property testing literature.

In order to prove the lower bound on locality, \ifnum\stoc=0 \Cref{thm: locality lower bound},\fi~we reduce from the ``point distribution problem" \ifnum\stoc=0 described in \Cref{sec:pd problem}\fi. In the point distribution problem, an algorithm is given a multiset $s$ of $n$ elements from some universe $\cY$ as input. For all $y\in\cY$, the algorithm must output $y$ whenever $s$ consists only of identical copies of $y$. Standard packing arguments suffice to show a lower bound on the size of the multiset for any differentially private algorithm that solves the point distribution problem. Our reduction then proceeds by arguing that a privacy wrapper that is $\lambda$-down local can be used to solve the point distribution problem with mutlisets of size $\lambda+1$.

Our second lower bound, \ifnum\stoc=0 \Cref{thm: query lower bound},\fi~states that every $(\eps,\delta)$-privacy wrapper that is $(\alpha,\beta)$-accurate on Lipschitz functions with range size $\rangesize$ must have query complexity $\abs{x}^{\Omega\paren{\frac{1}{\eps}\log\min\paren{\frac\rangesize\beta,\frac1\delta}}}$. Thus, \ifnum\stoc=0 \Cref{thm: query lower bound} immediately implies that\fi~the query complexity of our privacy wrappers is optimal. Additionally, while all of our privacy wrappers are down-local, \ifnum\stoc=0 \Cref{thm: query lower bound}\else the query complexity lower bound\fi~applies even to privacy wrappers that do no satisfy this guarantee (i.e., privacy wrappers that query $f$ on arbitrary datasets $z$).
To prove our query complexity lower bound, we take advantage of the fact that if two points $x,z\in\dom$ are ``close", but $f(x)$ and $f(z)$ are ``far", then every privacy wrapper must be inaccurate on $x$ or $z$. Leveraging this property, we construct distributions $\cN$ and $\cP$ over pairs $(x,f)$ where $x$ is a dataset and $f$ is a function from $[n]^*$ to $ \R$ with the following properties: First, every algorithm that gets query access to $f$ and input $x$ must make many queries to $f$ to distinguish whether $(x,f)\in\supp(\cN)$ or $(x,f)\in\supp(\cP)$. And second, every privacy wrapper that is accurate for Lipschitz functions is inaccurate when $(x,f)\in\supp(\cN)$ and accurate when $(x,f)\in\supp(\cP)$. Combining these two observation, we show that a query-efficient privacy wrapper that is accurate for all Lipschitz functions can be used as a subroutine to construct a low query algorithm for distinguishing $\cN$ and $\cP$. It follows that privacy wrapper that is accurate for all Lipschitz functions must make many queries to $f$.

\subsection{Related Work}
\label{sec:related}

\paragraph{Private Evaluation of Black-box Functions.} Privacy in the context of black-box functions was first explicitly considered by \cite{JhaR13}. They connect the claimed sensitivity setting with the concept of \textit{local filter} from the sublinear algorithms literature. This line of work \cite{JhaR13,AwasthiJMR14,LangeLRV25} constructs privacy wrappers from local filters for Lipschitz functions. Their constructions query the function $f$ on inputs obtained by both insertions and deletions. Their query complexity depends on the universe size, and the analyst's sensitivity bound must allow for the insertion of arbitrary outliers in the dataset.

\paragraph{Local Sensitivity and Robustness.} Soon after the introduction of differential privacy, the research community aimed to identify properties of a function that allow for accurate differentially private approximation. A key concept was the \textit{local sensitivity}~\cite{NissimRS07} of $f$ at $x$, the maximum change in $f$ that can be incurred by inserting or removing one element (or a small number of elements) in $x$. Starting with \cite{NissimRS07} and \cite{DworkL09}, a rich line of work found different ways to enforce and take advantage of (variants of) low local sensitivity~(e.g., \cite{AsiD20,AsiUZ23,HopkinsKMN23}, to mention only a few). 
Low local sensitivity within a neighborhood of insertions and deletions is equivalent to the notion of \emph{adversarial robustness} from the robust statistics literature. 
Many natural estimators do not have that type of robustness: for example, sample means and ordinary least-squares regression estimates  can be moved arbitrarily far by the insertion of a single outlier. They often require function-specific modifications (such as trimming, huerization, and so forth) to make them robust. Indeed, there is a large literature in the design of robust versions of popular estimators~\cite{HuberRbook,MaronnaMY06book}.

\paragraph{Look Down: Lipschitz Extensions and Down Sensitivity.} Another rich line of work in the privacy literature develops techniques for settings where we expect a computation of interest to be much more sensitive to insertions than removals of entries from the input dataset. One branch, initially motivated by node-private algorithms for graphs~\cite{BlockiBDS13,ChenZ13,KasiviswanathanNRS13,BorgsCS15,RaskhodnikovaS16,RaskhodnikovaS16survey,DayLL16,BorgsCSZ18,SealfonU21,KalemajRST23,JainSW24}, developed \emph{Lipschitz extensions}, which extend a function  from a subdomain of input datasets that 
has low sensitivity 
to the entire domain. The subdomains of interest were generally closed under removals but not insertions. The concept of down sensitivity (at depth 1) was introduced by Chen and Zhou~\cite{ChenZ13} to help describe these subdomains; they also studied the Lipschitz constant of $f$ on the power set of $x$. (For the supermodular functions studied in \cite{ChenZ13,KasiviswanathanNRS13}, these coincide.)
These works generally focused on efficient constructions in the white-box setting, mostly for monotone functions. However, some extensions, such as that of \cite{CummingsD20}, can be interpreted as a black-box construction (as in \Cref{table:results}). 

\paragraph{Down Sensitivity at Bounded Depth.} Two recent works on privacy study mechanisms that look only at the behavior of the function on large subsets of the input.  
Fang, Dong, and Yi~\cite{FangDY22} focus on white-box approaches to releasing monotone functions without an a-priori sensitivity bound, while Kohli and Laskowski~\cite{KohliL23} consider the black-box, claimed sensitivity setting. The KL mechanism was recently extended to %
general outputs in~\cite{GhaziKKMMZ23-user-sco,GhaziKKMMZ23-user-few-examples}, though their improvements do not affect our setting. 

\paragraph{Resiliency in Robust Statistics.} 
The two notions of accuracy that we consider---expressed in terms of the diameter of $f(\DN_\lambda(x))$ and the Lipschitz constant on $\DN_\lambda(x)$---correspond to two different notions of resiliency. The diameter-based notion corresponds (up to reparameterization) to the definition of \textit{resiliency} \cite{SteinhardtCV18}, whereas the Lipschitz notion has not been considered explicitly before in the statistics and learning literature, to our knowledge.

One can interpret the information-theoretic results of \cite{SteinhardtCV18} as a generic transformation that takes a function $f$ and produces a version that is robust in the neighborhood of an input $x$ whenever $x$ is resilient with respect to the original function $f$. 
However, we are not aware of a generic way to transform that into a differentially private mechanism using previous work, without still needing to explore the values of $f(z)$ for all sets $z$ that differ from $x$ in $\lambda$ insertions and deletions. 
We make the \cite{SteinhardtCV18} transformation explicit and discuss its consequences in \ifnum\stoc=0 \Cref{sec:resilience}\else the \fullcite\fi. 

\paragraph{Generic Statistical Results.} 
Dwork and Lei \cite{DworkL09} and Smith \cite{Smith11} give generic transformations that create differentially private versions of statistical estimators. These results are incomparable to ours, since they rely on specific properties of the estimators, such as asymptotic normality and low bias.

\section{Preliminaries}

In this section, we formally define privacy wrappers, our main object of study. Let $\cU$ be an arbitrary countable universe of elements and $\dom$ be the set of finite subsets of $\cU$. First, we recall the definition of differential privacy. 

\begin{definition}[Neighboring sets, differential privacy \cite{DworkMNS06}]\label{def:dp}
    Two sets $x,y\in\dom$ are \emph{neighbors} if either $x=y\cup \{i\}$ or $y=x\cup\{i\}$ for some $i\in\cU$. A randomized algorithm $\cM$ is \emph{$(\eps,\delta)$-differentially private} (DP) if for all neighboring $x,y\in\dom$ and all measurable subsets $E$ of the set of outputs of $\cM$, 
    $$\Pr[\cM(x)\in E]\leq e^\eps\Pr[\cM(y)\in E]+\delta.$$
    When $\delta=0,$ this guarantee is referred to as {\em pure differential privacy}; the guarantee with $\delta>0$ is called {\em approximate differential privacy}.
\end{definition}

\begin{definition}[Diameter]
    The \emph{diameter} of a bounded set $Y\subset \R$ is $\sup_Y(x)-\inf_Y(x)$. Moreover for all $f:\dom\to\R$ and $S\subseteq \dom$, we define $f(S)$ as the set $\{f(x): x\in S\}$.
\end{definition}

\begin{definition}[Lipschitz functions and global sensitivity]\label{def:lipschitz}
Fix $c\geq 0$. Given a domain $D\subseteq \dom$, a function $f:\dom\to\R$ is \emph{$c$-Lipschitz over $D$} if $|f(x)-f(y)|\leq c$ for all neighbors $x,y\in D$. For brevity, we use ``Lipschitz'' instead of ``1-Lipschitz''. When $D=\dom$, we just say ``$c$-Lipschitz'', without specifying the domain. The smallest constant $c$ for which function $f$ is $c$-Lipschitz is called the {\em Lipschitz constant} or the {\em (global) sensitivity} of $f$.
    \end{definition}

 \begin{definition}[Monotone functions] 
  Given a domain $D\subseteq \dom$, a function $f:\dom\to\R$ is {\em monotone} on $D$ if $f(x)\leq f(y)$ for all $x,y\in D$ such that $x\subset y$.
\end{definition}

Next, we define privacy wrappers. Informally, a privacy wrapper is a differentially private algorithm $\cW$ that gets two types of inputs: public and sensitive. The public inputs consist of a function $f$, to which the wrapper has query access, and a possibly empty list of parameters. The sensitive input is a data set $x\in \dom$. The algorithm run with query access to $f$ is denoted $\cW^f$, and its output on dataset $x$ is denoted $\cW^f(x)$. (We treat $\cW^f(x)$  as a random variable.) We omit explicit notation for the parameters.

\begin{definition}[Privacy wrapper  \cite{KohliL23}]
	\label{def: privacy wrapper}
Fix a universe $\cU$ and privacy parameters $\eps>0$ and $\delta\in[0,1)$. Consider a randomized algorithm $\cW$ that gets 
as input a dataset $x\in\dom$, additional parameters, and query access to a function $f:\dom\to\R$. It then produces output in $\R\cup\{\perp\}$. Algorithm $\cW$ is an $(\eps,\delta)$-privacy wrapper if, for every function $f$ and every choice of the additional parameters, the algorithm $\cW^f$ is $(\eps,\delta)$-differentially private. 
\end{definition}

\begin{definition}[Accuracy]
    We define two types of accuracy guarantees for a privacy wrapper $\cW$:  
    \begin{itemize}
        \item\textbf{Accuracy:} $\cW$ is $(\alpha,\beta)$-accurate for 
        a function $f$ and a dataset $x$ if 
        \ifnum\stoc=0
        $$
        \Pr\big[|\cW^f(x)-f(x)|\geq\alpha \big]\leq\beta.
        $$
        \else
        $
        \Pr\big[|\cW^f(x)-f(x)|\geq\alpha \big]\leq\beta.
        $
        \fi
        \item\textbf{Distribution:} $\cW$ \emph{has noise distribution $\cD$} for 
        a function $f$ and a dataset $x$ if 
        \ifnum\stoc=0
        $$
        \cW^f(x)\sim f(x) + \cD.
        $$
        \else
        $
        \cW^f(x)\sim f(x) + \cD.
        $
        \fi
    \end{itemize}
\end{definition}

Below, we present the Laplace distribution and corresponding Laplace mechanism. 

\begin{definition}[$\Lap$ and $\TLap$ distributions]
    \label{def: laplace}
    The \emph{Laplace} distribution, denoted $\Lap(b)$, is defined over $\R$ by the probability density function $f(x)=\frac1{2b}e^{-|x|/b}$. The \emph{truncated Laplace} distribution, denoted $\TLap(b,\tau)$, is given by the probability density function $f(x)=a_{b,\tau}\cdot\frac1{2b}e^{-|x|/b}$ when $|x|\leq\tau$ and $0$ otherwise, where  $a_{b,\tau}$ is a normalizing constant.
\end{definition}
\begin{fact}[Laplace mechanism \cite{DworkMNS06}]
    \label{fact: laplace mechanism}
    Fix $\eps>0$ and let $f:\dom\to \R$ be a $c$-Lipschitz function. Then, the algorithm that gets a query $x\in \dom$ as input, samples $J\sim \Lap(\frac c\eps)$, and outputs $g(x)=f(x)+J$, is $(\eps,0)$-DP. Additionally, for all $\alpha>0$, we have $|g(x)-f(x)|\leq \frac c\eps\ln\frac1\alpha$ with probability at least $1-\alpha$. Moreover, the same mechanism implemented with noise $J\sim\TLap(\frac c\eps,\frac c\eps\ln\frac1\delta)$ is $(\eps,\delta)$-DP.
\end{fact}

We repeatedly use the following tail bound for Laplace random variables. 

\begin{fact}[Laplace tails]\label{fact:laplace} For all $s>0$ and $\beta\in (0,1)$, if $Z\sim \Lap(s)$ then $\Pr\Bparen{|Z|\geq s \ln\frac 1\beta} = \beta$.
\end{fact}

We use the following well known facts regarding the composition of differentially private mechanisms and postprocessing. 
\ifnum\stoc=1
(See \cite{BookAlgDP}).
\else
These can be found in 
\cite{BookAlgDP}.
\fi

\begin{fact}[Composition]
\label{fact:composition}
Fix $\eps_1,\eps_2>0$ and $\delta_1,\delta_2\in(0,1)$. Suppose $\cM_1$ and $\cM_2$ are (respectively) $(\eps_1,\delta_1)$-DP and $(\eps_2,\delta_2)$-DP. Then, the mechanism that, on input $x$, outputs $(\cM_1(x),\cM_2(x))$ is $(\eps_1+\eps_2,\delta_1+\delta_2)$-DP. 
\end{fact}

\begin{fact}[Postprocessing]
\label{fact:postprc}
Fix $\eps>0$ and $\delta\in(0,1)$. Let~$\cA$ be an algorithm and $\cM$ be an $(\eps,\delta)$-DP mechanism. Then the algorithm that, on input $x$, runs $\cA$ on the output of $\cM(x)$, is $(\eps,\delta)$-DP.  
\end{fact}

\section{Privacy Wrappers with Automated Sensitivity Detection}\label{sec:generalized-shifted-inverse}

In this section, we state and prove \Cref{thm:generalized-shifted-inverse-mechanism}, which provides privacy wrappers for situations when the analyst gives no information about the sensitivity of the black-box function $f$ they provide.

\begin{theorem}[\ASDalgfull\ privacy wrapper]
    \label{thm:generalized-shifted-inverse-mechanism}
    For every universe $\cU$, privacy parameters $\eps>0$ and $\delta\in[0,1)$, error probability $\beta \in [0,1)$,  and finite range $\cY\subset\R$ with size $\rangesize \defeq |\cY|$, 
    there exist a $\lambda$-down-local privacy wrapper $\cW$ such that for every  function $f:\dom\to\cY$ and dataset $x\in\dom$, %
    with probability at least $1-\beta$,
    \[ 
    \cW^f(x) \in [\min f \bparen{\DN_\lambda(x)}, \max f \bparen{\DN_\lambda(x)}] \, ,
    \]
    where $\beta,\delta,\lambda$ satisfy the following: 
    \begin{enumerate}
        \item If $\beta>0$, then $\delta=0$ and $\lambda = O\bparen{\frac1\eps\log\frac{k}\beta}$;
        \item If $\delta>0$, then $\beta=0$ and $\lambda = \frac 1 \eps \cdot 2^{O(\log^* \rangesize )}\log \frac 1 \delta$. (In this case, the privacy wrapper is correct with probability 1.) 
    \end{enumerate}
\end{theorem}

In particular, the accuracy guarantee of this wrapper implies that $|\cW^f(x)-f(x)|\leq DS^f_\lambda(x)$. 
In order to state results in this section more compactly, we define a single function $\lambda(\eps,\delta,\beta,\rangesize)$ that captures both cases of \Cref{thm:generalized-shifted-inverse-mechanism}, with the convention that it is invoked with exactly one of $\delta$ and $\beta$ being nonzero:  
\begin{align}\label{eq:lambda-function}
 \lambda(\eps,\delta,\beta,\rangesize) =
\begin{cases}
O\left(\frac1\eps\log\frac{\rangesize}\beta\right)& \text{for } \beta>0 \text{ and }\delta= 0,\\
    \frac 1 \eps \cdot 2^{O(\log^* \rangesize )}\log \frac 1 \delta & \text{for } \delta>0 \text{ and }\beta= 0.
\end{cases} 
\end{align}

In \Cref{sec:shifted-inverse-monotone}, we describe a privacy mechanism that works under a promise that the function $f$ is monotone. In \Cref{sec:genShI}, we transform it to a privacy wrapper that satisfies the conditions of \Cref{thm:generalized-shifted-inverse-mechanism}.

\subsection{Shifted Inverse Mechanism: Promised Monotone Functions}\label{sec:shifted-inverse-monotone}

In this section, we describe a privacy mechanism that works under a promise that the black-box function $f$ is monotone. It can be viewed as a privacy wrapper under the promise, but we do not call it a privacy wrapper to stress that, in contrast to our privacy wrappers, this mechanism is {\em not} private when the promise is broken. Our mechanism is a novel variant~\cite{DPorg-down-sensitivity}
of the {\em Shifted Inverse (\ShI)} mechanism of Fang, Dong, and Yi \cite{FangDY22}. The original \ShI mechanism satisfies {\em pure} differential privacy for all monotone functions $f$. Our variant gets better dependence on the size of the range of function $f$ at the expense of providing only approximate differential privacy.

\begin{lemma}[Shifted Inverse Mechanism with approximate or pure DP]
\label{thm:genSIM}  
\label{lem:shifted-inverse-both}
    There exists a function $\lambda_{\text{ShI}}(\eps,\delta,\beta,k)$ satisfying \eqref{eq:lambda-function} such that, for every universe $\cU$, privacy parameters $\eps>0$ and $\delta\in[0,1)$, failure probability $\beta\in(0,1)$,  range size $\rangesize$, and range $\cY\subset\R$ of size $\rangesize$, there exists a mechanism $\cM$  such that 
for every \textbf{monotone} function $f:\dom\to\cY$:
     \begin{itemize}
         \item the mechanism $\cM^f$ is $(\eps,\delta)$-DP and
     \item
     for $\lambda=\lambda_{\text{ShI}}(\eps,\delta,\beta, k)$
     and every dataset $x\in\dom$, the mechanism $\cM^f$ is $\lambda$-down local and,  with probability at least $1-\beta$, satisfies 
     \begin{equation}
         f(x) - DS_\lambda^f(x) \leq \cM^f(x) \leq f(x) \, .
     \end{equation}
      \end{itemize}

\end{lemma}

The shifted inverse mechanism reduces the task of constructing a privacy wrapper for a monotone function to a generalized interior point problem. We state the definition of this problem from \cite{BunDRS18}. A different---but equivalent---formulation appears in \cite{CohenLNSS23}.

\begin{definition}[Generalized Interior Point Problem {\cite{BunDRS18}}]\label{def:GIPP}
    A function $g : \dom \times [k] \to [0,1]$ gives a generalized interior point problem with sensitivity $\Delta$ if it satisfies the following.
    \begin{itemize}
        \item The function $g$ has sensitivity $\Delta$ in its first argument; i.e., $|g(x,j)-g(y,j)| \le \Delta$ for all neighboring $x,y \in \dom$ and all $j \in [k]$.
        \item The function $g$ is %
        nondecreasing in its second argument; i.e., $0 \le g(x,j)\le g(x,j+1) \le 1$ for all $j \in [k-1]$ and all $x \in \dom$.
    \end{itemize}
    For notational convenience, define $g(x,0)=0$ and $g(x,k+1)=1$ for all $x \in \dom$.

    A solution to the generalized interior point problem given by $g$ on an input $x \in \dom$ is an index $j \in [k+1]$ such that $g(x,j)>0$ and $g(x,j-1)<1$.
\end{definition}

To understand where this problem comes from, consider the (non-generalized) interior point problem \cite{BunNSV15}: We are given $x \in \dom$ where $\univ=[k]$ and seek to output $j \in [k]$ such that $\min x \le j \le \max x$. Being between the minimum and maximum means being in the ``interior'' of the dataset; hence the name. This is a relaxation of the problem of finding a median. We can convert this into a generalized interior point problem by setting $g(x,j) = \frac{|x \cap [0, j]|}{\max\{|x|,1/\Delta\}}$. Then $g(x,j)>0 \iff j \ge \min x$ and $g(x,j-1)<1 \iff j-1 < \max x \iff j \le \max x$, assuming $|x| \ge 1/\Delta$.

The complexity of the generalized interior point problem is measured by the sensitivity $\Delta$, which roughly corresponds to the reciprocal of the sample complexity $n=1/\Delta$.

Under pure DP, we can solve the generalized interior point problem using the exponential mechanism. I.e., 
\ifnum\stoc=0
$\Pr[\cM(x)=j] \propto \exp\left(\frac{\varepsilon}{2\Delta}\min\{g(x,j),1-g(x,j-1)\}\right)$.
\else
$$\Pr[\cM(x)=j] \propto \exp\left(\frac{\varepsilon}{2\Delta}\min\{g(x,j),1-g(x,j-1)\}\right).$$
\fi 
The sample complexity is $n=1/\Delta=O(\log(k)/\varepsilon)$. 
Under concentrated DP \cite{DworkR16,BunS16} or Gaussian DP \cite{DongRS19}, we can solve the generalized interior point problem using noisy binary search \cite{KarpK07} over the index $j \in [k+1]$ where we add Gaussian noise to each value $g(x,j)$. The sample complexity is $O(\sqrt{\log k})$.
Under approximate DP, we are able to obtain dramatically better sample complexity.

\begin{proposition}
[\cite{BunDRS18,CohenLNSS23}]\label{prop:gipp}
    For all $\varepsilon,\delta\in(0,1)$ and $k \in \mathbb{N}$, there exists a parameter $\lambda=O\left(\frac{\log(1/\delta)}{\varepsilon}  \cdot 2^{O(\log^* k)}\right)$ such that the following holds.
    Let $g : \dom \times [k] \to [0,1]$ be a generalized interior point problem with sensitivity $\Delta \le 1/\lambda$.
    Then there exists an $(\varepsilon,\delta)$-differentially private algorithm $\cM : \dom \to [k+1]$ 
    which, on each input $x\in\dom$,
    outputs a solution to the generalized interior point problem given by $g$ on input $x$.
\end{proposition}
Proposition \ref{prop:gipp} guarantees that the output is a solution to the generalized interior point problem with probability 1.
Cohen et al.~\cite{CohenLNSS23} only guarantee a success probability of $9/10$. Their algorithm can be modified to achieve success probability 1 or, alternatively, we can use a generic reduction \cite{DPorg-fail-prob} that amplifies the success probability to 1 (at the expense of a constant factor increase in the privacy parameters and an additive $O(\log(1/\delta)/\varepsilon)$ in the sample complexity).

Now we %
present our generalization of the 
\ShI mechanism from \cite{FangDY22}.
We begin by defining the inverse loss function: Given a function $f : \dom \to \mathbb{R}$, we define $\ell^f : \dom \times \mathbb{R} \to [0,\infty]$ by\footnote{There is an annoying technicality: If $y < f(\emptyset)$, then $\ell^f(x,y)=\min\emptyset=+\infty$.}
\begin{equation}
    \ell^f(x,y) := \min\{ |x \setminus s| : s \subseteq x, f(s) \le y \} . \label{eq:invloss}
\end{equation}
In words, $\ell^f(x,y)$ is the number of points that need to be removed from the input $x$ until the  value of the function $f$ becomes less than or equal to $y$.
We can invert the inverse loss function to recover the original function:
For all $f,x,y$, we have
\ifnum\stoc=0
\begin{equation}
    \ell^f(x,y) = 0 \iff f(x) \le y
    ~~\text{ or, equivalently, }~~
    \ell^f(x,y) > 0 \iff f(x) > y \label{eq:invloss0}
\end{equation}
\fi
\ifnum\stoc=1
\begin{align}
        &\ell^f(x,y) = 0 \iff f(x) \le y \nonumber\\
    \text{or, equivalently, \ \ \ \ }&
    \ell^f(x,y) > 0 \iff f(x) > y.
    \label{eq:invloss0}
\end{align}
\fi
Hence,
\begin{equation}
    f(x) = \min \{ y \in \mathbb{R} : \ell^f(x,y) = 0 \}
    =  \sup \{ y \in \mathbb{R} : \ell^f(x,y) > 0 \} .
\end{equation}
We can also relate the inverse loss to down sensitivity:
\begin{equation}
    \ell^f(x,y) \le \lambda \implies y \ge f(x) - DS^f_\lambda(x). \label{eq:invlossDS}
\end{equation}
Combining \eqref{eq:invloss0} and \eqref{eq:invlossDS} tells us that, if we can find $y \in \mathbb{R}$ such that $0<\ell^f(x,y)\le\lambda$, then $f(x) - DS^f_\lambda(x) \le y < f(x)$.
Such a set $y$ is precisely what the shifted inverse mechanism tries to find.

The advantage of the inverse loss function is that it has low sensitivity, even when $f$ has high sensitivity. However, this only holds when $f$ is monotone. The following lemma encapsulates an elegant insight of Fang et al.~\cite{FangDY22}. 

\begin{lemma}[Sensitivity of inverse loss function for monotone functions]\label{lem:invsens}
    Let $f : \dom \to \mathbb{R}$ be monotone. 
    Define $\ell^f : \dom \times \mathbb{R} \to [0,\infty]$ as in \eqref{eq:invloss}.
    Then $\ell^f$ has sensitivity $1$ in its first argument. I.e., for all $x,x'\in\dom$ and all $y \in \mathbb{R}$, $|\ell^f(x,y)-\ell^f(x',y)| \le |x \setminus x'|+|x' \setminus x|$.
\end{lemma}

\ifnum\stoc=1 
    The proof of \Cref{lem:invsens} (due to \cite{FangDY22}) appears in the \fullcite.
\else
We present a proof for completeness. 
\begin{proof}
    Fix $y \in \mathbb{R}$ and $x,x'\in\dom$. 
    We break the proof into two claims:
    \begin{itemize}
    \item\textbf{Claim I:} If $x'' \subset x'$ and $f$ is monotone, then $\ell^f(x'',y) \le \ell^f(x',y)$.
    \item\textbf{Claim II:} If $x'' \subset x$, then $\ell^f(x,y) \le \ell^f(x'',y) + |x \setminus x''|$.\footnote{Claim II does not require monotonicity, but Claim I does.}
    \end{itemize} 
    Assuming the claims, the lemma can be proved by setting $x'' = x \cap x'$. Then 
    \ifnum\stoc=0
    \[
    \ell^f(x,y) \overset{\text{Claim II}}{\le} \ell^f(x'',y) + |x \setminus x''| \overset{\text{Claim I}}{\le} \ell^f(x',y) + |x \setminus x''| = \ell^f(x',y) + |x \setminus x'|,
    \] 
    \fi
    \ifnum\stoc=1
    \[
    \begin{split}
    \ell^f(x,y) \overset{\text{Claim II}}{\le} \ell^f(x'',y) + |x \setminus x''| &\overset{\text{Claim I}}{\le} \ell^f(x',y) + |x \setminus x''| \\
    & = \ell^f(x',y) + |x \setminus x'|,
    \end{split}
    \]
    \fi
    which establishes $\ell^f(x,y)-\ell^f(x',y) \le |x \setminus x'|+|x' \setminus x|$. The other direction follows by symmetry.
    \begin{proof}[Proof of Claim I]
        Let $x_* = x' \setminus x'' \subset x'$ . Since $x'' \subseteq x'$, we have $x'' = x' \setminus x_*$.
        
        Let $s' \subset x'$ be such that $\ell^f(x',y)=|x'\setminus s'|$ and $f(s') \le y$.
        Let $s'_* = s' \setminus x_*$.
        Since $s'_* \subseteq s'$, we have $f(s'_*) \le f(s') \le y$ by monotonicity.
        Also, since $s' \subseteq x'$, we have $s'_* \subset x' \setminus x_* = x''$.
        Thus
        \begin{align*}
            \ell^f(x'',y) &= \min\{ |x'' \setminus s''| : s'' \subseteq x'', f(s'') \le y \} \\
            &\le |x'' \setminus s'_* | \\
            &= |(x' \setminus x_*) \setminus ( s' \setminus x_*) | \\
            &\le |x' \setminus s'| \\
            &= \ell^f(x',y).\qedhere
        \end{align*}
    \end{proof}
    \begin{proof}[Proof of Claim II]
        Let $s'' \subset x''$ be such that $\ell^f(x'',y)=|x'' \setminus s''|$ and $f(s'') \le y$.
        Since $x'' \subseteq x$, we have $s'' \subseteq x$ and, hence,
        \begin{align*}
            \ell^f(x,y) &= \min\{ |x \setminus s| : s \subseteq x, f(s) \le y \} \\
            &\le |x \setminus s''| \\
            &= |x'' \setminus s''| + |x \setminus x''| \\
            &= \ell^f(x'',y) + |x \setminus x''|.\qedhere
        \end{align*}
    \end{proof}
This completes the proof of \Cref{lem:invsens}.
\end{proof}

\fi
Now we can prove our result on the shifted inverse mechanism.
Note that the differential privacy guarantee of $\cM^f$ depends on the monotonicity of $f$. 

\begin{proof}[Proof of \Cref{thm:genSIM}]
    As in \eqref{eq:invloss}, define $\ell^f : \dom \times \mathbb{R} \to [0,\infty]$ by
    \[\ell^f(x,y) := \min\{ |x \setminus s| : s \subseteq x, f(s) \le y \}.\]
    By \Cref{lem:invsens}, $\ell^f$ has sensitivity $1$ in its first argument. 
    Also $\ell^f$ is non-increasing in its second argument. I.e., $y_1 \le y_2 \implies \ell^f(x,y_1) \ge \ell^f(x,y_2)$.
    Let $\mathcal{Y}=\{y_1 \le y_2 \le \cdots \le y_k\}$ be an ordered enumeration of $\mathcal{Y}$.
    Define $g : \dom \times [k] \to [0,1]$ by \[g(x,j) := \max\Big\{ 0 , 1 - \frac{1}{\lambda+1} \ell^f(x,y_j)\Big\} .\]
    Then $g$ gives a generalized interior point problem with sensitivity $\Delta=\tfrac{1}{\lambda+1}$.

    We claim that, if $j$ is a solution to the generalized interior point problem given by $g$, then \[f(x) - DS_\lambda^f(x) \le y_j \le f(x).\]
    To prove the claim, suppose $j \in [k+1]$ is a solution to the generalized interior point problem given by $g$.
    Then $g(x,j)>0$, which implies $\ell^f(x,y_j)<\lambda+1$ (i.e., $\ell^f(x,y_j)\le\lambda$) and, hence, $y \ge f(s) \ge f(x) - DS^f_{\lambda}(x)$ for some $s \subseteq x$ with $|x \setminus s| \le \lambda$.
    Also $g(x,j-1)<1$, which implies $\ell^f(x,y_{j-1})>0$ and, hence, $f(x)>y_{j-1}$. Since $f(x)=y_{j'}$ for some $j' \in [k]$, we have $f(x) \ge y_j$. 
    (If $j=1$, then $y_{j-1}$ is undefined, but the conclusion $f(x) \ge y_j $ still holds trivially because $y_j = y_1 = \min \mathcal{Y}$.) 
    (Note that $j=k+1$ is not a valid solution since this requires $g(x,k)<1$, which implies $\ell^f(x,y_k)>0$, which implies $f(x)>y_k = \max \mathcal{Y}$---a contradiction.)

    Given this claim, it now suffices to solve the generalized interior point problem given by $g$.
    For the case of approximate differential privacy ($\delta>0$), we can apply the algorithm given by \Cref{prop:gipp}. This yields the second term in the minimum.
    For the case of pure differential privacy ($\delta=0$), we can apply the exponential mechanism. (This case  was already analyzed by Fang, Dong, and Yi \cite{FangDY22}. We include this analysis for completeness.) That is, our algorithm is defined by 
    \[\Pr[\cM^f(x)=y_j] = \frac{\exp\left(\frac{\varepsilon}{2\Delta} \min \{ g(x,j) , 1 - g(x,j-1) \}\right)}{\sum_{\ell \in [k]} \exp\left(\frac{\varepsilon}{2\Delta} \min \{ g(x,\ell) , 1 - g(x,\ell-1) \}\right)},\] where we define $g(x,0)=0$. 
    Since $g$ has sensitivity $\Delta = \tfrac{1}{\lambda+1}$ in its first argument, $\cM^f$ is $(\varepsilon,0)$-differentially private.
    In terms of utility, with probability $\ge 1-\beta$ over $j \gets M(x)$, we have 
    \ifnum\stoc=0
    \[ 
    \min \{ g(x,j) , 1 - g(x,j-1) \} \ge \max_{\ell \in [k]} \min \{ g(x,\ell) , 1 - g(x,\ell-1)\} - \frac{2\Delta}{\varepsilon} \log (k/\beta) .
    \]
    \fi
    \ifnum\stoc=1
    \[ 
    \begin{split}
    \min \{ g(x,j) , 1 - g(x,j-1) \} 
    &\ge \max_{\ell \in [k]} \min \{ g(x,\ell) , 1 - g(x,\ell-1)\} \\
    &- \frac{2\Delta}{\varepsilon} \log (k/\beta) .
    \end{split}
    \]
    \fi
    Thus it suffices to show that \[\max_{\ell \in [k]} \min \{ g(x,\ell) , 1 - g(x,\ell-1)\} - \frac{2\Delta}{\varepsilon} \log (k/\beta) > 0.\]
    Since $g(x,0)=0$ and $g(x,k)=1$, there exists some $\ell \in [k]$ such that $g(x,\ell)>\frac12$ and $g(x,\ell-1)\le\frac12$, which implies $\min \{ g(x,\ell) , 1 - g(x,\ell-1)\} \ge \frac12$.
    Thus it suffices to have $\frac{2\Delta}{\varepsilon} \log (k/\beta) < \frac12$, which is equivalent to \[\lambda = \frac{1}{\Delta}-1 > \frac{4}{\varepsilon} \log(k/\beta) - 1.  \]

    Finally, we consider what access to the function $f$ is required to execute $\mathcal{M}$. 
    We must compute $g(x,j)$ for each $j \in [k]$, which depends on $\ell^f(x,y_j)$.
    Note that $\ell^f(x,y)$ only depends on the values $f(s)$ for $s \subseteq x$.
    
    Observe that $g(x,j)=1 \iff \ell^f(x,y_j) \ge \lambda+1$. That is, the exact value of $\ell^f(x,y)$ does not matter past the threshold $\lambda+1$. So we do not need to compute the exact value of $\ell^f(x,y)$ in this case.
    Therefore, we compute $g(x,j)$ using only the values of $f(s)$ on $s \in \DN_\lambda(x)$.
    In symbols, $g(x,j)$ is the maximum of 0 and the following expression:
    \begin{align*}
        1 &- \frac{1}{\lambda+1} \ell^f(x,y_j)
        = 1 - \frac{1}{\lambda+1} \min\{ |x \setminus s| : s \subseteq x, f(s) \le y \} \\       
        &=  1 - \frac{1}{\lambda+1} \min\left( \{\lambda+1\} \cup \{ |x \setminus s| : s \subseteq x, f(s) \le y \}\right) \\       
        &=1 - \frac{1}{\lambda+1} \min\left( \{\lambda+1\} \cup \{ |x \setminus s|  : s \in \DN_\lambda(x), f(s) \le y \}\right) .\qedhere
    \end{align*}
\end{proof}

\subsection{\ASDalgfull: 
 A Wrapper for General Functions
}
\label{sec:genShI}

The two variants of the Shifted Inverse mechanism discussed in the previous section are not privacy wrappers because they are only private under the promise that the function $f$ is monotone. In this section, we generalize \ShI to work for all functions.

\begin{proof}[Proof of \Cref{thm:generalized-shifted-inverse-mechanism}]
The main idea in the generalization of the \ShI mechanism is to construct a monotone function $g$ from the original function $f$ and use $g$ in the \ShI mechanism. The value of $g$ at point $x$ 
will be computed from the down neighborhood of $x$. We parameterize $g$ by the lowest level (i.e., the set size) we include in the down neighborhood. 
\begin{definition}[Monotonization of $f$]\label{def:monotonization} Fix a universe $\univ$ and a range $\cY\subseteq \R$.
 For each $\ell\in\Z$, the {\em level-$\ell$ monotonization} of a function $f:\dom\to\cY$ is the function $\monfl:\dom\to\cY$ defined by\footnote{ We use the convention that if $\cY$ is unbounded below, then $\inf(\cY)=-\infty$.}
$$\monfl(x)=\max\big(\{f(z): z\subseteq x, |z|\geq \ell\}\cup\{\inf(\cY)
\}\big).$$   
\end{definition}

   The following properties of monotonization follow directly from its definition.
\begin{observation}[Properties of monotonization]\label{observation:monotonization}
 For a level $\ell\in\Z$ and a function $f:\dom\to\R$, let $\monfl$ be the level-$\ell$ monotonization of $f$. Then the following properties hold:
 \begin{enumerate}
     \item\label{item1:mono} The function $\monfl$ is monotone.
     \item\label{item2:mono} If $f$ is monotone then $\monfl=f$. 
     \item\label{item3:mono} The value $\monfl(x)$ can be computed by querying $f$ on all subsets of $x$ of size at least $\ell.$        
 \end{enumerate}
\end{observation}

Mechanism \ASDalgshort{} is stated in \Cref{alg:genshi}. It first uses Laplace mechanism to choose appropriate level $\ell$ and then runs \ShI with query access to the monotonization of function $f$ at level $\ell.$ It uses our version of \ShI from \Cref{lem:shifted-inverse-both}.
\begin{algorithm}[htb]
\caption{\ASDalgshort}\label{alg:genshi} 
	\begin{algorithmic}[1]
            \Statex \textbf{Parameters:} privacy parameters $\eps>0$ and $\delta\in[0,1),$ failure probability $\beta\in(0,1)$, and finite range $\cY\subset\R$
	        \Statex \textbf{Input:} dataset $x\in\dom$  and query access to $f:\dom\to\cY$
	        \Statex \textbf{Output:} $y\in \R$ 
            \State 
            Set $\lambda\gets 2 \cdot \lambda_{\text{ShI}}(\eps/2,\delta,\beta/2, |\cY|)$, where $\lambda_{\text{ShI}}$ is given by \Cref{lem:shifted-inverse-both}.
            \ifnum\stoc=1 \quad \else \newline \fi
            \Comment{$\lambda$ is set so that \ShI run with the parameter \ifnum\stoc=0 settings below uses depth parameter $\lambda/2$.\fi}
            \ifnum\stoc=1
            \\ \Comment{\ \ settings below uses depth parameter $\lambda/2$.}
            \fi
            \State
             {\bf Release} $\ell\gets\lfloor|x|-\frac 34 \lambda+Z\rfloor$ where $Z\sim\Lap(\frac 2\eps)$.\label{step:release-ell} 
            \State\label{step:run-ShI}Run \ShI from \Cref{lem:shifted-inverse-both} with privacy parameters $\frac \eps 2$ and $\delta$, failure probability $\frac \beta 2$, range $\cY$, input dataset $x$, and query access to the level-$\ell$ monotonization $\monfl$ of $f$ (see \Cref{def:monotonization}) and {\bf return} the answer.
	\end{algorithmic}
 \end{algorithm}

Next, we analyze privacy of \ASDalgshort. \Cref{step:release-ell} 
uses the Laplace mechanism. Since $|x|$ (and, consequently, $\ell$) is a Lipschitz function of $x$, \Cref{fact: laplace mechanism} guarantees that this step is $(\eps/2,0)$-DP. Since $\monfl$ is monotone, the \ShI mechanism run in \Cref{step:release-ell} is $(\eps/2,\delta)$-DP. By composition (\Cref{fact:composition}), \ASDalgshort{} is $(\eps,\delta)$-DP for all functions $f$. 

The two failure events we consider are (1) the noise variable $Z$ in \Cref{step:release-ell} has large absolute value, $|Z|>\frac 2 \eps \ln \frac 2 \beta$, and (2) \ShI fails. Each of these events happens with probability at most $\beta/2$, by \Cref{fact:laplace} and the 
setting of parameters given to \ShI. 
By the union bound over these two events, the overall failure probability is at most $\beta.$

Now, we analyze accuracy of \ASDalgshort{}. Suppose that neither failure event occurred. Then $|Z|\leq \frac 2 \eps \ln \frac 2\beta \leq \frac \lambda 4$. (We assume w.l.o.g.\ that $c$ is sufficiently large for this inequality to hold.) Consequently, 
\begin{align}\label{eq:bounds-on-ell}
|x|-\lambda\leq\ell\leq|x|-\frac \lambda 2.    
\end{align}

Let $\cW$ denote \Cref{alg:genshi} and $\cM$ denote the \ShI mechanism. Recall that $\lambda$ is set so that \ShI in \Cref{step:run-ShI}  is run with the depth parameter $\lambda/2$. By the accuracy guarantee of \ShI, we get
\begin{align*}
\monfl(x) - DS_{\lambda/2}^{\monfl}(x) \leq \cM^{\monfl}(x) \leq \monfl(x) \, .    
\end{align*}
By construction of the privacy wrapper, $\cW^f(x)=\cM^{\monfl}(x)$. Since $\ell\leq |x|-\lambda/2$, the set of subsets of $x$ of size at least $\ell$ is nonempty. By the definition of monotonization $\monfl$, the fact that $\ell\geq |x|-\lambda$, and the definition of the down sensitivity, we get
\ifnum\stoc=0
\begin{align*}
\cW^f(x) = \cM^{\monfl}(x) 
\leq \monfl(x) 
= \max\big(\{f(z): z\subseteq x, |z|\geq \ell\}\big)
\leq \max_{z\in\DN_\lambda(x)} f(z)
\leq f(x) + DS_\lambda^f(x).
\end{align*}
\fi
\ifnum\stoc=1
\begin{align*}
\cW^f(x) 
&= \cM^{\monfl}(x) 
\leq \monfl(x) 
= \max\big(\{f(z): z\subseteq x, |z|\geq \ell\}\big)\\
&\leq \max_{z\in\DN_\lambda(x)} f(z)
\leq f(x) + DS_\lambda^f(x).
\end{align*}
\fi
Using monotonicity of $\monfl$ and the definition of the down sensitivity, we get
\ifnum\stoc=0
\begin{align*}
\cW^f(x) 
&=\cM^{\monfl}(x)
\geq \monfl(x) - DS_{\lambda/2}^{\monfl}(x)\\
&=\monfl(x) - \max_{x'\in\DN_{\lambda/2}(x)} (\monfl(x) -\monfl(x'))
=\min_{x'\in\DN_{\lambda/2}(x)} \monfl(x').
\end{align*}
\fi
\ifnum\stoc=1
\begin{align*}
\cW^f(x) 
&=\cM^{\monfl}(x)
\geq \monfl(x) - DS_{\lambda/2}^{\monfl}(x)\\
&=\monfl(x) - \max_{x'\in\DN_{\lambda/2}(x)} (\monfl(x) -\monfl(x'))\\
&=\min_{x'\in\DN_{\lambda/2}(x)} \monfl(x').
\end{align*}
\fi
By \eqref{eq:bounds-on-ell} and definition of monotonization, for all $x'\in\DN_{\lambda/2}(x)$, there exists $x''\in\DN_{\lambda}(x)$ such that $\monfl(x')=f(x'')$. Thus, 
\begin{align*}
    \cW^f(x)
    &\geq\min_{x'\in\DN_{\lambda/2}(x)} \monfl(x') 
    \\ 
    &\geq \min_{x''\in\DN_{\lambda}(x)} f(x'')
    \qquad = \quad   f(x)  - DS_\lambda^f(x).    
\end{align*}
Thus, with probability at least $1-\beta$, we get $|\cW^f(x)-f(x)|\leq DS^f_\lambda(x)$, as claimed.

Finally, to evaluate $\monfl$, the wrapper only needs to query $f$ on the subsets of the dataset $x$ of size at least~$\ell$. Recall that $\ell$ satisfies \eqref{eq:bounds-on-ell} with probability at least $1-\beta.$ When it does, all queries of $\cW$ are within $\DN_\lambda(x),$ completing the proof of \Cref{thm:generalized-shifted-inverse-mechanism}.
\end{proof}

\section{Privacy Wrappers with Claimed Sensitivity 
\ifnum\stoc =0 Bound \fi
}
\label{sec:extension_based}

\ifnum\stoc=1
Now we state our main result for the setting when the analyst provides a sensitivity bound $c$ along with a black-box function $f$.
\Cref{thm: subset extension} gives an $(\eps,\delta)$-privacy wrapper for functions with unbounded range.
\else
In this section, we state and prove \Cref{thm: subset extension}, our main result for the setting when the analyst provides a sensitivity bound $c$ along with a black-box function $f$.
\Cref{thm: subset extension} gives an $(\eps,\delta)$-privacy wrapper for functions with unbounded range.
\fi

\begin{theorem}[\ifnum\stoc =0 
    \algSE\ 
\else
    Subset extension
\fi 
privacy wrapper]
    \label{thm: subset extension}
    There exists a constant $a>0$ such that for every universe $\cU$, privacy parameters $\eps>0,\delta\in(0,1)$, and Lipschitz constant $c>0$, there exists an $(\eps,\delta)$-privacy wrapper $\cW$ over $\cU$ with noise distribution $\Lap\paren{\frac {a\cdot c}{\eps}}$ for all $c$-Lipschitz functions $f:\dom\to\R$ and all $x\in\dom$. Moreover, $\cW$ is $O\paren{\frac1\eps\log\frac{1}{\delta}}$-down local for all $x\in\dom$.
\end{theorem}

Since the \algSE\ mechanism is down local, it can be viewed as the following general feasibility result: Given a function $f$ and a dataset $x$, Lipschitzness of $f$ on large subsets of $x$ suffices for private and accurate release of $f(x)$.  

\ifnum\stoc=1
We defer the construction and proof to the full version \cite{LinderRSS25}.
\fi

\ifnum\stoc=0 %

\subsection{Stabilization and Conditional-Monotonization Operators and Their Properties}\label{sec:operators}

In this section, we introduce the $(\ell,h)$-stabilization operator, $\stab{\ell}{h}{\cdot}$, where the parameters $\ell$ and $h$ can be intuitively thought of as set sizes, and define the conditional-monotonization operator, $\cmon{\cdot}{}$. 
These operators are applied in the proof \Cref{thm: subset extension} as follows: given a function $f$, we first transform $f$ into $\cmonf$ and subsequently transform $\cmonf$ into $\mstablhf$. The composition of the two operators has three important properties. First, the value $\mstablhf(x)$ can be computed by querying $f$ on the down neighborhood of $x$; second, if $f$ is a Lipschitz function then $f(x)$ can be efficiently recovered from $\mstablhf(x)$; and third, for all neighboring $x\subset y$, the sequences $\{\mstablhf(x)\}_{h\geq\ell}$ and $\{\mstablhf(y)\}_{h\geq \ell}$ are ``interleaved". We use the first property to ensure that our mechanism is down local, the second to ensure it is accurate for the class of Lipschitz functions, and the third to guarantee that our mechanism is differentially private.

We start by recalling a notion of stability from \cite{KohliL23}. Given a function $f$, a point $u\in\dom$ is stable with respect to $f$ if, intuitively, $f$ is Lipschitz on large subsets of $u$. 

\begin{definition}[$\ell$-stable \cite{KohliL23}]
    \label{def: stability}
    Let $f:\dom\to\cY$ where $\cY\subseteq \R$. For $\ell\in\Z$, a point $x\in\dom$ is \emph{$\ell$-stable} with respect to $f$ if $|x|\geq\ell$ and %
    $f$ is Lipschitz over the domain $\{x'\subseteq x: |x'|\geq\ell\}$. 
\end{definition}

The key observation made in \cite{KohliL23} is that if $x$ and $y$ are $\ell$-stable and $|x\cap y|\geq \ell$ then $|f(x)-f(y)|\leq 2(\max(|x|,|y|)-\ell)$. They directly apply this observation to obtain an $(\eps,\delta)$-privacy wrapper.
There is no accuracy analysis provided in \cite{KohliL23}. For completeness, we show in \Cref{sec:Kohli-Laskowski} that an (adjusted) version of their algorithm has
$(\frac1{\eps^2}\log\frac1\delta\log\frac1\beta,\beta)$-accuracy for the class of Lipschitz functions.
In the proof of \Cref{thm: subset extension}, we use our new  operator, $(\ell,h)$-stabilization, to obtain an $(\eps,\delta)$-privacy wrapper with the stronger guarantee of $(\frac1\eps\log\frac1\beta,\beta)$-accuracy.

Next, we define the $(\ell,h)$-stabilization operator $\stab{\ell}{h}{\cdot}$. For all $f:\dom\to\R$, all $x\in\dom$, and all $\ell\leq h\leq |x|,$ the function $\stablhf$ evaluated at $x$ returns the maximum value achieved by $f$ on the $\ell$-stable subsets of $x$ with at least $h$ elements. Note that $h$ can be less than $\ell$, and when this setting of parameters is realized, $\stablhf=\stab{\ell}{\ell}{f}$ (since all $\ell$-stable subsets have size at least $\ell$).
The definition is illustrated in \Cref{fig:stabset}.

\begin{definition}[$\stabset{\ell}{h}{f}$, $(\ell,h)$-stabilization $\stab{\ell}{h}{f}$]\label{def:stab-functions}
    Let $f:\dom\to\cY$ where $\cY\subseteq \R$. For all $\ell, h\in \Z$, let $\stabsetlhf(x) = \{ x' \subseteq x: |x'|\geq h \text{\ and $x'$ is $\ell$-stable w.r.t.\ $f$}\}$. 
      Define the \emph{$(\ell,h)$-stabilization of $f$} as the function\footnote{Recall that we use the convention that if $\cY$ is unbounded below, then $\inf(\cY)=-\infty$.}
    \[
    \stablhf(x)=\max\Bparen{\bset{f(x'): x'\in \stabsetlhf(x)}\cup\bset{\inf(\cY)}}.
    \]
\end{definition}
    \begin{figure}[ht]
        \centering        \includegraphics[scale=.45]{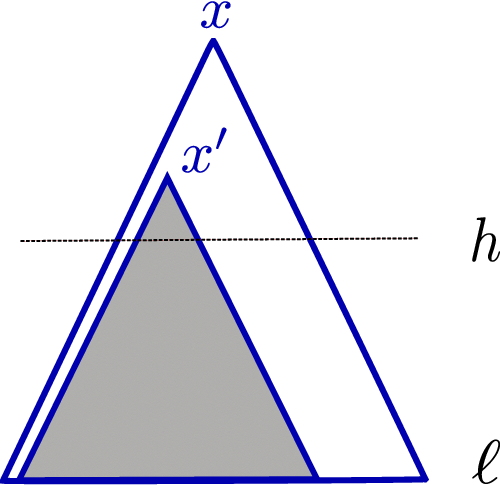}
        \caption{A set $x$ with a subset $x'$ of size at least $h$. If a function $f$ is Lipschitz on the shaded region---that is, on the subsets of $x'$ of size at least $\ell$, then $x'$ is $\ell$-stable with respect to $f$. Using our notation, $x'\in \stabsetlhf(x)$.}\label{fig:stabset}
    \end{figure}
    
  \Cref{lem: structure} identifies several important structural properties of the $(\ell,h)$-stabilization operator. First, we show that the sequences $\{\mstablhf(x)\}_{h\geq\ell}$ and $\{\mstablhf(y)\}_{h\geq \ell}$ are ``interleaved" for neighboring $x$ and $y$; this will be important to prove privacy of our privacy wrapper. Second, we prove that whenever $f$ is monotone and Lipschitz then $\stablhf(x)=f(x)$; this will be important for analyzing accuracy.

\begin{lemma}[Structure of $\stablhf{}$]
        \label{lem: structure}
        For all $f:\dom\to\cY$, where $\cY\subseteq\R$, and all  $\ell,h\in\Z$, where $h\geq\ell$:
        \begin{enumerate}\itemsep0em
            \item\label{item: structure 1} The function $\stab{\ell}{\cdot}{f}(u)$ is nonincreasing on  $\{\ell,\ell+1,\dots\}$, that is, $\stablhf(u)\geq \stab{\ell}{h+1}{f}(u)$.
            \item\label{item: structure 2} Let $u,v\in\dom$ be neighbors such that $v\subset u$. Then 
            \[
            \stab{\ell}{h+1}{f}(u)-1\leq\stab{\ell}{h}{f}(v) \, \leq \stab{\ell}{h}{f}(u).
            \]
            \item\label{item: structure 3} Let $u\in\dom$ and suppose that $h\leq|u|$. If the restriction of $f$ to the domain $\DN_{|u|-\ell}$ is Lipschitz and monotone then $\stablhf(u)=f(u)$.
        \end{enumerate}
    \end{lemma}

    \begin{figure}[ht]
        \centering
        \includegraphics[scale=.45]{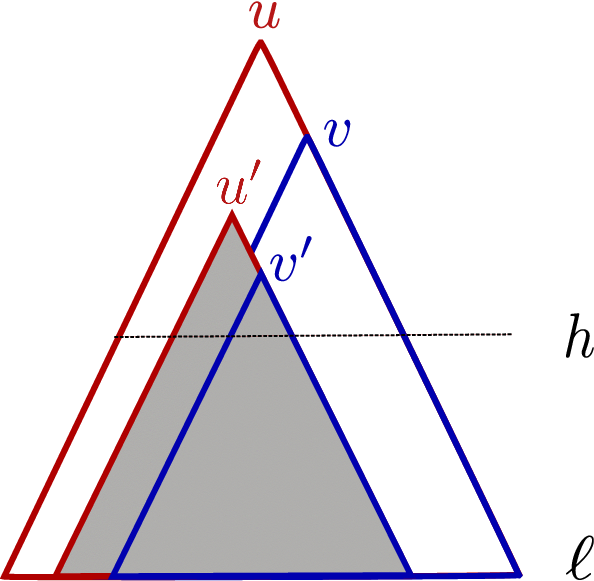}
        \caption{Sets $u$ and $v$ with $\ell$-stable subsets $u'$ and $v'$, each of size at least $h$. Notice that every stable subset of $u$ is at distance $1$ from a stable subset of $v$.}\label{fig:interleaving}
    \end{figure}
    
    \begin{proof}
        When $h<0$, then by \Cref{def:stab-functions}, the sets $\stabsetlhf(u)$ and $\stabset{\ell}{h+1}{f}(u)$ are the same. Thus, for each item in \Cref{lem: structure}, we can without loss of generality assume $h\geq 0$. We encourage the reader to reference \Cref{fig:interleaving} throughout the proof.
        
        To prove \Cref{item: structure 1}, notice that $\stabset{\ell}{h+1}{f}(u)\subseteq \stabsetlhf(u)$. By definition of $\stablhf(u)$ (a max over the set $\stabsetlhf(u)$), and the fact that $h\geq \ell$, we see that $\stab{\ell}{h+1}{f}(u)\leq \stablhf(u)$. 
        
        Next, we prove \Cref{item: structure 2}. To prove the second inequality, observe that if $v\subset u$ then $\stabsetlhf(v)\subseteq \stabsetlhf(u)$. Inspecting the definition of $\stablhf$, we see that $\stablhf(v)\leq \stablhf(u)$. To prove the first inequality, suppose that $\stabset{\ell}{h+1}{f}(u)\neq\emptyset$. Then $h<|u|$ and hence $\ell\leq h\leq |v|$. Moreover, for each
        $u'\in \stabset{\ell}{h+1}{f}(u)$, the neighbor $v'=u'\cap v$ is a subset of $v$ and, since $|u'|\geq h+1$, we have $|v'|\geq h\geq\ell$. Since $u'$ is $\ell$-stable and $v'$ is a subset of $u$ with $|v'|\geq h\geq \ell$ we see that $v'$ is $\ell$-stable and $v'\in \stabsetlhf(v)$. Since $|u'|=|v'|+1$ we have $f(v')\geq f(u')-1$. It follows that $\stab{\ell}{h+1}{f}(u)-1\leq \stablhf(v)$ whenever $\stabset{\ell}{h+1}{f}(u)\neq\emptyset$. On the other hand, if $\stabset{\ell}{h+1}{f}(u)=\emptyset$ then $\stab{\ell}{h+1}{f}(u)=\inf(\cY)$ which by definition is at most $\stablhf(v)$. 
        
        To prove \Cref{item: structure 3}, fix $u\in\dom$ and $\ell\leq h\leq |u|$. Suppose $f$ is Lipschitz and monotone on the domain $\DN_{|u|-\ell}$. %
        Then $u\in\stabsetlhf(u)$ and $f(u)\geq f(v)$ for all $v\in\stabsetlhf(u)$. Hence, $\stablhf(u)=f(u)$.
\end{proof}

We now define the conditional-monotonization operator $\cmonf$. Informally, \Cref{lem: cmonf} states that if $f$ is Lipschitz then $\cmonf$ is Lipschitz and monotone. 
Recall that by \Cref{lem: structure}, 
whenever $g$ is Lipschitz and monotone 
we have $\stab{\ell}{h}{g}(x)=g(x)$. It follows that when $f$ is Lipschitz then $\mstablhf(x)=\cmonf(x)$.

\begin{definition}[Conditional monotonization $\cmonf$]
    \label{def: cmonf}
    Fix $f:\dom\to\R$ and define the conditional monotonization of $f$ as the function 
    \[
        \cmonf(x)=\tfrac{1}{2}(f(x)+|x|).
    \]
\end{definition}

\begin{lemma}[Lispchitz to monotone Lipschitz]
    \label{lem: cmonf}
    Fix a function $f:\dom\to\R$, a point $x\in\dom$, and an integer $\tau\in\Z$. If $f$ is Lispchitz on $\DN_\tau(x)$ then the function $\cmonf$ is Lipschitz and monotone on $\DN_\tau(x)$.
\end{lemma}
\begin{proof}
Suppose $f$ is Lipschitz on $\DN_\tau(x)$. Consider the function $g(x)=f(x)+|x|.$ Let $u,v\in\DN_{\tau}(x)$ be neighbors such that $v\subset u.$  Since $f$ is Lipschitz, $f(u)-f(v)$ is in $[-1,1]$, so $g(u)-g(v)=f(u)-f(v)+1$ is in $[0,2].$ Thus, function $g$ is monotone and 2-Lipschitz. Hence, $\cmonf=\frac g 2$ %
is monotone and Lipschitz. 
\end{proof}

Given a function $f$, we successively apply the operators $\cmon{\cdot}{}$ and $\stab{\ell}{h}{\cdot}$ to transform $f$ into a well behaved function. We obtain the following crucial property of the composition of the two operators: By Lemmas~\ref{lem: structure} and \ref{lem: cmonf}, if $f$ is Lipschitz, then $\mstablhf(x)=\cmonf(x)$. By the definition of $\cmonf$, we get $2\mstablhf(x)-|x|=f(x)$ (for appropriate choices of $\ell$ and $h$). That is, for Lipschitz $f$, we can easily recover the value $f(x)$ from the value of the transformed function on $x.$

\subsection{\algSE\ and Proof of \Cref{thm: subset extension}}
\label{sec:subset-extension}

In this section, we present the \algSE\ mechanism (\Cref{alg: subset extension}) and use it to prove \Cref{thm: subset extension}. 
One of the key ideas employed by our mechanism is that %
for all functions $f$ and neighbors $v\subset u$, the diameter of $f$ on the set of $\ell$-stable subsets of $u$ and $v$ can be bounded above by $|u|-\ell$. We take advantage of this observation via a carefully defined proxy function and a preliminary test step that, on input $x$, ensures there is a sufficiently large $\ell$-stable subset of $x$. 

\begin{proof}[Proof of \Cref{thm: subset extension}]
Our main tool in the proof of \Cref{thm: subset extension} is the following proxy function. 

\begin{definition}[Proxy function $\proxyTltf$]
    \label{def: subset extension proxy}
    Let $f:\dom\to\R$ and fix $\ell\in\Z$. Let 
    \[
    \maxstab{\ell}{f}(x)=\max\set{|u|:u\in \stabset{\ell}{\ell}{f}(x)}.
    \]
    That is, $\maxstablf(x)$ is the size of the largest subset of $x$ that is $\ell$-stable with respect to $f$. When there is no ambiguity, we write $\maxstabl$ instead of $\maxstablf$. Fix $\tau\in\N$ and define the following function:
    \[
    \proxyTltf(x)=\Ex_{h\sim\{\maxstabl(x)-\tau,\dots, \maxstabl(x)\}}\bbrackets{\mstablhf(x)}.
    \]
\end{definition}

Next, we present the \algSE\ mechanism (\Cref{alg: subset extension}) and complete the proof of \Cref{thm: subset extension} by arguing that \Cref{alg: subset extension} is a privacy wrapper with the desired properties.

\begin{algorithm}[H]
	\begin{algorithmic}[1]
		\caption{\label{alg: subset extension} \algSE\ Mechanism}
        \Statex \textbf{Parameters:} privacy parameters $\eps>0$ and $\delta\in(0,1)$
	    \Statex \textbf{Input:} $x\in\dom$, query access to $f:\dom\to\R$, Lipschitz constant $c>0$
	    \Statex \textbf{Output:} $y\in\R \cup \{\perp\}$
        \State $\eps_0\gets\frac \eps 3, \delta_0\gets\frac \delta 2, q\gets 20$, and $\tau\gets\lceil\frac1{\eps_0}\ln\frac1{\delta_0}\rceil$
		\State  \textbf{release} $\ell\gets\lceil|x|-q\tau+ R_0\rceil$ where $R_0\sim \TLap(\frac1{\eps_0},\tau)$
        \Comment{\Cref{def: laplace}}
            \State \textbf{release} $b\gets  \Ind\Bset{ \maxstablf(x)+R_1\leq\frac12(|x|+\ell)+5\tau}$ where $R_1\sim \TLap(\frac2{\eps_0},2\tau)$
        \If{$b=0$}  
        \State \Return $2\proxyT{\ell}{\tau}{f}(x)-|x|+Z$ where $Z\sim\Lap\paren{\frac {10q}{\eps_0}}$
        \Else{ \Return$\perp$}
        \EndIf
	\end{algorithmic}
\end{algorithm}

Before completing the proof of \Cref{thm: subset extension} we bound the sensitivity of the proxy function $\proxyTltf$ defined above.

\subsubsection{Bounding the sensitivity of the proxy function $\proxyTltf$}
\label{sec:T_bound}

To bound the sensitivity of $\proxyTltf$, we first relate the sizes of the largest $\ell$-stable subsets of two neighboring datasets. 

\begin{claim}[Sensitivity of $\maxstabl$]
    \label{claim:m-l}
    Let $f:\dom\to\R$ and fix $\ell\in\Z$ and $\tau\in\N$. Fix two neighbors $u,v\in\dom$ such that $v\subset u$ and $|v|\geq \ell$. Then $\maxstabl(v)=\maxstabl(u)$ or $\maxstabl(v)=\maxstabl(u)-1$.
    \end{claim}
    \begin{proof}
        Since $v\subset u$, we have $\maxstabl(v)\leq \maxstabl(u)$. Let $p$ be an $\ell$-stable subset of $u$ with $|p|=\maxstabl(u)$. If $p$ is a subset of $v$ then $\maxstabl(v)=\maxstabl(u)$. Otherwise, consider the set $q=p\cap v$. Then $q$ is $\ell$-stable  because $p$ is $\ell$-stable. Since $q\subset v$ and $|q|=|p|-1$, it follows that $\maxstabl(v)\geq\maxstabl(u)-1$.
    \end{proof}

Next, we bound the diameter of the image of $\cmonf$ on the set of sufficiently large $\ell$-stable subsets of two neighboring datasets. We use the following notation: For all functions $g:\dom\to\R$ and sets $S\subset \dom$, let $g(S)=\{g(x) : x\in S\}$. 

\begin{claim}[Bounded diameter]
    \label{claim:diameter_bound}
    Let $f,\ell,\tau, u,$ and $v$ be as in the premise of \Cref{claim:m-l}, and suppose $\maxstabl(v)-\tau>\frac12(|v|+\ell)$. Then the diameter of the set $\cmonf\paren{\stabset{\ell}{\maxstabl(u)-\tau}{f}(u)\cup \stabset{\ell}{\maxstabl(v)-\tau}{f}(v)}$ is at most $|u|-\ell$. 
\end{claim}
\begin{proof}
Consider sets $p,q$ in $ \stabset{\ell}{\maxstabl(u)-\tau}{f}(u)\cup \stabset{\ell}{\maxstabl(v)-\tau}{f}(v).$
We prove the claim by comparing $\cmonf(p)$ and $\cmonf(q)$ to $\cmonf(p\cap q)$. 
Since $p$ and $q$ are $\ell$-stable, the function $f$ is Lipschitz on $\{p'\subseteq p: |p'|\geq\ell\}$ and on $\{q'\subseteq q: |q'|\geq\ell\}$. By \Cref{lem: cmonf}, the function $\cmonf$ is also Lipschitz on these sets. Next, we demonstrate that $p\cap q$ belongs to both of them by showing that $|p\cap q|\geq \ell$.

Since $\maxstabl(v)-\tau>\frac12(|v|+\ell)$ and $\maxstabl(u)\geq\maxstabl(v)$, we get that
$\min(|p|,|q|)>\frac12(|v|+\ell)$. Consequently,
    $$\max(|u\setminus p|,|u\setminus q|)< |u|-\frac12(|v|+\ell)= \frac12(|u|-\ell+1).$$
Since  $2\max(|u\setminus p|,|u\setminus q|)$ and $|u|-\ell+1$ are integers, we obtain that    
     \begin{align}\label{eq:max-cordinality}
     2\max(|u\setminus p|,|u\setminus q|)\leq |u|-\ell,
     \end{align}
     and hence
     $
     |p\cap q|\geq|u|-|u\setminus p|-|u\setminus q|\geq 2\max(|u\setminus p|,|u\setminus q|)\geq \ell.
     $

Therefore, $p\cap q$ is in  $\{p'\subseteq p: |p'|\geq\ell\}$ and in $\{q'\subseteq q: |q'|\geq\ell\}$. Since $\cmonf$ is Lipschitz on these two sets, we get that $\cmonf$ is Lipschitz on $\{p,p\cap q, q\}$.
By the triangle inequality,
\begin{align}
    |\cmonf(p)-\cmonf(q)|
    &\leq|\cmonf(p)-\cmonf(p\cap q)|+|\cmonf(q)-\cmonf(p\cap q)|\nonumber\\
     &\leq|p\setminus (p\cap q)|+|q\setminus (p\cap q)| \label{eq:cmonf-is-Lipschitz}\\
    &\leq|u\setminus q|+|u\setminus p| \label{eq:p-q-subsets-of-u}\\
    &\leq |u|-\ell \label{eq:bound-on-diameter-of-cmonf},
\end{align}
where \eqref{eq:cmonf-is-Lipschitz} holds because $\cmonf$ is Lipschitz on $\{p,p\cap q, q\}$, then \eqref{eq:p-q-subsets-of-u} holds because $p$ and $q$ are subsets of $u$, and \eqref{eq:bound-on-diameter-of-cmonf} holds by \eqref{eq:max-cordinality}.
Thus, the diameter of $\cmonf\paren{\stabset{\ell}{\maxstabl(u)-\tau}{f}(u)\cup \stabset{\ell}{\maxstabl(v)-\tau}{f}(v)}$ is at most $|u|-\ell$.
\end{proof}

Next, we use Claims~\ref{claim:m-l} and \ref{claim:diameter_bound} to bound the sensitivity of $\proxyTltf$.

\begin{lemma}[Sensitivity of $\proxyTltf$]
    \label{lem:T_bound} 
   Let $f,\ell,u,\tau,$ and $v$ be as in the premise of \Cref{claim:diameter_bound}.
    Then
    \[
    |\proxyTltf(u)-\proxyTltf(v)|\leq 1+\frac{2(|u|-\ell)}{\tau}.
    \]
\end{lemma}

\begin{proof}
    For convenience, we introduce the following notation: For all $h\geq \ell$, define the function $g_h:\dom\to\R$ by $g_h(z)=\mstablhf(z)$. Additionally, let $H$ denote the set $\{\maxstabl(u)-\tau,\dots,\maxstabl(u)-1\}$.

    First, we expand the definition of $\proxyTltf$ to get
    \[
    \abs{\proxyTltf(u)-\proxyTltf(v)}=\abs{\Ex_{h_1\sim\{\maxstabl(u)-\tau,\dots,\maxstabl(u)\}}[g_{h_1}(u)]-\Ex_{h_2\sim\{\maxstabl(v)-\tau,\dots,\maxstabl(v)\}}[g_{h_2}(v)]}.
    \]  
    By definition of $H$, the random variable $h_1$ is supported on 
    the set $H\cup\{\maxstabl(u)\}$. By \Cref{claim:m-l}, the support of the random variable $h_2$ is contained in the set 
    $H\cup\{\maxstabl(v), \maxstabl(v)-\tau\}$. By the law of total expectation and the triangle inequality, 
    \begin{align}
    \abs{\proxyTltf(u)-\proxyTltf(v)}
    &\leq \abs{\Ex_{h\in H}[g_h(u)-g_{h}(v)]}\nonumber\\
    &+\abs{g_{\maxstabl(u)}(u)-g_{\maxstabl(v)}(v)}\cdot\frac{\Ind\brackets{\maxstabl(u)=\maxstabl(v)}}{\tau}\label{eq:difference-possibility1}\\
    &+\abs{g_{\maxstabl(u)}(u)-g_{\maxstabl(v)-\tau}(v)}\cdot\frac{\Ind\brackets{\maxstabl(u)=\maxstabl(v)+1}}{\tau}.\label{eq:difference-possibility2}
    \end{align}
    At most one of \eqref{eq:difference-possibility1} and \eqref{eq:difference-possibility2} is nonzero. By \Cref{claim:diameter_bound}, and the fact that $g_h=\mstablhf$ is the maximum over a set of points with bounded diameter, each of them is at most $\frac{|u|-\ell}{\tau}$.
    
    Next, we bound $\abs{\Ex_{h\in H}[g_h(u)-g_{h}(v)]}$. By \Cref{item: structure 2} 
 of \Cref{lem: structure}, we have $g_{h+1}(u)-1\leq g_h(v)\leq g_{h}(u)$ for all $h\in H$. The upper bound on $g_h(v)$ allows us to remove the absolute value, and the lower bound allows us to replace $g_h(v)$ by $g_{h+1}(u)-1$---that is,
    \[
    \abs{\Ex_{h\in H}[g_h(u)-g_{h}(v)]}\leq \Ex_{h\in H}[g_h(u)-g_{h+1}(u)+1]=1+\frac{g_{\maxstabl(u)-\tau}(u)-g_{\maxstabl(u)}(u)}{\tau}\leq 1+\frac{|u|-\ell}{\tau},
    \]
    where the equality follows since all but the first and last terms in the expectation telescope, and the final inequality follows from \Cref{claim:diameter_bound}. Putting it all together yields the desired conclusion of
    \[
    \abs{\proxyTltf(u)-\proxyTltf(v)}\leq 1+\frac{2(|u|-\ell)}{\tau}.\qedhere
    \]
\end{proof}

\subsubsection{Completing the proof of \Cref{thm: subset extension}}  

To prove \Cref{thm: subset extension}, we show that \Cref{alg: subset extension} is $(\eps,\delta)$-DP and $O(\frac1\eps\log\frac1{\delta})$-down local, and, whenever $f$ is Lipschitz, outputs $f(x)+\Lap(\frac{10q}\eps)$.

\paragraph{Privacy.} Fix a function $f:\dom\to\R$. To analyze the privacy of  $\cW^f$ (\Cref{alg: subset extension}), we consider the steps of $\cW^f$ as separate algorithms defined as follows:

\begin{enumerate}
    \item Let $\cL(x)$ be the algorithm that releases $\lceil|x|-q\tau+R_0\rceil$ where $R_0\sim\TLap(\frac1{\eps_0},\tau)$, and let $\widehat\cL(x)$ denote the set of possible outputs of $\cL(x)$.
\end{enumerate}

Additionally, for all fixed $\ell\in\Z$,  

\begin{enumerate}[resume]
    \item Let $\cT_\ell(x)$ be the algorithm that releases 
    $b\gets \Ind\Bset{ \maxstablf(x)+R_1\leq\frac12(|x|+\ell)+5\tau}$ where $R_1\sim \TLap(\frac2{\eps_0},2\tau)$.
    \item Let $\cA_\ell(x)$ be the algorithm which releases $2\proxyT{\ell}{\tau}{f}(x)-|x|+Z$ where $Z\sim\Lap\paren{\frac {10q}{\eps_0}}$.
    \item Let $\cP_\ell(x)$ be the algorithm which releases $\cA_\ell(x)$ if $\cT_\ell(x)=0$ and returns $\perp$ otherwise. 
\end{enumerate} 

To prove that $\cW^f$ is private, we first argue that $\cP_\ell$ is private for all $\ell\in\widehat\cL(x)\cup\widehat \cL(y)$. The proof follows the propose-test-release framework of \cite{DworkL09}. First, we show that the ``test" algorithm $\cT_\ell$ is $(\eps_0,\delta_0)$-DP. 

\begin{definition}
    \label{def:indistinguishable}
    Random variables $Z$ and $Z'$ over $\R$ are $(\eps,\delta)$-indistinguishable, denoted $Z\approx_{\eps,\delta} Z$, if for all measurable sets $E\subseteq\R$, we have $\Pr[Z\in E]\leq e^\eps\Pr[Z'\in E]+\delta$ and $\Pr[Z'\in E]\leq e^\eps\Pr[Z\in E]+\delta$. 
\end{definition}

    \begin{claim}
        \label{claim:private_test}
        Fix $\ell\in \Z$ and neighbors $x,y\in\dom$ such that $\ell\leq \min(|x|,|y|)$. Then $\cT_\ell(x)\approx_{\eps_0,\delta_0}\cT_\ell(y)$.
    \end{claim}
    \begin{proof}
        Let $g(x)=\maxstablf(x)-\frac12(|x|+\ell)+2\tau$. By \Cref{claim:m-l}, we have $|\maxstabl(x)-\maxstabl(y)|\leq 1$, and hence $|g(x)-g(y)|\leq 2$. By \Cref{fact: laplace mechanism}, the mechanism that releases $g(x)+\TLap(\frac{2}{\eps_0},2\tau)$ is $(\eps_0,\delta_0)$-DP. Since $\cT_\ell(x)$ is a postprocessing of this mechanism, \Cref{fact:postprc} implies the claim. 
    \end{proof}

 Next, we argue that if $\proxyTltf$ is not Lipschitz on the set $\{x,y\}$, then $\cT_\ell(x)$ and $\cT_\ell(y)$ both output $1$. Let $G_\ell=\set{(x,y) \colon \text{$x,y\in\dom$ are neighbors and } \abs{\proxyTltf(x)-\proxyTltf(y)}\leq 3q}$.

    \begin{claim}
    \label{claim:bad_set_test_bound}
        Let $x,y\in\dom$ be neighbors and fix $\ell\in \widehat\cL(x)\cup \widehat\cL(y)$. If $(x,y)\not\in G_\ell$ then $\cT_\ell(x)=\cT_\ell(y)=1$. 
    \end{claim}
    \begin{proof}
        Assume w.l.o.g. that $x\subset y$. By the definition of $\cL$, since $\ell\in\widehat\cL(x)\cup \widehat\cL(y)$, we have $\ell\geq |x|-(q+1)\tau$. Thus, $|y|-\ell\leq 1+(q+1)\tau$ and $1+2(|y|-\ell)/\tau\leq 3q$. Now, by \Cref{lem:T_bound}, if $(x,y)\not\in G_\ell$ then $\maxstabl(x)-\tau\leq \frac12(|x|+\ell)$. Since the randomness $R_1$ sampled by $\cT_\ell$ is at most $2\tau$, we have $\maxstabl(x)+R_1\leq \frac12(|x|+\ell)+3\tau$, and therefore $\cT_\ell(x)=1$. To see why $\cT_\ell(y)=1$, recall that \Cref{claim:m-l} implies $\maxstabl(x)\geq \maxstabl(y)-1$. Therefore, $\maxstabl(y)+R_1\leq \frac12(|x|+\ell) + 3\tau+1$. Since $\tau\geq 1$ and $|y|=|x|+1$, we have $\maxstabl(y)+R_1\leq \frac12(|y|+\ell)+5\tau$, and therefore $\cT_\ell(y)=1$. 
    \end{proof}
    
Next, we use Claims~\ref{claim:private_test} and \ref{claim:bad_set_test_bound} to prove \Cref{lem:fixed_ell_privacy}, which states that $\cP_\ell$ is DP for all $\ell\in\widehat\cL(x)\cup\widehat\cL(y)$.

\begin{lemma}[Privacy for fixed $\ell$]
    \label{lem:fixed_ell_privacy}
    Let $x,y$ and $\ell$ be as in \Cref{claim:bad_set_test_bound}. Then $\cP_\ell(x)\approx_{2\eps_0,\delta_0}\cP_\ell(y)$.
\end{lemma}
\begin{proof}
    Consider the following two cases. 
    \subparagraph{Case 1.} Suppose $(x,y)\in G_\ell$. By the definition of $G_\ell$,
    \[
    \abs{2\proxyTltf(x)-|x| - 2\proxyTltf(y) + |y|}\leq 6q+1.
    \]  
    Since the noise is sampled from $\Lap\paren{\frac{10q}{\eps_0}}$, 
    \Cref{fact: laplace mechanism} (about privacy of
    the Laplace mechanism) implies that $\cA_\ell(x)\approx_{\eps_0,0}\cA_\ell(y)$. Moreover, since $\ell\in\widehat\cL(x)\cup\widehat\cL(y)$, we have $\ell\leq\min(|x|,|y|)$, and hence, by  \Cref{claim:private_test}, we have $\cT_\ell(x)\approx_{\eps_0,\delta_0}\cT_\ell(y)$. Thus, Facts~\ref{fact:composition} and \ref{fact:postprc} (about composition and postprocessing) 
imply that $\cP_\ell(x)\approx_{2\eps_0,\delta_0}\cP_\ell(y)$, which completes the analysis of the first case. 
    
    \subparagraph{Case 2.} Suppose $(x,y)\not\in G_\ell$. Then, since %
    $\ell\in \widehat\cL(x)\cup\widehat\cL(y)$, \Cref{claim:bad_set_test_bound} implies $\cT_\ell(x)=\cT_\ell(y)=1$. Therefore, $\cP_\ell(x)$ and $\cP_\ell(y)$ both output $\perp$.     
    
    It follows that in both cases $\cP_\ell(x)\approx_{2\eps_0,\delta_0} \cP_\ell(y)$, which completes the proof \Cref{lem:fixed_ell_privacy}.
\end{proof}

It remains to prove that $\cW^f$ is $(\eps,\delta)$-DP. Since $x$ and $y$ are neighbors, the Laplace mechanism (see \Cref{fact: laplace mechanism}) implies that $\cL(x)\approx_{\eps_0,\delta_0} \cL(y)$. Additionally, \Cref{lem:fixed_ell_privacy} implies that $\cP_\ell(x)\approx_{2\eps_0,\delta_0}\cP_\ell(y)$ for all $\ell\in \widehat\cL(x)\cup\widehat\cL(y)$. Since $\cW^f(x)$ first releases $\ell\sim \cL(x)$ and then releases $\cP_\ell(x)$, basic composition (\Cref{fact:composition}) implies that $\cW^f(x)\approx_{3\eps_0,2\delta_0} \cW^f(y)$. Since $\eps_0=\eps/3$ and $\delta_0=\delta/2$, algorithm $\cW^f$ is $(\eps,\delta)$-DP.

\paragraph{Locality.}
Next, we prove the down locality guarantee. By the setting of $\ell$ in \Cref{alg: subset extension} and the fact that $|R_0|\leq \tau$, we have $|x|-\ell\leq 2q\tau$. Therefore, $\cW$ need only query $f$ on $\DN_{2q\tau}(x)$. Since $\tau=O(\frac1{\eps_0}\ln\frac{1}{\delta_0}+1)=O(\frac1{\eps}\log\frac{1}{\delta}+1)$ the locality is $O(\frac1\eps\log\frac1\delta + 1)$.

\paragraph{Accuracy.}
Observe that whenever $x$ is $\ell$-stable, we have $\maxstabl(x)=|x|$. Therefore, 
\[
\maxstabl(x)-\frac12(|x|+\ell)-5\tau \geq \frac12(q\tau-\tau-1)-5\tau\geq q\tau/2-6\tau.
\]
Since $q>16$ we have $\maxstabl(x)-\frac12(|x|+\ell)-5\tau>2\tau$.
Since $|R_1|\leq 2\tau$, algorithm $\cT_\ell(x)$ outputs $0$ for all $\ell\in\widehat\cL(x)$. Hence, for all $\ell\in\widehat\cL(x)$ algorithm $\cP_\ell(x)$ outputs $2\proxyTltf(x)-|x|+Z$ where $Z\sim\Lap\paren{\frac{10q}{\eps_0}}$. By \Cref{lem: cmonf}, if $f$ is Lipschitz then $\cmonf$ is Lipschitz and monotone. 
Thus, by \Cref{lem: structure}, we obtain $\stab{\ell}{h}{\cmonf}(x)=\cmonf(x)$, and therefore $2\proxyTltf(x)-|x|=f(x)$.
\end{proof}

\fi %

\section{Locality Lower Bound}
\label{sec: locality lower bound}
\newcommand{\pl}{\textsf{pl}}
\newcommand{\med}{\textsf{med}}
In this section, we 
\ifnum\stoc =0 prove 
\else state \fi
a lower bound on the down locality of every privacy wrapper with an $(\alpha,\beta)$-accuracy guarantee for constant functions, and a lower bound on privacy wrappers that achieve the same accuracy guarantee as that of \Cref{thm:generalized-shifted-inverse-mechanism}. 

\begin{theorem}[Locality lower bound]
    \label{thm: locality lower bound}
    For all $\alpha>0$, all $\eps,\delta,\beta\in(0,1)$, and all $r>2\alpha$, every $(\eps,\delta)$-privacy wrapper that is $\lambda$-down local, and $(\alpha,\beta)$-accurate for all constant functions $f:\dom\to\set{2\alpha,4\alpha,...,r}$ and $x\in\dom$, must have $\lambda\geq \Omega\bparen{\frac{1}{\eps}\log\min\bparen{\frac{r}{\alpha\cdot\beta},\frac1\delta}}$.
\end{theorem}

An important feature of \Cref{thm: locality lower bound} is that it holds even for privacy wrappers that are only accurate for constant functions. Since constant functions are a subset of Lipschitz functions, the lower bound implies that the locality of many of our constructions is tight. Additionally, we deduce an analogous lower bound for privacy wrappers that are $(DS_\lambda^f(x),\beta)$-accurate. Such privacy wrappers are, in particular, $(\alpha,\beta)$-accurate for constant functions (the down sensitivity is zero) and all $\alpha\geq0$. Taking $\alpha=\frac 1 2$ in \Cref{thm: locality lower bound} we obtain \Cref{cor: locality lower bound}.

\begin{corollary}
    \label{cor: locality lower bound}
    For all $\eps,\delta,\beta\in(0,1)$, and $r\in \N$, every $(\eps,\delta)$-privacy wrapper that is $\lambda$-down local, and $(DS_\lambda^f(x),\beta)$-accurate on all functions $f:\dom\to[r]$ and all inputs $x$ must have $\lambda\geq \Omega\bparen{\frac{1}{\eps}\log\min\bparen{\frac{r}{\beta},\frac1\delta}}$.
\end{corollary}

Our next theorem states that the locality of any privacy wrapper that achieves the same accuracy guarantee as that of \Cref{thm:generalized-shifted-inverse-mechanism} must depend on the cardinality of the range of the function. 

\begin{theorem}[Dependence on range for automated sensitivity detection]
    \label{thm:univ-dependent-lower-bound}
    Fix $\eps,\delta,\beta\in(0,1)$. Let $\cW$ be an $(\eps,\delta)$-privacy wrapper that is $\lambda$-down local, and has the following accuracy guarantee: For all $f:\dom\to[r]$, and $x\in\dom$
    \[
    \Pr\brackets{\cW^f(x)\in [\min f \bparen{\DN_\lambda(x)}, \max f \bparen{\DN_\lambda(x)}]}\geq1-\beta.
    \]
    Then $\cW$ must have locality $\lambda=\Omega\paren{\log^*(r)}$.    
\end{theorem}

\ifnum\stoc=1  %
We defer  proofs of \Cref{thm: locality lower bound,thm:univ-dependent-lower-bound}  to the full version \cite{LinderRSS25}.
\else

In the remainder of the section we prove \Cref{thm: locality lower bound,thm:univ-dependent-lower-bound}. The proofs proceed via reductions from the ``point distribution problem", and ``interior point problem" respectively. 

\begin{remark}[Between sets and multisets] 
\label{rem:set_multiset_map} One syntactic difficulty that arises in our proofs is that our privacy wrappers are defined for functions over sets, and the point distribution and interior point problems are concerned with multisets. We circumvent this issue by defining a mapping from multisets to sets, and a mapping from sets to multisets. This allows us to apply our privacy wrappers to functions over multisets. 

For a set $\cY$, let $\widetilde{\cY}$ denote the set of finite multisets of elements in $\cY$. Define the map $\phi$ by sending each multiset $s\in\widetilde{\cY}$ to a set of tuples $\phi(s)\in(\N\times\N)^*$. The map $\phi(s)$ sends each element $j\in s$ to the element $(j,i)$ for a unique $i\in\N$ (i.e., $\phi$ assigns unique labels to the elements of $s$). We also define the map $\psi:(\N\times\N)^*\to\widetilde{\cY}$ by setting $\psi(x)$ to the multiset consisting of the projection of every tuple $t\in x$ onto its first coordinate. Notice that $\psi(\phi(s))=s$ and that $|\phi(s)|=|s|$. In the remainder of the section, we will use the maps $\phi$ and $\psi$ to complete the proofs of \Cref{thm: locality lower bound,thm:univ-dependent-lower-bound}.
\end{remark}

\subsection{The Point Distribution Problem and The Proof of \Cref{thm: locality lower bound}} 
\label{sec:pd problem}

In this section, we define the point distribution problem, and prove \Cref{thm: locality lower bound}. 
Recall that $\widetilde{\cY}$ denotes the set of finite \emph{multisets} of elements in $\cY$.

\begin{definition}[Point distribution problem, sample complexity]
    \label{def: pd problem}
    Fix a set $\cY$, an integer $n\in\N$, and a failure probability $\beta\in(0,1)$. An algorithm $\cA$ \emph{solves the point distribution problem over $\cY$ with probability at least $1-\beta$ and sample complexity $n$} if for all $y\in\cY$ and input 
    $s\in\widetilde{\cY}$ such that $|s|=n$ the algorithm $\cA$ outputs $y$ with probability at least $1-\beta$ whenever $s$ consists of $n$ identical copies of $y$. 
\end{definition}

Our reduction will show that an $(\eps,\delta)$-privacy wrapper that is $\lambda$-down local can be used as a subroutine to solve the point distribution problem with sample complexity $\lambda+1$. Hence, in order to prove a lower bound on the locality $\lambda$ of every privacy wrapper, we require a lower bound on the sample complexity of any algorithm that solves the point distribution problem. 

\begin{lemma}[Point distribution hardness]
    \label{lem: pd hardness}
    There exists a constant $c>0$ such that for all sets $\cY$, and all privacy parameters $\eps,\delta\in(0,1)$, every $(\eps,\delta)$-DP algorithm that solves the point distribution problem over $\cY$ with probability at least $1-\beta$ must have sample complexity $n\geq \frac c\eps\log\min(\frac{|\cY|}{\beta},\frac1\delta)$. 
\end{lemma}

The proof of \Cref{lem: pd hardness} proceeds via standard packing arguments and can be found in \cite{BookAlgDP}.

\begin{proof}[Proof of \Cref{thm: locality lower bound}]

Fix parameters $\eps,\delta,\alpha$ and $r$ as in \Cref{thm: locality lower bound}. In order to prove the lower bound, we will construct a universe $\cY$, and an algorithm $\cA$ that calls an $(\eps,\delta)$-privacy wrapper $\cW$ with locality $\lambda$, and solves the point distribution problem over $\cY$ with probability at least $1-\beta$ and sample complexity $\lambda+1$. \Cref{lem: pd hardness}, then implies that $\lambda\geq \Omega(\frac1\eps\log\min(\frac{|\cY|}{\beta},\frac1\delta))$. We state and prove this reduction formally below. 

\begin{lemma}[Reduction from point distribution]
    \label{lem: reduction from pd}
    Fix parameters $\alpha>0$, $r\geq2\alpha$, and $\eps,\delta\in (0,1)$. Let $\cY=\{2\alpha,4\alpha,\dots,r\}$ and $\cU=\N\times\N$. Let $\cW$ be an $(\eps,\delta)$-privacy wrapper over $\cU$ that is $(\alpha,\beta)$ accurate for all constant functions $f:\dom\to\cY$ and all inputs $x\in\dom$. Suppose that $\cW$ is $\lambda$-down local for some $\lambda\in\N$. Then there exists an algorithm $\cA$ that solves the point distribution problem over $\cY$ with probability at least $1-\beta$ and sample complexity $\lambda+1$.
\end{lemma}
\begin{proof}
    The main idea in the reduction is to simulate $\cW$ on the plurality function. Notice that for every multiset $s\in\widetilde{\cY}$ consisting of identical copies of some $y\in\cY$, the plurality function is constant on subsets of $s$ of size at least $1$. Hence, if $\lambda<|s|$ then the $(\alpha,\beta)$-accuracy guarantee implies $\cW$ will output a value $a$ such that $|a-y|\leq\alpha$ with probability at least $1-\beta$. Since $\cY=\{2\alpha,4\alpha,\dots,r\}$, the elements of $\cY$ all differ by at least $2\alpha$. It follows that with probability at least $1-\beta$, the output of $\cW$ will be sufficient to exactly recover the plurality of $s$. We remark that although the plurality function is not constant on the entire domain, since $\cW$ is only allowed to make queries in $\DN_\lambda(s)$, a region where the plurality function is constant, it cannot distinguish between the plurality function and a function that is constant everywhere. Hence, it must satisfy the $(\alpha,\beta)$-accuracy guarantee. 
    
    Using the maps $\phi$ and $\psi$ defined in \Cref{rem:set_multiset_map}, we formally demonstrate the reduction. Let $\pl:(\N\times\N)^*\to\cY$ be the function that sends $x$ to the plurality of $\psi(x)$, with range truncated to the set $\cY$, that is, if $\pl(x)\not\in\cY$ then set $\pl(x)=r$. Let $\cA$ be the following algorithm: On input $s\in\widetilde{\cY}$ such that $|s|=\lambda+1$, simulate $\cW$ with query access to $\pl$ and input $\phi(s)$. Next, let $a\gets\cW^{\pl}(\phi(s))$ and output $\arg\min\{|j-a|: j\in\cY\}$. Suppose $s\in\widetilde{\cY}$ consists of identical copies of an element $y\in\cY$. Then for all nonempty subsets $\phi(s')\subseteq \phi(s)$ we have $\pl(\phi(s'))=y$. Since $\lambda<|s|=|\phi(s)|$ the function $\pl$ is constant on the domain $\DN_{\lambda}(\phi(s))$. By the $(\alpha,\beta)$-accuracy guarantee $|\cW^\pl(\phi(s))-y|\leq\alpha$ with probability at least $1-\beta$. Since the elements of $\cY$ all differ by at least $\alpha$ the algorithm $\cA$ outputs $y$ with probability at least $1-\beta$. Hence, $\cA$ solves the point distribution problem over $\cY$ with probability at least $1-\beta$ and sample complexity $\lambda+1$.
\end{proof}

To complete the proof of \Cref{thm: locality lower bound}, we combine \Cref{lem: pd hardness,lem: reduction from pd}, and the fact that $|\cY|=\frac r{2\alpha}$ to obtain $\lambda\geq \Omega\paren{\frac 1\eps\log\min(\frac{|r|}{2\alpha\beta},\frac1\delta)}$. 
\end{proof}

\subsection{The Interior Point Problem and The Proof of \Cref{thm:univ-dependent-lower-bound}}

In this section, we introduce the interior point problem and complete the proof of \Cref{thm:univ-dependent-lower-bound}.

\begin{definition}[Interior point problem]
    \label{def: ip problem}
    Fix a set $\cY$, an integer $n\in\N$, and a failure probability $\beta\in(0,1)$. An algorithm $\cA$ \emph{solves the interior point problem over $\cY$ with probability at least $1-\beta$ and sample complexity $n$} if for all inputs 
    $s\in\widetilde{\cY}$ of size $n$, the algorithm $\cA$ outputs $y\in\brackets{\min\{i\in s\},\max\{i\in s\}}$ with probability at least $1-\beta$. 
\end{definition}

Our next reduction shows that an $(\eps,\delta)$-privacy wrapper that is $\lambda$-down local, and satisfies the accuracy guarantee of \Cref{thm:univ-dependent-lower-bound}, can be used to solve the interior point problem with sample complexity $\lambda+1$. To complete the proof of \Cref{thm:univ-dependent-lower-bound}, we use the following result of \cite{BunNSV15}.

\begin{lemma}[Interior point hardness (Theorem 1.2 \cite{BunNSV15})]
    \label{lem:ip hardness}
    There exists a constant $c>0$ such that for all sets $\cY$, and all privacy parameters $\eps,\delta\in(0,1)$, every $(\eps,\delta)$-DP algorithm that solves the interior point problem over $\cY$ with probability at least $1-\beta$ must have sample complexity $n\geq c\log^*|\cY|$. 
\end{lemma}

\begin{proof}[Proof of \Cref{thm:univ-dependent-lower-bound}]
Fix parameters $\eps,\delta$ and $r$ as in \Cref{thm:univ-dependent-lower-bound} and let $\cY=[r]$. In order to prove the lower bound, we will construct an algorithm $\cA$ that uses an $(\eps,\delta)$-privacy wrapper $\cW$ that has locality $\lambda$, and satisfies the accuracy guarantee of \Cref{thm:univ-dependent-lower-bound}, to solve the interior point problem over $\cY$ with probability at least $1-\beta$ and sample complexity $\lambda+1$. 

\begin{lemma}[Reduction from interior point]
\label{lem: reduction from ip}
Fix parameters $\eps,\delta,\beta\in(0,1)$, and $r\in\N$. Let $\cY=[r]$, and $\univ=[r]\times \N$. Let $\cW$ be an $(\eps,\delta)$-privacy wrapper over $\univ$ that is $\lambda$-down local, and suppose that for all $f:\dom\to\cY$ and $x\in\dom$, the wrapper outputs $\cW^f(x)\in\brackets{\min f(\DN_\lambda(x)), \max f(\DN_\lambda(x)}$ with probability at least $1-\beta$. Then there exists an algorithm $\cA$ that solves the interior point problem over $\cY$ with probability at least $1-\beta$ and sample complexity $\lambda+1$.     
\end{lemma}
\begin{proof}
    Let $\med:\widetilde{\cY}\to\cY$ be the a function that outputs a median of $x$ for all $x\in\widetilde{\cY}$. Notice that for all $x\in\widetilde{\cY}$ such that $\lambda>|x|$, we have $\min \med(\DN_\lambda(x))\geq \min\{i\in x\}$, and $\max\med(\DN_\lambda(x))\leq \max\{i\in x\}$. Below, we use this fact to prove the reduction from interior point. 

    Recall the maps $\phi$ and $\psi$ defined in \Cref{rem:set_multiset_map}, and let $\med':\dom\to\cY$ be the function which takes as input $x\in\dom$, and returns $\med(\psi(x))$. Let $\cA$ be the following algorithm for outputting an interior point of a set $x\in\widetilde{\cY}$ such that $|x|>\lambda$. On input $x$, simulate $\cW$ with query access to $\med'$ and input $\phi(x)$ and output the result. 
    
    To see why $\cA$ solves the interior point problem, observe that by the down locality and accuracy guarantees of $\cW$, we have $\cW^{\med'}(\phi(x))\in\brackets{\min\med'(\DN_\lambda(\phi(x)), \max\med'(\DN_\lambda(\phi(x))}$ with probability at least $1-\beta$. Since this is the same as the interval $\brackets{\min\med(\DN_\lambda(x)),\max\med(\DN_\lambda(x))}$, the above analysis implies that $\cW^{\med'}(\phi(x))$ is an interior point of $x$ with probability at least $1-\beta$, and hence $\cA$ solves the interior point problem with probability at least $1-\beta$ and sample complexity $\lambda+1$.  
\end{proof}

Combining \Cref{lem:ip hardness,lem: reduction from ip} yields $\lambda=\Omega\paren{\log^*(r)}$.
\end{proof}

\fi %

\section{Query Complexity Lower Bound}
\label{sec: query lower bound}

In this section, we 
\ifnum\stoc =1 state
\else prove \Cref{thm: query lower bound}, \fi
a lower bound on the query complexity of a privacy wrapper over universe $\cU=[n]$ with a weak accuracy guarantee for the class of Lipschitz functions.

\paragraph{From General Universes to the Hypercube $\{0,1\}^{n}$.}

For the remainder of this section we represent $\dom=\cP([n])$ using $\zo^n$. Each point $x\in\zo^n$ is an indicator string for the corresponding set $\{i:x_i=1\}$ in $\dom$, and the order is given by the usual subset relation $\subseteq$. 

\begin{theorem}[Query complexity with provided sensitivity bound]
    \label{thm: query lower bound}
    Fix $\constb\in(0,1)$ sufficiently small. Let $\cW$ be an $(\eps,\delta)$-privacy wrapper over $\univ=[n]$ that is $(\alpha,\beta)$-accurate for the class of Lipschitz functions $f:\dom\to[0,r]$. Suppose $\alpha<r/2$, $\eps,\beta\in(0,b)$, and $\delta\in[0,\eps^2]$. Let 
    \ifnum\stoc=1 $\eta$ denote $\frac 1\eps\log\min(\frac{r}{\alpha\beta},\frac1\delta)$ and \fi
    $q$ be the worst case expected query complexity of $\cW$.
    \begin{enumerate}
        \item\label{item: query thm 1} If \ifnum\stoc=1 $\eta \else $\frac 1\eps\log\min(\frac{r}{\alpha\beta},\frac1\delta) \fi
        \leq r\leq n^{0.49}$ then $\displaystyle q=n^{\ifnum\stoc=1 \Omega(\eta) \else\Omega\left(\tfrac1\eps\log\min\left(\tfrac{r}{\alpha\beta},\tfrac1\delta\right)\right)\fi}$.
        \item\label{item: query thm 2} If $r\leq\min(\ifnum\stoc=1 \eta\else\frac 1\eps\log\min(\frac{r}{\alpha\beta},\frac1\delta)\fi, n^{0.49})$ then $q=n^{\Omega(r)}$.
        \item\label{item: query thm 3} If $\alpha\leq\eps n$ then $q\geq \exp(\Omega(\min(\frac1\eps,\sqrt n)))$.
        
    \end{enumerate}
\end{theorem}

\ifnum\stoc =1 
The proof is deferred to the full version \cite{LinderRSS25}.
\fi

\paragraph{Query complexity vs locality bounds}
A lower bound on query complexity directly implies a lower bound on locality, since a $\lambda$-down local algorithm makes at most $\binom{|x|}{\lambda}$ distinct queries. However, the locality lower bound implied by \Cref{thm: query lower bound} 
is weaker than \Cref{thm: locality lower bound,thm:univ-dependent-lower-bound}---specifically, the locality bounds implied by \Cref{thm: query lower bound} do not capture the correct dependence on the range size~$r$.  
When
$\delta=0$, the locality lower bound given by \Cref{thm: locality lower bound} is $\frac1\eps\log\frac{r}{\alpha\beta}$, whereas the bound implied by \Cref{thm: query lower bound} is at most $\frac1\eps\log\frac{n}{\alpha\beta}$.
When $\delta>0$,  in the automated sensitivity detection setting, \Cref{thm:univ-dependent-lower-bound} implies that the locality must have at least $\log^*$ dependence on the range; in contrast, the locality lower bound implied by \Cref{thm: query lower bound} has no dependence on the range.

\begin{remark}[Tightness of our results in the automated sensitivity detection setting]
    \label{rem: query lower bound}
    Since all Lipschitz functions $f$ have $DS^f_\alpha(x)\leq \alpha$ for all datasets $x$, a privacy wrapper that is $(DS^f_\alpha(x),\beta)$-accurate for \textit{all} functions $f$ and datasets $x$ is also $(\alpha,\beta)$-accurate for all \textit{Lipschitz} functions $f$ and datasets $x$. It follows that \Cref{thm: query lower bound} also holds for privacy wrappers that are $(DS^f_\alpha(x),\beta)$-accurate for all functions $f:\dom\to\{0,1,\dots,r\}$. Recall that \Cref{thm:generalized-shifted-inverse-mechanism} gives a privacy wrapper for the automated sensitivity detection setting that has query complexity $n^{\lambda(\eps,\delta,\beta,r)}$. In the setting where $\delta=0$, \Cref{item: query thm 1} of \Cref{thm: query lower bound} implies that the query complexity of this privacy wrapper cannot be improved. In the setting of $\delta>0$, the query complexity of our privacy wrapper differs from the lower bound in \Cref{item: query thm 1} of \Cref{thm: query lower bound} by a factor of $2^{O\paren{\log^*r}}$. 
\end{remark}

\begin{remark}[Tightness our results in the claimed sensitivity bound setting]
    Recall that \Cref{thm: subset extension} gives a privacy wrapper with $(\Theta(\frac1\eps\log\frac1\beta),\beta)$-accuracy for the class of Lipschitz functions that has query complexity $n^{O(\frac1\eps\log\frac1\delta)}$. By \Cref{item: query thm 1} of \Cref{thm: query lower bound}, this query complexity is tight for the setting where $r$ is unbounded. 
    \ifnum\stoc  =1 
    The results of this section also show tight bounds on the query complexity of our privacy wrappers for the bounded range setting, which are presented and discussed in the full version.
    \else
    Moreover, in \Cref{sec: small diameter,sec:GenShI-with-nice-noise}, we give two privacy wrappers for the setting where the range of $f$ is $[0,r]$. In particular, \Cref{thm: small diameter} gives a privacy wrapper with query complexity $n^{O(r)}$, and \Cref{thm:GenShi-with-nice-noise} gives a privacy wrapper that has query complexity $n^{O\paren{\frac{1}{\eps}\log\frac{r}{\beta}}}$ with probability at least $1-\beta$. By \Cref{item: query thm 2,item: query thm 1} of \Cref{thm: query lower bound}, the query complexity of these privacy wrappers is optimal.
    \fi
\end{remark}

\ifnum\stoc = 0 %

\paragraph{Lower bounds for relaxations of our setting}
Next, we highlight some important features of \Cref{thm: query lower bound}. In particular, we explain how the lower bound also applies to privacy wrappers subject to qualitatively weaker requirements than the ones satisfied by our constructions from Sections \ref{sec:generalized-shifted-inverse}, \ref{sec:extension_based}, \ref{sec: small diameter}, and \ref{sec:GenShI-with-nice-noise}.

First, although the result is formulated for the hypercube, the lower bound applies for privacy wrappers over any set $\cU$ of size at least $n$, by fixing an arbitrary subset of $\univ'\in\univ$ with $|\univ'|=n$, and considering functions $f:(\univ')^*\to\R$.

Second, all of our wrappers make queries only on large subsets of the input dataset (``down local'' in \Cref{def: privacy wrapper}). The lower bound applies to wrappers that can query $f$ at any dataset contained in $[n]$, regardless of its size or relation to the input set. And additionally, the lower bound also applies to privacy wrappers that are only guaranteed  to be accurate on functions that are Lipschitz on the entire domain. In other words, the difficulty of building a privacy wrapper is not due to locality, per se. Our ``hard instances'', defined below, show that the challenge lies in finding regions where the Lipschitz constraint might be violated. 

Third, \Cref{item: query thm 3} of \Cref{thm: query lower bound} demonstrates a lower bound for privacy wrappers that have considerably worse accuracy guarantees than our constructions. In particular, it states that any privacy wrapper that has a very weak accuracy guarantee ($\alpha\approx\eps n$) requires $\exp\paren{\Omega\paren{1/\eps}}$ queries. \\

\Cref{thm: query lower bound} follows from the following more detailed lemma. It relates the minimum query complexity of any privacy wrapper to the size of the dataset size and the desired privacy and accuracy parameters.

\begin{lemma}[Detailed query complexity lower bound]
    \label{lem: query lower bound}
    There exist constants $\consta\in\N$ and $\constb\in(0,1)$, such that for all sufficiently large $n\in\N$, all $\eps,\beta\in(0,\constb]$, $\delta\in[0,\eps\constb)$, $\rho,\alpha\in\N$ such that $\rho\in[\consta, \constb n]$ and $\alpha<\rho/2$, if $\cW$ is an $(\eps,\delta)$-privacy wrapper over $\univ=[n]$ that is $(\alpha,\beta)$-accurate for the class of Lipschitz functions $f:\dom\to[0,\rho]$, then there exists a function $f:\dom\to[0,\rho]$ and a dataset $x\in\dom$ such that $\cW^f(x)$ has expected query complexity $\left(\frac{n}{\rho\kappa}\right)^{\Omega(\kappa)}$, where $\kappa=\min\left(\rho,\frac{n}{2\rho},\frac 1\eps\log\min\left(\frac{\rho}{\alpha\beta},\frac\eps\delta\right)\right)$. 
\end{lemma}

The choice of $\kappa$ in the lemma statement ensures that the base of the exponent in the query lower bound, $\frac{n}{\rho\kappa}$, is always at least 2.

\begin{proof}[Proof of \Cref{thm: query lower bound}]
    Plugging in parameters to \Cref{lem: query lower bound}, we prove each item of \Cref{thm: query lower bound}. To prove \Cref{item: query thm 1} we set $\rho=r$, then, since $\rho\leq n^{0.49}$ we see that $\kappa=\frac1\eps\log\min(\frac{\rho}{\alpha\beta},\frac\eps\delta)$, and thus $\frac{n}{\rho\kappa}=n^{\Theta(1)}$, which completes the proof of \Cref{item: query thm 1}. To prove \Cref{item: query thm 2} we set $\rho=r$, then, since $r\leq\min(\frac 1\eps\log\min(\frac{r}{\alpha\beta},\frac1\delta), n^{0.49})$, we see that $\kappa=\rho$ and $\frac{n}{\rho\kappa}=n^{\Theta(1)}$, which completes the proof of \Cref{item: query thm 2}. Last, to prove \Cref{item: query thm 3} we set $\alpha=\eps n$ and $\rho=3\eps n$, observe that this implies $\kappa=\frac{n}{2\rho}=\Omega(\frac1\eps)$, and $\frac{n}{\rho\kappa}=2$, completing the proof of \Cref{item: query thm 3}.
\end{proof}

\paragraph{Construction of Hard Distributions} We prove \Cref{lem: query lower bound} by constructing a pair of distributions that cannot be distinguished by any query-efficient algorithm but can be distinguished using $\cW$. Let $\Delta(x,y)$ denote the Hamming distance between $x$ and $y$.

\begin{definition}[Functions $f^k_x$ and $F^{k,s}_{x,y}$. Distributions $\cN$, $\cP$, and $\cD$]
    \label{def: distributions} Fix $\Gamma,\rho,n\in\N$. %
    For all $x\in\zo^n$ and $k\in[0,\rho]$ define $f^k_x:\zo^n\to[0,\rho]$ by
    \[
     f^k_{x}(z)=\max(k-\Delta(x,z),0).
    \]
    Additionally, for all $x,y\in\dom$ and $k,s\in[0,\rho]$ define $F^{k,s}_{x,y}:\zo^n\to[0,\rho]$ by 
    \[ F^{k,s}_{x,y}(z) = \begin{cases} 
          f^{k}_x(z) & \Delta(x,z)<\Delta(y,z) \\
          f^{s}_y(z) & \Delta(y,z)<\Delta(x,z).
       \end{cases}
    \]
    For all $\rho,\Gamma,n\in\N$ such that $\Gamma$ is odd and $\Gamma\leq \min(\rho,n)$, and all $\alpha>0$ such that that $2\alpha$ divides $\rho$, let $\cN[\alpha,\rho,\Gamma,n]$ and $\cP[\alpha,\rho,\Gamma,n]$ be distributions over $(x,f)$ where $f:\zo^n\to[0,\rho]$ and $x\in\zo^n$ are obtained by the following sampling procedure: 
        \begin{enumerate}
            \item Sample $x\sim\zo^n$ and $y\sim\{z:\Delta(x,z)=\Gamma\}$ uniformly at random.
            \item Sample $k,s\sim\{2\alpha,4\alpha,6\alpha,\dots,\rho\}$ uniformly without replacement.
            \item If sampling from $\cN$, return $(x,F^{k,s}_{x,y})$. If sampling from $\cP$, return $(x,f^{k}_x)$.
        \end{enumerate}
    We omit the parameters $\alpha,\rho,\Gamma,n$ when they are clear from context.
\end{definition}

\begin{figure}[H]
    \includegraphics[scale=.4]{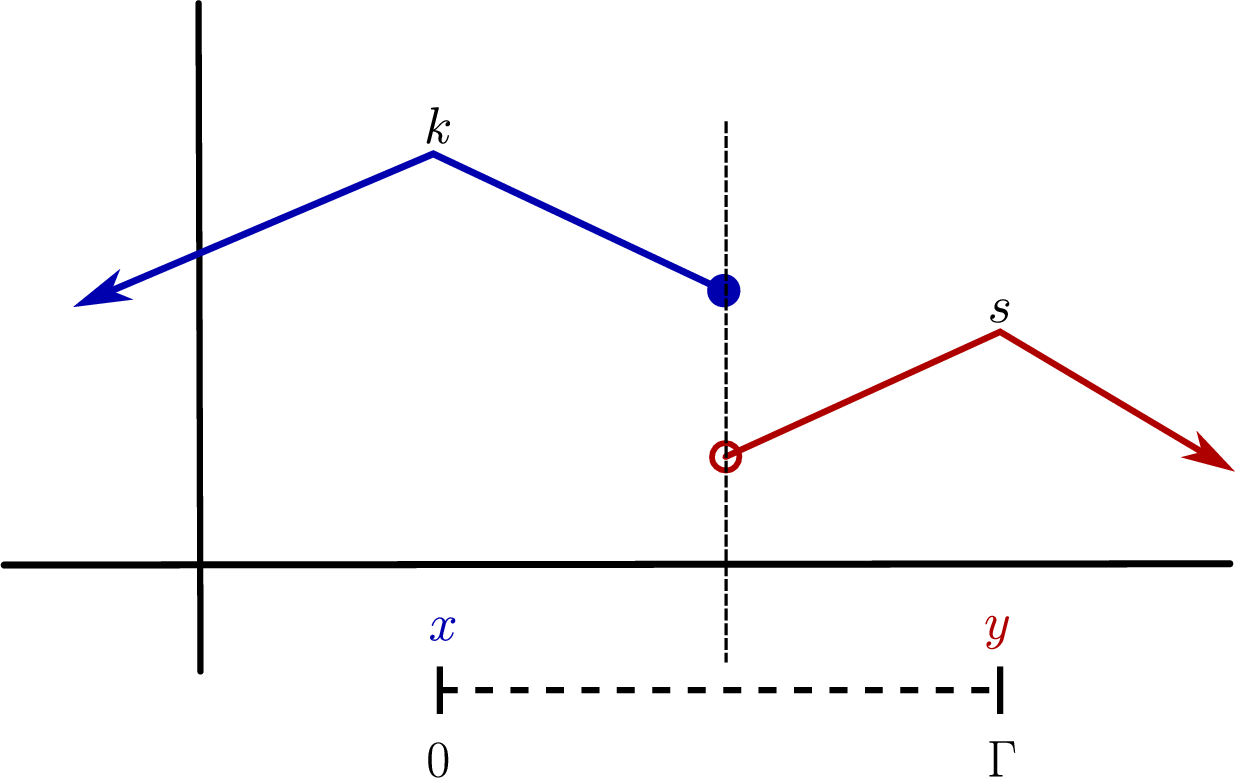}
    \hfill
    \includegraphics[scale=.4]{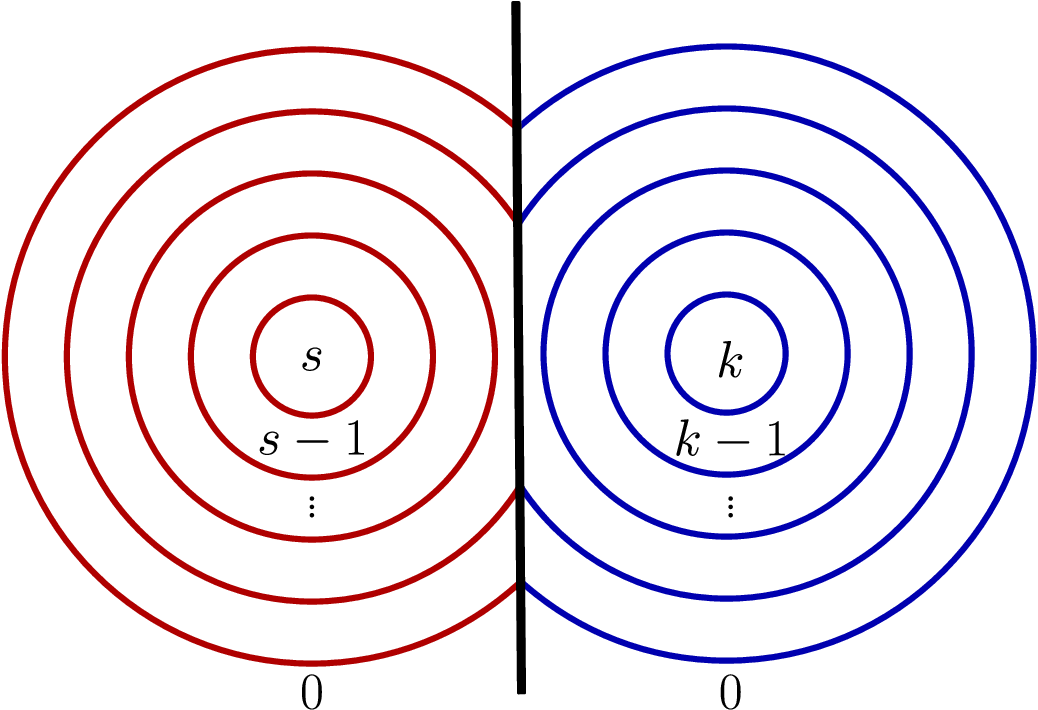}
    \caption{Function $F^{k,s}_{x,y}$ side-view (left) and top-view (right). Observe that $F^{k,s}_{x,y}$ contains a large ``jump" between the ball around $x$ and the ball around $y$.}\label{fig:hard_function}
\end{figure}

\begin{remark}
    The distance $\Delta(x,y)=\Gamma$ is chosen to be odd so that for all $z\in\zo^n$ we have $\Delta(x,z)\neq\Delta(y,z)$. This ensures that $F^{k,s}_{x,y}=F^{s,k}_{y,x}$---that is, the distribution $\cN$  is symmetric with respect to $x$ and $y$. 
\end{remark} 

One can visualize the graph of $f^k_x$ as an inverted cone of height $k$ centered at $x$, and the graph of $F^{k,s}_{x,y}$ as two inverted cones of height $k$ and $s$, centered at $x$ and $y$ respectively. Intuitively, functions $f^k_x$ are always Lipschitz, while functions $F^{k,s}_{x,y}$ are guaranteed to be non-Lipschitz whenever $\max(k,s)>\Delta(x,y)/2$.

At a high level, our privacy wrappers distinguish $\cN$ from $\cP$ as follows: Functions sampled from $\cN$ are sometimes non-Lipschitz, while functions sampled from $\cP$ are always Lipschitz. Thus, every privacy wrapper $\cW$ that is $(\alpha,\beta)$-accurate on Lipschitz functions satisfies $|\cW^f(x)-f(x)|\leq \alpha$ with probability at least $1-\beta$ whenever $(x,f)\sim\cP$. However, when $(x,f)\sim\cN$ the privacy wrapper satisfies $|\cW^f(x)-f(x)|\leq \alpha$ with probability much less than $1-\beta$ (over the randomness of both the privacy wrapper and the distribution $\cN$). Hence, such a privacy wrapper can be used to distinguish $\cN$ from $\cP$.   

\subsection{Proof of Query Complexity Lower Bound (\Cref{lem: query lower bound})}

The main steps in the lower bound proof are \Cref{lem: indisinguishable}, which relates the query complexity of an algorithm to its advantage in distinguishing $\cN$ and $\cP$, and \Cref{lem: inaccuracy}, which upper bounds the probability that a privacy wrapper outputs an accurate answer on inputs sampled from $\cN$. Our proof proceeds by contradiction: we show that if $q\leq \left(\frac{n}{\rho\kappa}\right)^{O(\kappa)}$, then one can construct a distinguisher which violates the advantage bound given by \Cref{lem: indisinguishable}.

\begin{lemma}[Indistinguishability]
    \label{lem: indisinguishable}
    Let $\constb,c\in(0,1)$ be sufficiently small constants, and let $\cT$ be a $q$-query randomized algorithm that takes as input parameters $\alpha,\rho,\Gamma,n$, a point $x\in\zo^n$, and query access to a function $f:\zo^n\to[0,\rho]$. Suppose that $n$ is sufficiently large, that $\Gamma\leq \rho\leq \constb n$, and that $\Gamma$ is odd. If $q\leq \bparen{\frac {n}{8\rho\Gamma}}^{\Gamma/4}$ then, 
    \[
    \left|\Pr_{\substack{(x,f)\sim\cP \\ \text{coins of } \cT}}\left[\cT^f(x)=1\right]-\Pr_{\substack{(x,f)\sim\cN \\ \text{coins of } \cT}}\left[\cT^f(x)=1\right]\right|\leq \bparen{\frac{n}{\rho\Gamma}}^{-c\Gamma}.
    \]
\end{lemma}

We defer the proof of \Cref{lem: indisinguishable} to \Cref{sec: indistinguishable} and complete the proof of the lower bound. Let $\cW$ be a privacy wrapper with the parameters of \Cref{lem: query lower bound}. By the $(\alpha,\beta)$-accuracy guarantee
\begin{equation}
    \label{eq:yes-case}
    \Pr_{\substack{(x,f)\sim\cP \\ \text{coins of } \cW}}\left[\left|\cW^f(x)-f(x)\right|\leq \alpha\right]\geq 1-\beta.
\end{equation}

To see what happens when $(x,f)\sim\cN$ we will use the following statement, the proof of which appears in \Cref{sec: inaccuracy}. Informally, \Cref{lem: inaccuracy} bounds the probability that a privacy wrapper $\cW$ outputs an answer that is within $\alpha$ of $f(x)$ when $f\sim\cN$. In particular, the lemma provides a bound in terms of the following two regimes: The first regime handles the case when $\beta$ is large relative to $\frac{\alpha}{\rho}$, while the second regime handles the case where $\beta$ is small relative to $\frac{\alpha}{\rho}$. Intuitively, the bound in the first regime is obtained by taking advantage of the random choice of values $k$ and $s$ in the definition of the hard distributions (\Cref{def: distributions}), while the bound in the second regime is obtained via the $(\alpha,\beta)$-accuracy guarantee of the privacy wrapper. 
    
    \begin{lemma}[Inaccuracy]
        \label{lem: inaccuracy}
        Let $\cW$ be a $q$-query $(\eps,\delta)$-privacy wrapper over $[n]$ that takes as input parameters $\rho,\Gamma,n$. %
        Fix $\alpha,\beta>0$, %
        and suppose that $\cW$ is $(\alpha,\beta)$-accurate for the class of Lipschitz functions. If $\Gamma\in\N$ is odd, $\Gamma\leq \min(\rho,n)$, $\alpha<\frac\rho2$, 
        and $2\alpha$ divides $\rho$, then 
        \[
        \Pr_{\substack{(x,f)\sim\cN \\ \text{coins of } \cW}}\left[\left|\cW^f(x)-f(x)\right|\leq\alpha\right]\leq\min\left( e^{\Gamma\eps}\left(\frac{3\alpha}{\rho}+q\cdot\left(\frac{8\rho\Gamma}{n}\right)^\frac{\Gamma}{2} +\frac\delta\eps\right), 1-\frac{1-\delta e^{\Gamma\eps}/\eps}{4(1+e^{\Gamma\eps})}\right).
        \]
    \end{lemma}

    Next, we will construct a ``distinguisher" $\cT$ for inputs sampled from $\cP$ and $\cN$. Let $\cT$ be the algorithm which calls $\cW$ as a subroutine and distinguishes $\cN$ from $\cP$. The algorithm $\cT$ gets as input, parameters $\alpha,\rho,\Gamma,n\in\N$, a point $x\in\zo^n$, and query access to a function $f:\zo^n\to[0,\rho]$. 
    To distinguish the distributions, $\cT$ runs $\cW^f(x)$ and, if $|\cW^f(x)-f(x)|\leq \alpha$ then $\cT$ outputs $1$; otherwise, $\cT$ outputs $0$. 
    
    The remainder of the proof shows that, when $q$ is small, the distinguisher $\cT$ has advantage better than the bound in \Cref{lem: indisinguishable}, yielding a contradiction. (The reader uninterested in these calculations may want to skip to the next section.)

    Fix sufficiently small constants $\constb,\constc\in(0,1)$, where $b\ll c$, and let $n$ be sufficiently large. Set $\eps,\beta\in(0,\constb]$, $\delta\in[0,\eps\constb)$, set range $\rho\in[(4/\constc)\log(1/\constb),\constb\constc n/(4\log(1/\constb)]$, $\alpha<\frac\rho2$, and set quantity $\kappa=\min\left(\rho,\frac {\constb n}{\rho},\frac 1\eps\ln\left(\constb\min\left(\frac{\rho}{\alpha\beta},\frac\eps\delta\right)\right)\right)$. Further suppose that $2\alpha$ divides $\rho$. While this setting of parameters roughly corresponds to those in the statement of \Cref{lem: query lower bound}, it is convenient for our analysis to replace the term $\frac{\rho}{\alpha\beta}$ in the $\ln$ with $\max(\frac\rho\alpha,\frac1\beta)$. This will facilitate a case analysis on $\beta\leq \frac\alpha\rho$, and $\beta\geq\frac\alpha\rho$. Set
    \[\textstyle
    \Gamma^*=\min\paren{
        \rho, \frac{\constb n}{\rho}, \frac1\eps\ln\Bparen{\constb\min\bparen{\frac\eps\delta, \max\bparen{\frac\rho\alpha, \frac1\beta}}}
        }.
    \] 
    Observe that $\ln(\max(\frac\rho\alpha,\frac1\beta))=\frac12\ln(\max(\frac\rho\alpha,\frac1\beta)^2)\geq \frac12\ln\frac{\rho}{\alpha\beta}$, and $\ln(\frac{\rho}{\alpha\beta})\geq \ln\max((\frac{\rho}{\alpha},\frac1\beta))$. Hence, $\frac\kappa2\leq \Gamma^*\leq \kappa$, so we can prove a lower bound of  $\left(\frac{n}{\rho\kappa}\right)^{\Omega(\kappa)}$ by proving a lower bound of $\left(\frac{n}{\rho\Gamma^*}\right)^{\Omega(\Gamma^*)}$. 
    
    In the remainder of the proof, we set $\Gamma$ to the largest odd integer that is at most $\Gamma^*$ (and thus $\Gamma=\Theta(\Gamma^*)$). Suppose for the sake of contradiction that $q\leq \left(\frac{n}{8\rho\Gamma}\right)^{\Gamma/4}$. We analyze $\cT$ by considering two cases. In each case, we will show that the upper bound on advantage implied by \Cref{lem: indisinguishable} is contradicted by distinguisher $\cT$. Since $\Gamma\leq \Gamma^*\leq\min(\rho,\frac {\constb n}{\rho})\leq\min(\rho,n)$, we can apply \Cref{lem: inaccuracy} in both cases.

    \paragraph{Case 1:} We first consider the case where $\beta\geq\frac\alpha\rho$. To analyze this case, we use the inequality 
    \[
    \Pr_{(x,f)\sim\cN}\left[\left|\cW^f(x)-f(x)\right|\leq\alpha\right]\leq e^{\Gamma\eps}\left(\frac{3\alpha}{\rho}+q\cdot\left(\frac{8\rho\Gamma}{n}\right)^\frac{\Gamma}{2} +\frac\delta\eps\right)
    \]
    given by \Cref{lem: inaccuracy}. Next, we combine \eqref{eq:yes-case} and \Cref{lem: indisinguishable}, to obtain
    \[
    1-\beta-e^{\Gamma\eps}\left(\frac{3\alpha}{\rho}+q\cdot\left(\frac{8\rho\Gamma}{n}\right)^\frac{\Gamma}{2} +\frac\delta\eps\right)\leq\Pr_{(x,f)\sim\cP}\left[\cT^f(x)=1\right]-\Pr_{(x,f)\sim\cN}\left[\cT^f(x)=1\right]\leq \left(\frac{n}{\rho\Gamma}\right)^{-c\Gamma}.
    \]
    Notice that the left hand side of the expression is at least $1-\beta-e^{\Gamma\eps}\left(\frac{3\alpha}{\rho}+\frac\delta\eps\right)$. By hypothesis, $\Gamma\leq \frac1\eps\ln(\constb\min(\frac{\rho}{\alpha},\frac\eps\delta))$, and therefore $e^{\Gamma\eps}\left(\frac{3\alpha}{\rho}+\frac\delta\eps\right)\leq 4\constb$. Thus, we obtain the inequality $1-5\constb\leq1-\beta-4\constb\leq\left(\frac{n}{\rho\Gamma}\right)^{-c\Gamma}\leq 2^{-c\Gamma}\leq 2^{-c}$. This is a contradiction, since for $\constb$ sufficiently small, the left hand side approaches $1$, while the right hand side is a fixed constant that is strictly less than $1$. It follows that $q\geq \left(\frac{n}{\rho\Gamma}\right)^{\Omega(\Gamma)}=\left(\frac{n}{\rho\kappa}\right)^{\Omega(\kappa)}$, proving the lemma for Case 1.

    \paragraph{Case 2:} Next, we consider the case where $\beta\leq\frac\alpha\rho$ (and again using the parameter settings stated before Case 1). To analyze this case, we use the inequality
    \[
        \Pr_{(x,f)\sim\cN}\left[\left|\cW^f(x)-f(x)\right|\leq\alpha\right]\leq1-\frac{1-(\delta e^{\Gamma\eps}/\eps)}{4(1+e^{\Gamma\eps})}.
    \]
    given by \Cref{lem: inaccuracy}. Proceeding as in the previous case, we obtain the inequality
    \[
    1-\beta - \left(1-\frac{1-(\delta e^{\Gamma\eps}/\eps)}{4(1+e^{\Gamma\eps})} \right)\leq \left(\frac{n}{\rho\Gamma}\right)^{-c\Gamma}. 
    \]
   Manipulating terms, we see that that the left hand side of the expression is equal to $\frac{1-(\delta e^{\Gamma\eps}/\eps)}{4(1+e^{\Gamma\eps})}-\beta$. Observe that if $\frac{1-(\delta e^{\Gamma\eps}/\eps)}{4(1+e^{\Gamma\eps})}\geq2\beta$, then the left hand side of the expression is at least $\beta$. Rearranging terms, we see that the left hand side is at least $\beta$ whenever $e^{\Gamma\eps}\leq \frac{1-8\beta}{8\beta+(\delta/\eps)}$. Setting $\constb$ sufficiently small, and $\Gamma\leq\frac1\eps\ln(\constb\min(\frac1\beta,\frac\eps\delta))$, we get $e^{\Gamma\eps}\leq \frac{1-8\beta}{8\beta+(\delta/\eps)}$, and hence, the left hand side of the inequality is indeed at least $\beta$. Putting together the above calculations, we obtain the inequality $\beta\leq \left(\frac{n}{\rho\Gamma}\right)^{-c\Gamma}$. Now, since $ \Gamma\leq\frac{\constb n}{\rho}\leq\frac{n}{2\rho}$, we have $\left(\frac{n}{\rho\Gamma}\right)^{-c\Gamma}\leq 2^{-c\Gamma}<2^{-c\Gamma^*/2}$. 
   
   To obtain a contradiction, observe that for all $\delta'\geq\delta$, and $\beta'\geq\beta$, every $(\eps,\delta)$-privacy wrapper with $(\alpha,\beta)$-accuracy is also an $(\eps,\delta')$-privacy wrapper with $(\alpha,\beta')$-accuracy. Thus, set $\delta'\gets\max(\delta, \eps\constb\cdot e^{-\constb\min(\rho,n\constb/\rho)})$, and set $\beta'\gets \max(\beta,\constb\cdot e^{-\constb\min(\rho,n\constb/\rho)})$. If $\beta'\geq \alpha/\rho$ then by case 1 we obtain $q\geq \paren{\frac{n}{\rho\kappa}}^{\Omega(\kappa)}$. On the other hand, if $\beta'<\alpha/\rho$, then the new value of $\Gamma^*$ is $\min(\rho,\frac{\constb n}{\rho})$, and by the analysis in case 2, we have $\beta'< 2^{-c\Gamma^*/2}$. By our setting of $\beta'$, this implies that $\constb\cdot e^{-\constb\min(\rho,n\constb/\rho)}<2^{-c\Gamma^*/2}=2^{-c\min(\rho,n\constb/\rho)/2}$, and by our choice of $b\ll c$, that $b<2^{-c\min(\rho,n\constb/\rho)/4}$. Finally, by our setting of $\rho$ we have that $\min(\rho,\frac{n\constb}{\rho})\geq\frac4\constc\log\frac1\constb$, and thus we obtain the contradiction of $b<b$. It follows that $q\geq \paren{\frac{n}{\rho\kappa}}^{\Omega(\kappa)}$.

    Applying Yao's principle \cite{Yao77} to the uniform mixture of $\cN$ and $\cP$ suffices to complete the proof. \qed

\subsection{Proof of Indisinguishability (\Cref{lem: indisinguishable})}
\label{sec: indistinguishable}

In order to show that $\cN$ and $\cP$ are hard to distinguish we bound the statistical distance between the view of any algorithm when its inputs are sampled from $\cN$, from the view of the algorithm when its inputs are sampled from $\cP$. Below, we define a notion of ``revealing point" such that if a point $z$ is not ``revealing" for $(x,y)$, then $F^{k,s}_{x,y}(z)=f_x^k(z)$. Hence, it suffices to demonstrate that the probability an algorithm queries a revealing point is small.

\begin{definition}[Bad event $B{[\cT,f,(x,y)]}$, revealing point]
    \label{def: B}
    For all $\Gamma,\rho,n\in\N$, and $x,y\in\zo^n$ such that $\Delta(x,y)=\Gamma$, a point $z\in\zo^n$ is \emph{a revealing point for $(x,y)$} if $\Delta(y,z)<\min(\Delta(x,z), \rho)$. 
    
    Let $\cT$ be a $q$-query algorithm that gets as input a point $x\in\zo^n$, and query access to a function $f:\zo^n\to\R$. Let $B[\cT,f,(x,y)]$ be the event that $\cT^f(x)$ queries a revealing point for $(x,y)$. 
\end{definition}

Next, we bound the probability of $B[\cT,f,(x,y)]$ when $(x,f)\sim\cN$ in terms of $\rho,\Gamma$, and $n$---that is, the probability that algorithm $\cT$ queries a revealing point is small.

\begin{claim}[$B_{\cT}$ bound]
    \label{claim: B bound}
    For all $\Gamma,\rho,n,q\in\N$ such that $\Gamma\leq\min(n,\rho)$, and $\Gamma$ is odd, and every $q$-query algorithm $\cT$, 
    \[
    \Pr_{\substack{(x,F^{k,s}_{x,y})\sim\cN \\ \text{coins of $\cT$}}}[B[\cT,F^{k,s}_{x,y}, (x,y)]\leq q\cdot \left(\frac{8\rho\Gamma}{n}\right)^{\frac\Gamma2}.
    \]
\end{claim}
\begin{proof}
   Without loss of generality, we prove the statement for deterministic algorithms. Fix a query $z$ made by $\cT$ and recall that by definition of $\cN$ (\Cref{def: distributions}), we have $\Delta(x,y)=\Gamma$. Observe that if $\Delta(x,z)<\frac{\Gamma}{2}$ then $z$ cannot be revealing since then $\Delta(y,z)\geq\Delta(y,x)-\Delta(x,z)>\frac{\Gamma}{2}>\Delta(x,z)$. Similarly, if $\Delta(x,z)>2\rho$ then $z$ cannot be revealing since then $\Delta(y,z)\geq\Delta(x,z)-\Delta(x,y)>2\rho-\Gamma\geq\rho$ (since $\Gamma\leq\rho$). Thus, if $z$ is revealing then $\frac{\Gamma}{2}\leq \Delta(x,z)\leq 2\rho.$ 
   
    We argue that over the randomness of $y$, if the query $z$ satisfies $\frac{\gap}{2}\leq \Delta(x,z)\leq2\rho$, then it is very unlikely that $y$ satisfies $\Delta(y,z)\leq \Delta(x,z)$. Let $A\subset[n]$ be the set of indices on which $x$ and $z$ agree, and let $\overline A$ denote $[n]\setminus A$. Consider sampling a point $y$ by choosing a set $S\subset [n]$ of $\gap$ indices uniformly and independently at random and flipping the bit $s_i$ for each $i\in S$. Let $m=|S\cap A|$. Since $\Delta(x,z)\geq \frac \gap2$, the point $y$ satisfies $\Delta(y,z)\leq\Delta(x,z)$ if and only if $m\leq\frac \gap2$. Moreover, since $|\overline A|=\Delta(x,z)\leq2\rho$, there are at most $\binom{n}{m}\binom{2\rho}{\gap-m}$ points $y$ that can be obtained by flipping the bits of $z$ at $m$ indices in $A$ and $\gap-m$ indices in $\overline A$. By summing over each value of $m$ we obtain the bound
    \[
    \Pr_y[\Delta(y,z)\leq\Delta(x,z)]\leq \gap\cdot\max_{0\leq m\leq\frac \gap2} \binom{n}{m}\binom{2\rho}{\gap-m}\binom{n}{\gap}^{-1}
    \]
    Next, we use the inequality $\binom{n}{m}\leq(\frac{ne}{m})^m$ to get
    \[
    \gap\binom{2\rho}{\gap-m}\binom{n}{m}\leq \gap(\frac{4\rho}{\gap})^{\gap-m}n^m\leq(\frac{8\rho\cdot n}{\gap})^{\frac\gap2},
    \]
    where the second inequality follows since $\gap-m\geq\frac\gap2$ and $\gap\leq2^{\frac\gap 2}$. Using the inequality $\binom{n}{\gap}^{-1}\leq(\gap/n)^\gap$ we obtain the following bound
    \[
    \Pr_y[\Delta(y,z)\leq\Delta(x,z)]\leq \Big(\frac{8\rho\cdot n}{\gap}\Big)^{\frac\gap2}\Big(\frac\gap n\Big)^{\gap}\leq\Big(\frac{8\rho\cdot \gap}{n}\Big)^{\frac\gap2}.
    \]
    A union bound over the $q$ queries now suffices to complete the proof.
\end{proof}

In order to complete the proof of \Cref{lem: indisinguishable}, we introduce the following standard material.

\begin{definition}[$D$-view]
    \label{def: D-view}
    For all $q$-query deterministic algorithms $\cA$, and all distributions $D$ over inputs to $\cA$, let $D$-view denote the distribution over query answers $a_1,...,a_q$ given to $\cA$ when the input is sampled according to $D$.
\end{definition}

\begin{definition}[Statistical distance]
For distributions $D$ and $D_0$ over a set $S$, the statistical distance between $D$ and $D_0$ is 
$$SD(D,D_0)=\max_{T\subset S}(|\Pr_{D}[x\in T]-\Pr_{D_0}[x\in T]|).$$
Additionally, for all $\delta>0$, let $D\approx_{\delta} D_0$ denote that the statistical distance between $D$ and $D_0$ is at most $\delta$.
\end{definition}

\begin{fact}[Claim 4 \cite{RS06}]
\label{fact: SD bound}
    Let $E$ be an event that happens with probability at least $1-\delta$, for some $\delta\in(0,1)$, under the distribution $D$ and let $D|_{E}$ denote $D$ conditioned on event $E$. Then, $D|_{E}\approx_{\delta'} D$ where $\delta'=\frac{\delta}{1-\delta}$.
\end{fact}

By Definitions~\ref{def: distributions} and \ref{def: B}, we have $\cN|_{\overline{B_{\cT}}}$-view$=\cP$-view. Hence, by \Cref{fact: SD bound} instantiated with $D=\cN$ and $E=\overline{B_{\cT}}$, 
\[
\cNv\approx_{\delta'}\cN|_{\overline{B_{\cT}}}\text{-view}=\cPv,
\] 
where $\delta'=\frac{\delta}{1-\delta}$ and $\delta$ is the bound given in \Cref{claim: B bound}. Since a deterministic algorithm can be viewed as a distribution over randomized algorithms, we without loss of generality consider a $q$-query deterministic algorithm $\cT$. Let $A$ be the set of query answers on which $\cT$ outputs $1$. By standard arguments,
\[
\begin{split}
    \left|\Pr_{(x,f)\sim\cN}[\cT^f(x)=1]-\Pr_{(x,f)\sim\cP}[\cT^f(x)=1]\right|
    &=\left|\Pr_{a\sim \cNv}[a\in A]-\Pr_{a\sim \cPv}[a\in A]\right|\\
    &\leq SD(\cNv, \cPv)\leq \frac{q\cdot \left(\frac{8\rho\Gamma}{n}\right)^{\frac\Gamma2}}{\left(1-q\cdot \left(\frac{8\rho\Gamma}{n}\right)^{\frac\Gamma2}\right)}.
\end{split}
\]
By our choice of $q\leq \bparen{\frac{n}{8\rho\Gamma}}^{\Gamma/4}$ and $\Gamma\leq\rho\leq \constb n$ for a sufficiently small constant $\constb\in(0,1)$, the right hand side is at most $\bparen{\frac{n}{\rho\Gamma}}^{-c\Gamma}$ for some universal constant $c\in(0,1)$.  \qed

\subsection{Proof of Inaccuracy (\Cref{lem: inaccuracy})}
\label{sec: inaccuracy}
    Our arguments rely on the following standard group privacy claim.
    \begin{claim}[Group Privacy]
        \label{claim: group privacy}
        Fix $\eps\in(0,1]$ and $\delta\in[0,1]$. Suppose $\cW$ is $(\eps,\delta)$-DP and let $E\subset\R$ be measurable. If $x,y\in\zo^n$ then,
        \[
        \Pr_\cW[\cW(x)\in E]\leq e^{\Delta(x,y)\eps}\left(\Pr_\cW[\cW(y)\in E]+\frac\delta\eps\right).
        \]
    \end{claim}
    \begin{proof}
        By $(\eps,\delta)$-DP, 
        \[
        \Pr_\cW[\cW(x)\in E]\leq e^{\Delta(x,y)\eps}\Pr_\cW[\cW(y)\in E]+\sum_{i=0}^{\Delta(x,y)-1}e^{i\cdot\eps}\delta.
        \]
        The series on the right hand side is geometric and can be bounded above by $e^{\Delta(x,y)\eps}\frac{\delta }{e^{\eps}-1}$. We use the inequality $1+z\leq e^z$ with $z$ set to $\eps$ and then factor out $e^{\Delta(x,y)\eps}$ to complete the proof. 
    \end{proof}

    We prove the two upper bounds given by the \Cref{lem: inaccuracy} separately. First, we prove the upper bound of $e^{\Gamma\eps}\left(\frac{2\alpha}{\rho}+q\cdot\left(\frac{8\rho\Gamma}{n}\right)^\frac{\Gamma}{2} +\frac\delta\eps\right)$.
    \begin{proof}[Proof of first bound]
        We begin by bounding the quantity of interest using group privacy. Since $\Delta(x,y)=\Gamma$,
        \[
        \Pr_{\substack{(x,f)\sim\cN \\ \cW}}\left[\left|\cW^f({\color{blue} x})-f(x)\right|\leq\alpha\right]\leq e^{\Gamma\eps}\left(\Pr_{\substack{(x,f)\sim\cN \\ \cW}}\left[\left|\cW^f({\color{blue} y})-f(x)\right|\leq\alpha\right] +\frac\delta\eps\right),
        \]
        where the events whose probabilities are given  on the left and right differ by the input to $\cW$: a private algorithm must give similar answers on $x$ and $y$.

        Next, we use the definition of event $B[\cW, f, (x,y)]$ from \Cref{def: B}. Notice that if $(x,F^{k,s}_{x,y})\sim\cN$, then, conditioned on the event $\overline{B[\cW, F^{k,s}_{x,y}, (x,y)]}$ (that is, no revealing points are observed), the distribution of $s=F^{k,s}_{x,y}(y)$, is uniform over $\{2\alpha,4\alpha,6\alpha,\dots,\rho\}\setminus\{k\}$. 
        Moreover, since $F^{k,s}_{x,y}=F^{s,k}_{y,x}$, the tuples $(x, F^{k,s}_{x,y})$ and $(y, F^{k,s}_{x,y})$ are identically distributed. 
        Thus, conditioned on the event $\overline{B[\cW, F^{k,s}_{x,y}, (y,x)]}$, the distribution of $k=F^{k,s}_{x,y}(x)$ is uniform over $\{\alpha,4\alpha,6\alpha,\dots,\rho\}\setminus\{s\}$, and hence, the probability that $|k-\cW^f(y)|\leq \alpha$ is at most $\frac{3\alpha}{\rho}$.  %
        Let $B$ denote the event $B[\cW, F^{k,s}_{x,y}, (y,x)]$, then
        \begin{align*}
        \Pr_{\substack{(x,f)\sim\cN \\ \cW}}\left[\left|\cW^f(y)-f(x)\right|\leq\alpha\right]
        &\leq 
        \Pr_{\substack{(x,f)\sim\cN \\ \cW}}\left[\left|\cW^f(y)-f(x)\right|\leq\alpha \Big|\overline{B} \right]
        +\Pr_{\substack{(x,f)\sim\cN \\ \cW}}\left[B\right]. \\
        &\leq \frac{3\alpha}{\rho}+q\cdot\left(\frac{8\rho}{n}\right)^\frac\Gamma2
        \end{align*}
        Where the bound on the second term follows from \Cref{claim: B bound}. Putting it all together, see that 
        \[
        \Pr_{\substack{(x,f)\sim\cN \\ \cW}}\left[\left|\cW^f(x)-f(x)\right|\leq\alpha\right]\leq e^{\Gamma\eps}\left( \frac{3\alpha}{\rho}+q\cdot\left(\frac{8\rho\Gamma}{n}\right)^{\frac\Gamma2}+\frac\delta\eps\right).
        \qedhere \]
    \end{proof}

    Next, we prove the upper bound of $1-\frac{1-(\delta e^{\Gamma\eps}/\eps)}{4(1+e^{\Gamma\eps})}$. The proof of this bound takes advantage of the symmetry of $F^{k,s}_{x,y}$ as well as the $(\alpha,\beta)$-accuracy of $\cW$.

    \begin{proof}[Proof of \Cref{lem: inaccuracy} (second bound)]
    For all $f:\zo^n\to\R$, a point $u$ is \emph{$\gamma$-distinguishing for $\cW$ on $f$} if $\Pr[|\cW^f(u)-f(u)|\geq\alpha]\geq\gamma$. We will make use of the following claim. 

    \begin{claim}
    \label{claim:distinguishing}
    If $\gamma<\frac{1-(\delta e^{\Gamma\eps}/\eps)}{1+e^{\Gamma\eps}}$, then at least one of $x$ or $y$ is $\gamma$-distinguishing for $\cW$ on $F^{k,s}_{x,y}$ (\Cref{def: distributions}). 
    \end{claim}

    We defer the proof of \Cref{claim:distinguishing} and use it to complete the proof of \Cref{lem: inaccuracy}. The essence of the argument is the symmetry of $x$ and $y$ in the generation of pairs from  $\cN$. The key observation is that, when $x,y,s$ are distributed as in \Cref{def: distributions}, the tuples $(x,y,F^{k,s}_{x,y})$ and $(y,x,F^{k,s}_{x,y})$ are identically distributed.

    To see why this is, observe that for every fixed $x,y$, the functions $F^{k,s}_{x,y}$ and $F^{s,k}_{y,x}$ are the same. When $x,y$ are generated randomly as in \Cref{def: distributions}, their distribution is symmetric---the pair $(x,y)$ is identically distributed to $(y,x)$. Similarly, since $k,s$ are generated uniformly at random and independent of $x$ and $y$, the pair $(s,k)$ is identically distributed to the pair $(k,s)$. This means that the tuple  $(x,y,F^{k,s}_{x,y})$ is identically distributed to $(y,x,F^{k,s}_{x,y})$.     
    Let $Bad(x,\cW,f)$ be the event that $x$ is $\gamma$-distinguishing for $\cW$ on $f$. Then 

     $$ \Pr_{(x,y,F^{k,s}_{x,y})\sim\cN}[Bad(x,\cW,F^{k,s}_{x,y})] = \Pr_{(x,y,F^{k,s}_{x,y})\sim\cN}[Bad(y,\cW,F^{k,s}_{s,y})]$$

     However, \Cref{claim:distinguishing} implies that for every fixed $x$, $y$, $k$, and $s$, at least one of $Bad(x,\cW,F^{k,s}_{x,y})$ and $Bad(y,\cW,F^{k,s}_{x,y})$ occurs. The sum of the two terms in the equality above is thus at least 1, and the terms are therefore 
     at least $1/2$. Recall that, conditioned on $Bad(x,\cW,f)$, the probability of an output wrong by more than $\alpha$ is at least $\gamma$. We conclude that the the overall probability of a bad outcome is at least $\frac \gamma2$. Thus, setting $\gamma=\frac{1-(\delta e^{\Gamma\eps}/\eps)}{2(1+e^{\Gamma\eps})}$ completes the proof of the lemma.
    \end{proof}

    \begin{proof}[Proof of \Cref{claim:distinguishing}]
    Fix $x,y,k$, and $s$, and let $f=F^{k,s}_{x,y}$ and suppose neither $x$ nor $y$ are $\gamma$-distinguishing for $f$. We aim to prove the following contradiction 
    \[
    1=\Pr_{\cW}\left[|\cW^{f}(x)-f(x)|\leq\alpha\right]+\Pr_{\cW}\left[|\cW^{f}(x)-f(x)|>\alpha\right]<1.
    \]
    We start by applying group privacy (\Cref{claim: group privacy}) to obtain,
    \[
    \Pr_{\cW}\left[|\cW^{f}(x)-f(x)|\leq\alpha\right]\leq e^{\Gamma\eps}\left(\Pr_{\cW}\left[|\cW^{f}(y)-f(x)|\leq\alpha\right] +\frac\delta\eps\right).
    \]
    Now, since $k$ and $s$ are sampled uniformly from $\{2\alpha,4\alpha,6\alpha,\dots,\rho\}$, the intervals $[f(x)\pm\alpha]$ and $[f(y)\pm\alpha]$ are disjoint. Thus, we can upper bound the probability that $|\cW^f(y)-f(x)|\leq \alpha$ by the probability that $|\cW^f(y)-f(y)|>\alpha$. But since $y$ is not $\gamma$-distinguishing, this occurs with probability at most $\gamma$. Hence
    \[
    \Pr_{\cW}\left[|\cW^{f}(x)-f(x)|\leq\alpha\right]\leq e^{\Gamma\eps}\left(\gamma +\frac\delta\eps\right).
    \]
    However, since $x$ is not $\gamma$-distinguishing,  $\Pr_{\cW}\left[|\cW^{f}(x)-f(x)|\geq\alpha\right]\leq \gamma$. Putting it all together yields
    \[
    1=\Pr_{\cW}\left[|\cW^{f}(x)-f(x)|>\alpha\right]+\Pr_{\cW}\left[|\cW^{f}(x)-f(x)|\leq\alpha\right]\leq \gamma+e^{\Gamma\eps}(\gamma+\frac\delta\eps).
    \]
    Thus, we obtain a contradiction whenever $\gamma+e^{\Gamma\eps}(\gamma+\frac\delta\eps)<1$. Rearranging terms, we see that one of $x$ or $y$ must be $\gamma$ distinguishing for all $\gamma<\frac{1-(\delta e^{\Gamma\eps}/\eps)}{1+e^{\Gamma\eps}}$.   
    \end{proof}

\section{General Partially-Ordered Sets}
\label{sec:posets}

In this section, we show how our privacy wrappers can be implemented for functions over more general domains. We consider the following three examples: multisets with adjacency via insertion or deletion of an element, hypergraphs with adjacency defined by insertion or deletion of a vertex, and hypergraphs with adjacency defined by insertion or deletion of an edge.

\begin{proposition}
    \label{prop:posets}
    All of our privacy wrappers (\Cref{thm:generalized-shifted-inverse-mechanism,thm: subset extension,thm: small diameter,thm:GenShi-with-nice-noise}) can be implemented for any partially ordered domain of datasets $(\D,\leq)$ that satisfies:
    \begin{enumerate}
        \item There exists a unique minimum element in $\D$ denoted $\emptyset$.
        \item There is a function $\size:\D\to\Znonneg$ such that, for all $u \in \D$, the partial order on the down neighborhood of $u$ is isomorphic to a hypercube $\{0,1\}^{\size(u)}$.  
        \item There exists a neighbor relation $\sim$ such that $u\sim v$ for all $u,v\in \D$ such that $v\leq u$ and $\size(v)=\size(u)-1$.  
    \end{enumerate}
\end{proposition}
\begin{proof}[Proof Sketch]
    All proofs in Sections \ref{sec:generalized-shifted-inverse}, \ref{sec:extension_based}, and \ref{sec:GenShI-with-nice-noise} proceed by fixing neighbors $u,v\in\dom$ and reasoning about their down neighborhoods. Hence, the statements hold for any partially ordered domain $(\D,\leq)$ that satisfies the above properties. 
\end{proof}

Below, we give some examples of spaces that satisfy the requirements of \Cref{prop:posets}. These spaces are also considered by \cite{FangDY22}. 

\paragraph{Multisets} Let $\D$ be the set of finite multisets of some universe $\cU$ with order given by $\subseteq$. 
For this partial order to satisfy the conditions of \Cref{prop:posets}, we make a syntactic change: we represent each multiset $x\in \D$ as a finite set $\phi(x)$ in $\N\times \cU$, replacing each item $s$ in $x$ with a pair $(i,s)$, for distinct indices $i \in [|x|]$. Every subset $u\subseteq \phi(x)$ can be mapped back to a multiset $\psi(u)$ that is contained in $x$. The map $\psi$ is not injective, but does preserve adjacency and size. Furthermore, any function with domain $\D$ can be viewed as a map whose domain is finite subsets of $\N\times \cU$ via composition with $\psi$. With this change the requirements of \Cref{prop:posets} are satisfied.  

\paragraph{Hypergraphs with node privacy} Define the set of \emph{hypergraphs} $\cG$ as follows: A hypergraph $G\in\cG$ is given by a pair $(V(G),E(G))$ where $V$ is a finite set of \emph{vertices}, and $E(G)$ is a collection of subsets of $V(G)$ called \emph{hyperedges}. Define the order $\leq$ on $\cG$ by $H\leq G$ if $H$ is a vertex induced subgraph of $G$. More formally, $H\leq G$ if $V(H)\subseteq V(G)$ and $E(H) \subseteq E(G)$ is the set of edges $e$ from $ E(G)$ such that
$e\subseteq V(H)$. Then $(\emptyset,\emptyset)$ is the unique minimal element. Define the neighbor relation $\sim$ by $H\sim G$ if $H\leq G$ and $V(H)=V(G)\setminus\{v\}$ for some $v\in V(G)$. For all $G$, the down neighborhood of $G$ under this ordering is isomorphic to the $|V(G)|$ dimensional hypercube, and the requirements of \Cref{prop:posets} are satisfied.

\fi %

\addcontentsline{toc}{section}{Acknowledgments}
\section*{Acknowledgments}

We are grateful to Jonathan Ullman for helpful conversations and discussion of our results, and notably their application to parameter estimation in Erdős–Rényi graphs.

\addcontentsline{toc}{section}{References}

\ifnum\stoc=1
\bibliographystyle{ACM-Reference-Format}
\else
\bibliographystyle{alpha}
\fi
\bibliography{references}

\ifnum\stoc = 0 %

\section*{Appendix}
\appendix

\section{Applications of Our Privacy Wrappers}
\label{sec:applications}

\subsection{Average of Real-valued Data}
\label{sec:average}
As a simple example, we consider computing the average of real-valued data. In particular, \Cref{cor:average} states that if all data points lie in in an interval $\pm\sigma$ around the mean $\mu$, then our privacy wrappers will release a very accurate answer. We consider two privacy wrappers given query access to the average function $\avg$.

\begin{corollary}[Average of real-valued data]
    \label{cor:average}
    Let $\consta>0$ be a sufficiently large constant. Fix parameters $\eps,\delta\in(0,1)$. Let $\cW_1$ and $\cW_2$ denote the $(\eps,\delta)$-privacy wrappers given by \Cref{thm:generalized-shifted-inverse-mechanism} (for the automated sensitivity detection setting) and \Cref{thm: subset extension} (for the claimed sensitivity bound setting), respectively.
    Set $\lambda_1=\frac1\eps\log\frac1\delta\cdot \exp\paren{{\consta\log^*|\cY|}}$ and $\lambda_2=\frac\consta\eps\log\frac1\delta$, and let $\sigma\geq 0$. For every  $x\in\univ^n$ such that $|u-v|\leq \sigma$ for all $u,v\in x$:
    \begin{enumerate}
    \item\label{cor:average-1} If $n\geq \lambda_1$ then
    \(
    \cW_1^{\avg}(x)\in\brackets{\avg(x)\pm \frac{\consta\lambda_1\sigma}{n-\lambda_1}}.
    \)
    \item\label{cor:average-2} If $n\geq \lambda_2$ then
    \(       \cW_2^{\avg}(x,\sigma)=\avg(x)+\Lap\bparen{\frac{\consta\sigma}{\eps(n-\lambda_2)}}.
    \)
    \end{enumerate}
\end{corollary}

The first wrapper has no knowledge of $\sigma$; it adapts automatically to the scale of the data. The second wrapper requires $\sigma$ as input, but provides a stronger guarantee on the output distribution---it is symmetric around $\avg(x)$ with a known distribution. 

\begin{proof}
    \Cref{cor:average-1} follows by observing the $\lambda_1$-down sensitivity of the average function is at most $\frac{\lambda_1\sigma}{n-\lambda_1}$. 
    (To see why, let $y$ be a subset of $x$ of size $n-\lambda_1$. Without loss of generality, assume that $\avg(x)=0$. The sum of elements in $y$ lies in $[-\sigma\lambda_1, \sigma\lambda_1]$, and thus the absolute value of its average is at most $\frac{\lambda_1\sigma}{n-\lambda_1}$.)
    We can then apply the guarantee of \Cref{thm:generalized-shifted-inverse-mechanism}. \Cref{cor:average-2} follows from observing that the Lipschitz constant of the average function on $\DN_{\lambda_2}(x)$ is $O(\frac{\sigma}{n-\lambda_2})$, and applying \Cref{thm: subset extension}.
\end{proof}

\subsection{User-Level Private Convex Optimization in One Dimension}
\label{sec:ERM}

In this section, we show how our privacy wrappers immediately yield improvements upon the convex optimization algorithms of \cite{GhaziKKMMZ23-user-sco} for one dimensional parameter spaces. Before stating our improvements, we first define user-level privacy and convex optimization. Define a \emph{dataset collection} $x$ as a set 
$\{x_1,\dots, x_n\}$ of smaller datasets corresponding to individual users, where each dataset $x_i$ is a set of $m$ elements $\{x_{i,1},\dots,x_{i,m}\}$ from an arbitrary set $\univ$---that is $x\in (\univ^m)^n$. We say two dataset collections $x,x'$ are neighbors if one can be obtained from the other by deleting the data of exactly one user. Additionally, we define \emph{$(\eps,\delta)$-user level privacy} as $(\eps,\delta)$-differential privacy with respect to the aforementioned notion of neighboring dataset collections. Next, we define convex optimization, the particular problem we will focus on is known as empirical risk minimization. Using the terminology of \cite{GhaziKKMMZ23-user-sco}, a convex optimization problem over parameter space $\cY\subseteq\R^d$ and domain $\univ$, is specified by a \emph{loss function} $\ell:\cY\times\univ\to\R$ that is convex in the first argument. The loss function $\ell$ is $G$-Lipschitz if $\norm{\nabla_\theta\ell(\theta,v)}{}\leq G$ for all $\theta\in\cY$ and $v\in\univ$. Solving the \emph{empirical risk minimization} (ERM) problem corresponds to minimizing the \emph{empirical loss}, defined by $\cL(\theta,x)=\frac{1}{nm}\sum_{i\in[n]}\sum_{j\in[m]}\ell(\theta,x_{i,j})$. 

We state the main result of \cite{GhaziKKMMZ23-user-sco} below. Informally, Theorem 4.1 of \cite{GhaziKKMMZ23-user-sco} provides an algorithm for empirical risk minimization that satisfies differential privacy at the user level and requires a number of users that is independent of the dimension.\footnote{In fact, \cite{GhaziKKMMZ23-user-sco} also provide guarantees for stochastic convex optimization. Our privacy wrappers yield an identical improvement in this setting, but we will focus on empirical risk minimization for simplicity of presentation.}. Let $S_{n,m}$ be the set of permutations over $[n]\times[m]$, and for each $\pi\in S_{n, m}$ and $x\in(\univ^m)^n$ let $x^\pi$ denote the dataset collection obtained by reassigning the data of users $x_1,\dots x_n$ according to $\pi$---that is, send each element $x_{i,j}\to x_{\pi(i,j)}$. Additionally, assume $\cY$ has $\ell_2$ diameter at most~$R$.

\begin{theorem}[Theorem 4.1 \cite{GhaziKKMMZ23-user-sco}]
    \label{thm:ERM-GKKMMZ23}
    For any $G$-Lipschitz loss $\ell$ and parameter space $\cY\subset\R^{\color{blue} d}$, there exists an $(\eps,\delta)$-user level DP mechanism $\cM$ that, for all $n\geq \widetilde{\Omega}\paren{\frac{\log(1/\delta)\log(m)}{\eps}}$, outputs $\hat\theta\in\cY$ such that for all $x\in(\univ^m)^n$
    \[
    \Ex_{\substack{\pi\sim S_{n,m} \\ \hat\theta\gets \cM(x^\pi)}}\brackets{\cL(\hat\theta, x^\pi)}-\cL(\theta^*, x^\pi)\leq O\paren{\frac{RG}{n}\sqrt{\frac{d\log n}{m}}\cdot{\color{blue} \log\paren{\frac{nm}{\delta}}^{2}\paren{\frac{\log(nm)}{\eps}}^{5/2}}}.
    \]
\end{theorem}

One of the main techniques employed by \cite{GhaziKKMMZ23-user-sco} is a higher dimensional analogue of the privacy wrapper of \cite{KohliL23}; the particular guarantees can be found in Theorem 3.3 of \cite{GhaziKKMMZ23-user-sco}. Recall that \Cref{thm: subset extension} provides a privacy wrapper for real-valued functions with locality $\lambda=O(\frac1\eps\log\frac1\delta)$, that outputs $f(x)+\Lap\paren{O(\frac1\eps)}$ whenever $f$ is Lipschitz on $\DN_\lambda(x)$. Hence, we can directly substitute the algorithm given by our \Cref{thm: subset extension} for the algorithm given by Theorem 3.3 of \cite{GhaziKKMMZ23-user-sco}. This immediately yields the following improvement to \Cref{thm:ERM-GKKMMZ23} for one-dimensional parameter spaces:

\begin{theorem}[Improved ERM in one-dimension via \Cref{thm: subset extension}]
    \label{thm:ERM-SE}
    For any $G$-Lipschitz loss $\ell$ and parameter space $\cY\subset\R$, there exists an $(\eps,\delta)$-user level DP mechanism $\cM$ that, for all $n\geq \widetilde{\Omega}\paren{\frac{\log(1/\delta)\log(m)}{\eps}}$, outputs $\hat\theta\in\cY$ such that for all $x\in(\univ^m)^n$
    \[
    \Ex_{\substack{\pi\sim S_{n,m} \\ \hat\theta\gets \cM(x^\pi)}}\brackets{\cL(\hat\theta, x^\pi)}-\cL(\theta^*, x^\pi)\leq O\paren{\frac{RG}{n}\sqrt{\frac{\log n}{m}}\cdot{\color{blue} \log\paren{\frac{nm}{\delta}}\paren{\frac{\log(nm)}{\eps}}^{3/2}}}.
    \]
\end{theorem}

In fact, since \cite{GhaziKKMMZ23-user-sco}'s proof of \Cref{thm:ERM-GKKMMZ23} only uses a bound on the magnitude of the noise added by their privacy wrapper---that is, it does not require unbiased noise---we can obtain an improvement similar to that of \Cref{thm:ERM-SE} via the automated sensitivity detection privacy wrapper of \Cref{thm:generalized-shifted-inverse-mechanism}. 

\begin{theorem}[Improved ERM in one-dimension via \Cref{thm:generalized-shifted-inverse-mechanism}]
    \label{thm:ERM-autosense}
    For any $G$-Lipschitz loss $\ell$ and parameter space $\cY\subset\R$, there exists an $(\eps,\delta)$-user level DP mechanism $\cM$ that, for all $n\geq \widetilde{\Omega}\paren{\frac{\log(1/\delta)\log(m)}{\eps}}$, outputs $\hat\theta\in\cY$ such that for all $x\in(\univ^m)^n$
    \[
    \Ex_{\substack{\pi\sim S_{n,m} \\ \hat\theta\gets \cM(x^\pi)}}\brackets{\cL(\hat\theta, x^\pi)}-\cL(\theta^*, x^\pi)\leq \frac{RG}{n}\sqrt{\frac{\log n}{m}}\cdot{\color{blue} \frac{\log(nm/\delta)\log(nm)}{\eps}\cdot\exp\paren{{O\paren{\log^*|\cY|}}}}.
    \]
\end{theorem}

In the remainder of the section, we explain how to modify the proofs of \cite{GhaziKKMMZ23-user-sco}, in order to obtain \Cref{thm:ERM-SE,thm:ERM-autosense}. We encourage the reader to familiarize themselves with the proof of Theorem 4.1 in \cite{GhaziKKMMZ23-user-sco}. 

\Cref{thm:ERM-SE} follows immediately by substituting the guarantees given by \Cref{thm: subset extension} for the guarantees given by \cite{GhaziKKMMZ23-user-sco} Theorem 3.3, in their proof of \Cref{thm:ERM-GKKMMZ23}. In particular, \cite{GhaziKKMMZ23-user-sco} use Theorem 3.3 to construct an ``output perturbation" algorithm, and subsequently use the output perturbation algorithm to prove \Cref{thm:ERM-GKKMMZ23}. Since \Cref{thm: subset extension} can be used to improve the accuracy guarantee of the output perturbation algorithm, we immediately obtain the corresponding improvement to the algorithm for private empirical risk minimization.

While the proof of \Cref{thm:ERM-SE} is straightforward, the proof of \Cref{thm:ERM-autosense} requires an additional step. In \Cref{lem:smallDS}, we extend Corollary 3.7 of \cite{GhaziKKMMZ23-user-sco} in order to bound the down sensitivity of the optimal solution. Corollary 3.7 of \cite{GhaziKKMMZ23-user-sco} bounds the Lipschitz constant of the optimal solution on $\DN_\lambda(x^\pi)$ by $O\paren{\frac{G}{sn}\sqrt{\frac{\lambda\log(n) +\log(1/\beta)}{m}}}$; however, naively applying a bound on the Lipschitz constant to bound the $\lambda$-down sensitivity yields a bound that is worse by a factor of $\lambda$. In \Cref{lem:smallDS}, we show that a more careful argument allows one to save a factor of $\sqrt\lambda$. 

\begin{lemma}[Extension of \cite{GhaziKKMMZ23-user-sco} Corollary 3.7]
    \label{lem:smallDS}
    Fix $s>0$ and let $\ell$ be an $G$-Lipschitz loss such that for all $u\in\univ$ and $\theta,\theta'\in\cY$ we have $|\nabla\ell(u,\theta)-\nabla\ell(u,\theta')|\geq s|\theta-\theta'|$. Then for all $x\in(\univ^m)^n$ and $\lambda\leq n/2$, with probability at least $1-\beta$ over a choice of random permutation $\pi\in S_{n,m}$ we have,
    \[
    |\theta^*(x^\pi)-\theta^*(z)|\leq O\paren{\frac{G\lambda}{sn}\sqrt{\frac{\log(n/\beta)}{m}}},
    \]
    for all $z\in\DN_\lambda(x^\pi)$.
\end{lemma}

\Cref{thm:ERM-autosense} now follows by applying \Cref{lem:smallDS} to bound the down sensitivity of $\theta^*$, and then using the privacy wrapper of \Cref{thm:generalized-shifted-inverse-mechanism} (for the automated sensitivity detection setting), instead of the privacy wrapper given by Theorem 3.3 of \cite{GhaziKKMMZ23-user-sco}, to construct the output perturbation algorithm given by Theorem 3.1 in their paper. This improves the accuracy guarantees of the output perturbation algorithm, and hence yields the corresponding improvement to the algorithm for private empirical risk minimization. 

To see why \Cref{lem:smallDS} holds, we prove the analogue of \cite{GhaziKKMMZ23-user-sco} Equation 6 in our setting. The remainder of the proof is identical. Fix $x\in(\univ^m)^n$, and let $\theta^*=\theta^*(x)$. By the definition of $\cL$, for all $z\in\DN_\lambda(x)$ we have
\[
\nabla\cL(\theta^*,x)=\frac{n-\lambda}{\lambda}\nabla\cL(\theta^*,z)+\frac{\lambda}{n}\nabla\cL(\theta^*, x\setminus z).
\]
Since $\nabla\cL(\theta^*,x)=0$, we obtain the following version of \cite{GhaziKKMMZ23-user-sco} Equation 6,
\[
\norm{\nabla\cL(\theta^*,z)}{}=\frac{\lambda}{n-\lambda}\norm{\nabla\cL(\theta^*,x\setminus z)}{}=\frac{1}{(n-\lambda)m}\norm{\sum_{u\in x\setminus z}\nabla\ell(\theta^*,u)}{}.
\]
\Cref{lem:smallDS} now follows by first applying \cite{GhaziKKMMZ23-user-sco} Lemma 3.8 to the quantity $\norm{\sum_{u\in x\setminus z}\nabla\ell(\theta^*,u)}{}$, second, bounding $|\theta^*(x)-\theta^*(z)|$ via the hypothesis on $\ell$, and third, applying the union bound over the $n^{O(\lambda)}$ sets $z\in\DN_\lambda(x)$. See the proof of Theorem 3.6 and Corollary 3.7 in \cite{GhaziKKMMZ23-user-sco} for details.

\subsection{Estimating the Density of Random Graphs}
\label{sec:ER-density}

Borgs, Chayes, Smith and Zadik (``BCSZ'') \cite{BorgsCSZ18} give a node-differntially private algorithm that, given a graph drawn from $G(n,p)$ for unknown $p$, produces an estimate $\hat p$ such that 
	\begin{equation}\label{eq:BCSZ-bound}
		\abs{\hat p - p } 
		\leq  
		\frac {\sqrt{p}} n + O\Bparen{\sqrt{\max(p, \frac{\log n}{n})}\cdot \frac{\log^2(1/\beta) }{ \eps n^{3/2} \sqrt{\log n}}} \, ,
	\end{equation}
with  probability $1-\beta$,  when $\beta < 1/n^t$ for sufficiently large $t$. (For constant $p$ and $\beta = 1/poly(n)$, this simplifies to $\frac 1 n + \tilde O(1 / \eps n^{3/2})$.)

Let $e(G)$ denote the edge density of $G$. For a set $S,T$, let $E(S,T)$ denote the number of edges from $S$ to $T$, and $E(S)$ denote the number of edges internal to $S$. Let $e(S)$ denote the edge density within $S$, that is, $e(S)\defeq E(S)/\binom{|S|}{2}$. 

Consider the following subset of graphs on $n$ nodes:
\[
	\mathcal{H}_C = \set{G  \Big| \quad 
    \forall  S\subseteq [n] \text{ s.t. } |S|\leq \frac n 2: \quad
    e(V\setminus S ) \in  \brackets{e(G) \pm C\cdot |S|
    \cdot \sqrt{\max\bparen{e(G),\frac{\log n}{n}}}
    \cdot \sqrt{\frac{\log n}{n^3}}}}.
\]

By standard concentration arguments, graphs drawn from $G(n,p)$ lie in this set with high probability. We use the following statement, from \cite{BorgsCSZ18}, which concerns the related model $G(n,m)$ which is uniformly distributed over graphs on $n$ vertices with exactly $m$ edges.

\begin{lemma}[Corollary of Lemma 9.3 in \cite{BorgsCSZ18}]\label{lem:ER-concentration}
    For all $C>48$ and positive integers $n$ and $m$ with $m\leq \binom n 2$: If $G\sim G(n,m)$ then $G\in \mathcal{H}_C$ with probability at least $1-n^{(C/16)-3}$.
\end{lemma}

BCSZ use this lemma along with further steps (a carefully truncated noise distribution and a general extension lemma for $(\eps,0)$ differentially private algorithms \cite{BorgsCSZ18ext}) to obtain a node-private algorithm that takes as additional input an upper bound on the parameter $p$, and achieves the error rate mentioned above. 

Our general results do not require these additional steps. Specifically, \Cref{lem:ER-concentration} provides a bound on the local down sensitivity of the nonprivate estimator $e(\cdot)$. 
Applying our results on automated sensitivity detection (\Cref{thm:generalized-shifted-inverse-mechanism} with $\rangesize = \binom n 2$, since the density can only take on $\binom n 2$ distinct values), and setting $C = \Theta(\log(1/\beta)/\log n)$ in \Cref{lem:ER-concentration}, we obtain the existence of a node-private estimator matching the error of \cite{BorgsCSZ18} (\eqref{eq:BCSZ-bound}), for the setting of $\beta < 1/n^t$ and sufficiently large $t$.

\section{Utility Analysis of Our Version of Kohli-Laskowski's TAHOE}
\label{sec:Kohli-Laskowski}

We analyze the accuracy of a modified version of TAHOE, the privacy wrapper given by \cite{KohliL23}. While their construction is for vector-valued functions, we will focus on the special case of real-valued functions. To facilitate the analysis, we modify TAHOE to use the standard Laplace mechanism instead of the tailored noise distribution from \cite{KohliL23}, and we also use some of our techniques and notation from \Cref{sec:subset-extension}. First, recall that $\stabsetlhf$ denotes the set of subsets of $x$ with size at least $h$ that are $\ell$-stable with respect to $f$ (\Cref{def: stability,def:stab-functions}). And second, we will use \Cref{claim:m-l}, which states that the sizes of the maximum $\ell$-stable subsets of neighboring datasets differ by at most one.  

\begin{proposition}[Modified TAHOE]
\label{prop:KL-analysis}
    Let $a>0$ be a sufficiently large constant. For every universe $\cU$, privacy parameters $\eps>0,\delta\in(0,1)$, and Lipschitz constant $c>0$, there exists an $(\eps,\delta)$-privacy wrapper $\cW$ over $\cU$ with noise distribution $\Lap\paren{\frac {a\cdot c}{\eps^2}\ln\frac1\delta}$ for all $c$-Lipschitz functions $f:\dom\to\R$ and all $x\in\dom$. Moreover, $\cW$ is $O\paren{\frac1\eps\log\frac1\delta}$-down local and has query complexity $|x|^{O\paren{\frac1\eps\log\frac1\delta}}$ for all $x\in\dom$.
\end{proposition}

While \cite{KohliL23} prove privacy guarantees for their construction, they give no formal accuracy guarantees. In \Cref{alg:modified-TAHOE}, we present a modified version of their construction that facilitates the accuracy analysis. We also prove the privacy of our modified version. 

\begin{proof}

Below we present \Cref{alg:modified-TAHOE}, and argue that it is a privacy wrapper with the locality, privacy, and accuracy guarantees stated in \Cref{prop:KL-analysis}.
    
\begin{algorithm}[H]
    	\begin{algorithmic}[1]
    		\caption{\label{alg:modified-TAHOE} Modified TAHOE}
            \Statex \textbf{Parameters:} Privacy parameters $\eps>0$ and $\delta\in(0,1)$
    	    \Statex \textbf{Input:} $x\in\dom$, query access to $f:\dom\to\R$
    	    \Statex \textbf{Output:} $y\in\R\cup\{\perp\}$
            \State set $\eps_0\gets\frac\eps4$ and $\delta_0\gets\delta/3$ and $\tau\gets\lceil\frac1{\eps_0}\ln\frac1{\delta_0}\rceil$
            \State \textbf{release} $\ell\gets|x|-11\tau-r_1$ where $r_1\sim \TLap\paren{\frac1{\eps_0},\tau}$ 
            \Comment{\Cref{def: laplace}}
            \State $h\gets |x|-2\tau-r_2$ where $r_2\sim\TLap\paren{\frac{2}{\eps_0},2\tau}$ \Comment{$h$ is not released}
            \If{$\stabsetlhf(x)=\emptyset$} \Return $\perp$ \Comment{\Cref{def:stab-functions}}
            \Else{} \Return $f(u)+Z$ where $u=\arg\max\{|v| : v\in \stabset{\ell}{|x|-4\tau}{f}(x)\}$, %
            and $Z\sim\Lap\paren{\frac{10\tau}{\eps_0}}$
            \EndIf
            \Comment{If more than one such $u$ exists then pick one arbitrarily}
    	\end{algorithmic}
    \end{algorithm}

We will use the notation $\approx_{\eps,\delta}$ from \Cref{def:indistinguishable}, and $\maxstabl$ from \Cref{def: subset extension proxy} throughout the proof. Let $\cW$ denote \Cref{alg:modified-TAHOE}. To analyze $\cW$, it will be convenient to break the mechanism down in steps as in the proof of \Cref{thm: subset extension}. 
\begin{enumerate}
    \item Let $\cL(x)$ denote the mechanism which releases $\ell\gets|x|-11\tau-r_1$ where $r_1\sim \TLap\paren{\frac1{\eps_0},\tau}$, and let $\widehat\cL(x)$ denote the set of possible outputs of $\cL(x)$. 
\end{enumerate}
Additionally, for all fixed $\ell\in\Z$,
\begin{enumerate}
    \item Let $\cT_\ell(x)$ denote the following mechanism: set $h\gets |x|-2\tau-r_2$ where $r_2\sim\TLap\paren{\frac{2}{\eps_0},2\tau}$; return $b\gets \Ind\brackets{\stabsetlhf(x)\neq\emptyset}$.
    \item Let $\cA_\ell(x)$ be the mechanism which returns $g_\ell(x)+Z$ where $g_\ell(x)=f(u)$ for $u=\arg\max\{|u| : u\in\stabset{\ell}{|x|-4\tau}{f}(x)\cup\emptyset\}$, and $Z\sim\Lap\paren{\frac{10\tau}{\eps_0}}$.
    \item Let $\cP_\ell(x)$ be the mechanism which runs $\cT_\ell(x)$ and outputs $\perp$ if $\cT_\ell(x)=0$ and outputs $\cA_\ell(x)$ otherwise.
\end{enumerate}

By inspection of \Cref{alg:modified-TAHOE}, one can easily see that $\cW(x)$ is equivalent to the following mechanism: Set $\ell\gets\cL(x)$ and output $\cP_\ell(x)$.

\begin{lemma}[Privacy for fixed $\ell$]
    \label{lem:KL_fixed_ell}
    Fix neighbors $x,y\in\dom$ and $\ell\in\widehat\cL(x)\cup\widehat\cL(y)$. Then $\cP_\ell(x)\approx_{3\eps_0,2\delta_0}\cP_\ell(y)$.  
\end{lemma}
\begin{proof}
    We prove the lemma via the following two claims. For each $\ell\in\Z$,  define the set of ``good'' points 
    \[
    G_\ell=\set{(x,y) \colon \maxstabl(x)\geq |x|-4\tau \wedge \maxstabl(y)\geq |y|-4\tau}.
    \]
    \begin{claim}
    \label{claim:KL_T_privacy-1}
       In the setting of \Cref{lem:KL_fixed_ell}, we have $\cT_\ell(x)\approx_{\eps_0,\delta_0} \cT_\ell(y)$. 
    \end{claim}
    \begin{proof}
        Observe that $\cT_\ell$ is a postprocessing of the mechanism which on input $|x|$ releases $\maxstabl(x)-|x| + 2\tau + r_2$ where $r_2\sim \TLap\paren{\frac{2}{\eps_0},2\tau}$. By \Cref{claim:m-l} the function $\maxstabl(\cdot) - |\cdot|$ has sensitivity at most $2$, and thus by the Laplace mechanism, $\cT_\ell(x)\approx_{\eps_0,\delta_0} \cT_\ell(y)$. 
    \end{proof}
    
    \begin{claim}
    \label{claim:KL_T_privacy-2}
        In the setting of \Cref{lem:KL_fixed_ell}, if $(x,y)\not\in G_\ell$ then $\Pr\brackets{\cT_\ell(x)=1},\Pr\brackets{\cT_\ell(y)=1}\leq \delta$.
    \end{claim}
    \begin{proof}
        Suppose $x\subset y$. Then by \Cref{claim:m-l} we have $\maxstabl(y)\geq \maxstabl(x)\geq \maxstabl(y)-1$. If $\maxstabl(x)<|x|-4\tau$ then $\maxstabl(y)\leq \maxstabl(x)+1<|y|-4\tau$, and thus both $\stabsetlhf(x)=\stabsetlhf(y)=\emptyset$ for all possible choices of $h$ in \Cref{alg:modified-TAHOE}. On the other hand, if $\maxstabl(y)<|y|-4\tau$, then $\maxstabl(x)\leq |x|-4\tau$. In this case, $\stabsetlhf(x)\neq \emptyset$ if and only if $h=|x|-4\tau$. Since this occurs with probability at most $\delta$, we have $\Pr\brackets{\cT_\ell(x)=1}\leq \delta$ which completes the proof. 
    \end{proof}

    To complete the proof of the lemma, consider the following two cases. 

    \subparagraph{Case 1.} $(x,y)\in G_\ell$. In this case, the sets $\stabset{\ell}{|x|-4\tau}{f}(x)$ and $\stabset\ell{|y|-4\tau}{f}(y)$ are nonempty. Observe that for all $u\in \stabset{\ell}{|x|-4\tau}{f}(x)$ and $v\in\stabset{\ell}{|y|-4\tau}{f}(y)$, we have $|u\cap v|\geq \ell$. Moreover, since $|u|\geq |x|-4\tau$, and $|v|\geq |y|-4\tau$, we must have $|u\setminus (u\cap v)|\leq 4\tau+1$ and $|v\setminus (u\cap v)|\leq 4\tau+1$. Thus, since $u$ and $v$ are $\ell$-stable, we have $|f(u)-f(v)|\leq 8\tau+2\leq 10\tau$. The Laplace mechanism now guarantees that $\cA_\ell(x)\approx_{\eps_0,0} \cA_\ell(y)$. Since $\cT_\ell(x)\approx_{\eps_0,\delta_0} \cT_\ell(y)$ by \Cref{claim:KL_T_privacy-1}, DP composition and postprocessing (\Cref{fact:composition} and \Cref{fact:postprc}) imply that $\cP_\ell(x)\approx_{3\eps_0,\delta} \cP_\ell(y)$.   

    \subparagraph{Case 2.} $(x,y)\not\in G_\ell$. By \Cref{claim:KL_T_privacy-2}, both $\cT_\ell(x)\approx_{0,\delta_0}\perp\approx_{0,\delta_0}\cT_\ell(y)$. Thus, by postprocessing, $\cP_\ell(x)\approx_{0,2\delta_0}\cP_\ell(y)$. 

    Thus, in both cases we have $\cP_\ell(x)\approx_{3\eps_0,2\delta_0} \cP_\ell(y)$.
\end{proof}

To see why $\cW$ is private, we can simply apply DP composition to the mechanisms $\cL$ and $\cP_\ell$. By \Cref{fact: laplace mechanism}, mechanism $\cL$ is $(\eps_0,\delta_0)$-DP, and thus, by \Cref{lem:KL_fixed_ell} and composition the mechanism $\cW$ is $(4\eps_0,3\delta_0)$-DP.  

To see why the accuracy guarantee holds, observe that if $x$ is $\ell$-stable with respect to $f$, then $x=\arg\max\{|v| : v\in\stabset{\ell}{h-4\tau}{f}(x)\}$, and hence $\cW^f(x)$ outputs $f(x)+\Lap\paren{\frac{10\tau}{\eps_0}}$.

\end{proof}

\section{Small-Diameter Subset Extension Mechanism}
\label{sec: small diameter}

In this section, we construct a mechanism for the claimed sensitivity bound setting, that, in addition to a claimed sensitivity bound, is provided with a range $[0,r]$, for which it must be accurate. \Cref{thm: small diameter} states that for functions with range $[0,r]$, there is an $(\eps,0)$-DP privacy wrapper that is $O(r)$-down local and has noise distribution $\Lap(O(\frac1\eps))$ for all $x\in\dom$ and Lipschitz $f$. By \Cref{thm: query lower bound}, the small diameter subset extension mechanism, given by \Cref{thm: small diameter}, has optimal query complexity for the setting of small $r$.  

\begin{theorem}[Small diameter subset extension mechanism]
    \label{thm: small diameter}
    There exists a constant $a>0$ such that for every universe $\cU$, privacy parameter $\eps>0$, range diameter $r>0$, and Lipschitz constant $c>0$, there exists an $(\eps,0)$-privacy wrapper $\cW$ over $\cU$ with noise distribution $\Lap\paren{\frac{a\cdot c}{\eps}}$ for all $c$-Lipschitz $f:\dom\to[0,r]$ and all $x\in\dom$. Moreover, $\cW$ is $O\paren{\frac rc}$-down local and has query complexity $|x|^{O\paren{\frac rc}}$ for all $x\in\dom$.  
\end{theorem}

Note that if the analyst provides a value $r'<r$ in attempt to fool the mechanism and cause a privacy violation, we can effectively truncate $f$ to the range $[0,r']$ by setting query answers $f(x)$ to $\min\{f(x),r'\}$. Moreover, if $r'=r$ (i.e., the client is honest) then for all queries $x$, we have $f(x)=\min\{f(x),r'\}$, so the accuracy guarantees of the mechanism are unaffected.

At a high level, the construction of the small diameter subset extension mechanism is similar to that of the ``filter mechanism" introduced by \cite{JhaR13}. The filter mechanism leverages techniques from the sublinear time algorithms literature to construct a local Lipschitz reconstruction algorithm (called a local Lipschitz filter). Local Lipschitz reconstruction is a special case of the local reconstruction paradigm introduced in \cite{SaksS10}. In the local reconstruction paradigm, an algorithm $\cA$ gets query access to a function $f$ and ``enforces" some property $P$ in the following sense: On input $x$, the algorithm outputs value $y_x$ such that $y_x=f(x)$ whenever $f$ satisfies $P$. Moreover, the function defined by $\{(x,y_x) : \text{for all $x$ in domain of $f$}\}$---that is, the outputs of $\cA$---satisfies $P$.

To prove \Cref{thm: small diameter}, we construct a Lipschitz reconstruction operator and use it to design an $(\eps,0)$-privacy wrapper that achieves optimal accuracy for the class of $c$-Lipschitz functions $f:\dom\to[0,r]$, and that is $(r/c)$-down local. As a corollary of our construction, we obtain a deterministic local Lipschitz filter, \Cref{cor: lipschitz filter}, that, on input $x$, only queries $f$ on subsets of $x$, and for bounded-range functions, only queries $f$ on the $(r/c)$-down neighborhood of $x$.

\begin{corollary}\label{cor: lipschitz filter}
    For every universe $\univ$, there exists a deterministic algorithm $\cA$ that gets as input a point $x\in\dom$, parameters $r,c$, and query access to a function $f:\dom\to[0,r]$, and produces output $y_x\in\R$ such that:  
    \begin{enumerate}
        \item If $f$ is $c$-Lipschitz on $\DN_{r/c}(x)$ then $y_x=f(x)$.
        \item For all $f:\dom\to[0,r]$, the function $\{(x,y_x):x\in\dom\}$ is $14c$-Lipschitz.
    \end{enumerate}
    Moreover, $\cA$ is $\frac r c$-down local and has query complexity $|x|^{O(r/c)}$.
\end{corollary}

For bounded-range functions, the query complexity of our local Lipschitz filter is smaller than the query complexity of the Lipschitz reconstruction algorithm given by \cite{CummingsD20} (discussed in \Cref{sec:wrapper-with-provided-c}) whenever $r/c<|x|$. Moreover, by the lower bound of \cite[Theorem 4.1]{LangeLRV25}, its query complexity is optimal among local Lipschitz filters, even when the domain of the function in infinite.

\subsection{Proof of \Cref{thm: small diameter}}
In this section, we construct the small diameter subset extension mechanism (\Cref{alg: small diameter}) and use it to prove \Cref{thm: small diameter}. The main idea in the construction is to define a ``Lipschitz reconstruction" operation that, given query access to a function $f$ and a point $u$, outputs a value $y_u$ such that the following two conditions hold: First, if $f$ is Lipschitz then $y_u=f(u)$, and second, $|y_u-y_v|$ is bounded above by a fixed constant for all neighbors $v$ of $u$. Informally, one can use the above operation to construct a private mechanism as follows: On input $u$, query $f$ to compute $y_u$ and then output $y_u+\Lap(O(\frac1\eps))$. We use these ideas to prove \Cref{thm: small diameter} below. 

One immediate issue that we must circumvent, is that the conditional monotonization operator given by $\cmonf(x)=\frac12(f(x)+|x|)$, blows up the diameter of $f$. In order to avoid this issue, and take advantage of the bounded range of $f$, we define the level-$\ell$ conditional monotonization operator $\cmonfl$.    

\begin{definition}[Level-$\ell$ conditional-monotonization $\cmonfl$]
    \label{def: cmonfl}
    Fix $f:\dom\to\cY$ where $\cY\subseteq\R$. For all $\ell\in\Z$, let the level-$\ell$ conditional-monotonization of $f$ be the function 
    \[
        \cmonfl(x)=\max\{\tfrac{1}{2}(f(x)+|x|-\ell), \ \inf(\cY)\}
    \]
\end{definition}

As in \Cref{lem: cmonf}, we argue that $\cmonfl$ is Lipschitz and monotone whenever $f$ is Lipschitz. 

\begin{lemma}[Lispchitz to monotone Lipschitz]
    \label{lem: cmonfl}
    Fix a function $f:\dom\to\R$, a point $x\in\dom$, and an integer $\tau\in\Z$. If $f$ is Lispchitz on $\DN_\tau(x)$ then, for all $\ell\in\Z$, the function $\cmonfl$ is Lipschitz and monotone on $\DN_\tau(x)$.
\end{lemma}
\begin{proof}
Suppose $f$ is Lipschitz on $\DN_\tau(x)$. Let $u,v\in\DN_{\tau}(x)$ be neighbors such that $v\subset u.$ Consider the function $g(x)=f(x)+|x|.$ Since $f$ is Lipschitz, $f(u)-f(v)$ is in $[-1,1]$, so $g(u)-g(v)=f(u)-f(v)+1$ is in $[0,2].$ Thus, $g(x)$ is monotone and 2-Lipschitz. Consequently, $g'(x)=\frac12(f(x)+|x|-\ell)$ is monotone and Lipschitz. Since $\cmonfl$ is the maximum of a monotone Lipschitz function and a constant function, it is monotone and Lipschitz.
\end{proof}

\begin{proof}[Proof of \Cref{thm: small diameter}]

Our main tool in the proof of \Cref{thm: small diameter} is the following ``proxy" function. It uses $\stablhf$ and $\cmonfl$, defined in \Cref{def: stability}, and \Cref{def: cmonfl}.
\begin{definition}[Proxy function $\proxyPtf$]
    \label{def: small diameter proxy}
    Let $f:\dom\to[0,r]$ and fix $\tau\in\N$. Define the function
    \[
     \proxyPtf(x)=\Ex_{\substack{\ell\sim\{|x|-2\tau,\dots,|x|-\tau\} 
     \\ 
     h \sim \set {|x|-\tau,\dots,|x|}}}\brackets{\mlstablhf(x)}. 
     \]
\end{definition}
Intuitively, $\proxyPtf(x)$ captures the following procedure. For each $\ell$ compute the average of $\mlstablhf(x)$ over $h$, and then average the results over $\ell$. Next, we provide some intuition for the definition of $\proxyPtf$. Recall that if $f$ is Lipschitz, then \Cref{lem: cmonfl} implies that $\cmonfl$ is monotone and Lipschitz, and \Cref{lem: structure} implies that $\mlstablhf(x)=\cmonfl(x)$. Thus, $\proxyPtf$ satisfies the following important property: if $f$ is Lipschitz, then $\proxyPtf(x)$ is an average of $\cmonfl(x)$. Since $\cmonfl(x)=\frac12(f(x)+|x|-\ell)$, a simple computation suffices to recover the value of $f(x)$. Furthermore, in \Cref{lem: lipschitz 1}, we show that for all $f$, the sensitivity of $\proxyPtf$ is bounded above by roughly $1+\frac r\tau$. Hence, for a suitable choice of $\tau$, the sensitivity $\proxyPtf$ is small, and thus we can apply the Laplace mechanism to release $\proxyPtf(x)$. Next, we use $\proxyPtf$ to construct the small diameter subset extension mechanism (\Cref{alg: subset extension}) and leverage the two key properties discussed above to complete the proof of \Cref{thm: small diameter}.

    \begin{algorithm}[H]
    	\begin{algorithmic}[1]
    		\caption{\label{alg: small diameter} Small diameter subset extension mechanism}
            \Statex \textbf{Parameters:} range diameter $r\in\R$ and privacy parameter $\eps>0$
    	    \Statex \textbf{Input:} $x\in\dom$, query access to $f:\dom\to[0,r]$, sample access to $\Lap$ distribution
    	    \Statex \textbf{Output:} $y\in\R$
            \State $\tau\gets 3r$
            \State \Return $2\proxyP{\tau}{f}(x)-\frac{3(\tau+1)}{2}+Z$ where $Z\sim\Lap\paren{\frac{10}{\eps}}$
    	\end{algorithmic}
    \end{algorithm}

    Let $\cW$ denote \Cref{alg: small diameter} and consider a fixed $f:\dom\to[0,r]$, $r\in\R$ and $\eps\in(0,1)$. To prove privacy, it suffices to show that $(2\proxyP{\tau}{f}(\cdot)-\frac{3(\tau+1)}{2})$ has low sensitivity and apply the privacy guarantee of the Laplace mechanism (\Cref{fact: laplace mechanism}). 

    \begin{lemma}[$\proxyPtf$ sensitivity bound]
    \label{lem: lipschitz 1}
    Let $f:\dom\to[0,r]$ and $\tau\in\N$. Fix two neighbors $v,u\in\dom$ such that $v\subset u$. Then 
    \[
    |\proxyPtf(u)-\proxyPtf(v)|\leq 4+\frac{3r}{\tau}.
    \]
    \end{lemma}
    We defer the proof of \Cref{lem: lipschitz 1} to \Cref{sec: lipschitz 1} and complete the proof of \Cref{thm: subset extension}. Let $x,z\in\dom$ be neighbors. By \Cref{lem: lipschitz 1} and the fact that $\tau=3r$,
    \begin{align*}
    |(2\proxyP{\tau}{f}(z)-\frac{3(\tau+1)}{2})-(2\proxyP{\tau}{f}(x)-\frac{3(\tau+1)}{2})|
    =2|\proxyP{\tau}{f}(z)-\proxyP{\tau}{f}(x)|
    \leq 2(4+\frac{3r}{\tau})\leq 10.
    \end{align*}
    By the privacy of the Laplace mechanism (\Cref{fact: laplace mechanism}), the algorithm $\cW^f$ is $(\eps,0)$-DP. 
    
    Next, we prove the accuracy guarantee. By \Cref{lem: cmonfl}, if $f$ is Lipschitz then $\cmonfl$ is monotone and Lipschitz. Recall that \Cref{lem: structure} states that if $h\in\{\ell,\dots,|x|\}$ then $\stablhf(x)=f(x)$ for all Lipschitz and monotone $f$. %
    Applying Lemmas~\ref{lem: structure} and~\ref{lem: cmonfl} yields $\mlstablhf(x)=\cmonfl(x)$. Since the range of $f$ is $[0,r]$ and $\ell<x$, we have that $\cmonfl=\frac12(f(x)+|x|-\ell)$. Thus,
    \[
    \proxyPtf=\Ex_{\substack{\ell\sim\{|x|-2\tau,\dots,|x|-\tau\} 
     \\ 
     h \sim \set {|x|-\tau,\dots,|x|}}}\brackets{\cmonfl(x)}=\frac12(f(x)+|x|-(|x|-\frac{3(\tau+1)}{2}))=\frac12(f(x)+\frac{3(\tau+1)}{2}),
    \]
    where $(|x|-\frac{3(\tau+1)}{2})$ is the expected value of $\ell$. Hence, $2\proxyPtf(x)-\frac{3(\tau+1)}{2}=f(x)$. The down local guarantee follows from the setting of $\tau=3r$ and by inspection of the definition of $\proxyPtf$.
\end{proof}

\subsubsection{Proof of $\proxyPtf$ Sensitivity Bound (\Cref{lem: lipschitz 1})}
\label{sec: lipschitz 1}

For all $\ell\in\Z$ and $x\in\dom$, let $\proxyPinltf(x)=\Ex_{h\sim\{|x|-\tau,\dots,|x|\}}[\mlstablhf(x)]$. Notice that $\proxyPtf$ is the average over $\ell$ of $\proxyPinltf$. The proof proceeds in two steps. First, we bound the sensitivity of $\proxyPinltf$, and then we bound the sensitivity of $\proxyPtf$. The essence of the proof is using the interleaving relationship $\stab{\ell}{h+1}{f}(u)-1\leq\stab{\ell}{h}{f}(v) \leq \stab{\ell}{h}{f}(u)$ for neighbors $v\subset u$ proven in \Cref{lem: structure} to interleave the terms in $\proxyPinltf(u)$ and $\proxyPinltf(v)$. Then, we can use the fact that the terms are interleaved to bound the average difference between them, and thus bound the sensitivity of $\proxyPinltf$. We state and prove this formally in \Cref{claim: lipschitz 1.1} below. 

\begin{claim}
    \label{claim: lipschitz 1.1}
    Let $f:\dom\to[0,r]$ and $\tau\in\N$. Fix two neighbors $v,u\in\dom$ such that $v\subset u$. Then for all $\ell\in\{|u|-2\tau,\dots,|v|-\tau\}$,
    \[
    |\proxyPinltf(u)-\proxyPinltf(v)|\leq 3+\frac{2r}{\tau}.
    \]
\end{claim}
\begin{proof}
     In order to simplify notation, for all $h\geq \ell$, define the function $g_h$ by $g_h(x)=\mlstablhf(x)$. We expand the definition of $\proxyPinltf$ to get
    \[
    |\proxyPinltf(u)-\proxyPinltf(v)|=\Big|\Ex_{h_1\sim\{|u|-\tau,\dots,|u|\}}[g_{h_1}(u)]-\Ex_{h_2\sim\{|v|-\tau,\dots,|v|\}}[g_{h_2}(v)]\Big|.
    \]
    Notice that in both expectations the random variables $h_1$ and $h_2$ are supported on $\{|u|-\tau,\dots,|v|\}$. Thus, by the law of total expectation and the inequality $\stablhf(v)\leq\stablhf(u)$ from \Cref{lem: structure},
    \begin{equation}
    \label{eq: lipschitz 1.1}
    \begin{split}
    |\proxyPinltf(u)-\proxyPinltf(v)|
    &
    \leq\Ex_{h\sim\{|u|-\tau,\dots,|v|\}}\Big[g_h(u)-g_h(v)\Big]\\
    &
    +\Big|g_{|u|}(u)-g_{|v|-\tau}(v)\Big|\cdot \frac1{\tau+1}.
    \end{split}
    \end{equation}
    We first bound the rightmost term. By the hypothesis on the range of $f$ and the definition of $\cmonfl$, we have $\cmonfl\geq 0$. Additionally, for all $x\in\dom$ and all $\ell\leq h\leq |x|$, we have $0\leq\mlstablhf(x)\leq\frac12(r+|x|-\ell)$. Hence,
    \begin{equation}
    \label{eq: lipschitz 1.2}
    \left|g_{|u|}(u)-g_{|v|-\tau}(v)\right|\cdot \frac1\tau\leq \frac12\left(r+|u|-\ell\right)\cdot\frac1\tau\leq1+\frac r\tau.
    \end{equation}
    Next, we bound the expected value term in \eqref{eq: lipschitz 1.1}. By \Cref{lem: structure}, we have the inequality $\stablhf(v)\geq \stab{\ell}{h+1}{f}(u)-1$, and therefore,
    \[
    \Ex_{h\sim\{|u|-\tau,\dots,|v|\}}[g_h(u)-g_h(v)]\leq \Ex_{h\sim\{|u|-\tau,\dots,|v|\}}[g_h(u)-g_{h+1}(u)+1].
    \]
    Since $\stab{\ell}{h+1}{f}(x)\leq\stablhf(x)$ for all $x\in\dom$ (\Cref{lem: structure}), and since $|u|=|v|+1$, we obtain the bound
    \[
    \Ex_{h\sim\{|u|-\tau,\dots,|v|\}}[g_h(u)-g_{h+1}(u)+1]\leq1+\left(g_{|u|-\tau}(u)-g_{|u|}(u)\right)\cdot \frac1{\tau}.
    \]
    By the same reasoning as used in \eqref{eq: lipschitz 1.2},
    \[
    \left(g_{|u|-\tau}(u)-g_{|u|}(u)\right)\cdot\frac1{\tau}\leq\frac12\left(r+|u|-\ell\right)\cdot \frac1\tau\leq1+\frac r\tau,
    \]
    and therefore, 
    \begin{equation}
    \label{eq: lipschitz 1.3}
    \Ex_{h\sim\{|u|-\tau,\dots,|v|\}}[g_h(u)-g_h(v)]\leq 2+\frac r\tau.
    \end{equation}
    Combining \eqref{eq: lipschitz 1.2} and \eqref{eq: lipschitz 1.3} suffices to bound \eqref{eq: lipschitz 1.1} and obtain the conclusion that
    \[
    |\proxyPinltf(u)-\proxyPinltf(v)|\leq 3+\frac{2r}{\tau}.\qedhere
    \]
\end{proof}

Next, we complete the proof of \Cref{lem: lipschitz 1}. We first expand the definition of $\proxyPtf$ to get
        \[
        |\proxyPtf(u)-\proxyPtf(v)|=\left|\Ex_{\ell_1\sim\{|u|-2\tau,\dots,|u|-\tau\}}[\proxyPin{\ell_1}{\tau}{f}(u)]-\Ex_{\ell_2\sim\{|v|-2\tau,\dots,|v|-\tau\}}[\proxyPin{\ell_2}{\tau}{f}(v)]\right|.
        \]
        As in the proof of \Cref{claim: lipschitz 1.1}, the random variables $\ell_1$ and $\ell_2$ are both supported on $\{|u|-2\tau,\dots,|v|-\tau\}$. By the law of total expectation and the triangle inequality,
        \begin{align*}
        |\proxyPtf(u)-\proxyPtf(v)|
        &\leq\left|\Ex_{\ell\sim\{|u|-2\tau,\dots,|v|-\tau\}}[\proxyPinltf(u)-\proxyPinltf(v)]\right|\\
        &+\left|\proxyPin{|u|-\tau}{\tau}{f}(u)-\proxyPin{|v|-2\tau}{\tau}{f}(v)\right|\cdot\frac1\tau.
        \end{align*}
        We first bound the rightmost term. By the argument used to deduce \eqref{eq: lipschitz 1.1} in the proof of \Cref{claim: lipschitz 1.1}, we have $0\leq \mlstablhf(x)\leq\frac12(r+|x|-\ell)$ for all $x\in\dom$ and $\ell\leq h\leq|x|$. Inspecting the definition of $\proxyPinltf$ we see that 
        \[
        \Big|\proxyPin{|u|-\tau}{\tau}{f}(u)-\proxyPin{|v|-2\tau}{\tau}{f}(v)\Big|\cdot\frac1\tau\leq\frac12\left(r+|u|-(|u|-2\tau)\right)\cdot\frac1\tau\leq 1+\frac r\tau.
        \]
        To bound the remaining term in the inequality, we apply \Cref{claim: lipschitz 1.1} and obtain 
        \[
        \Big|\Ex_{\ell\sim\{|u|-2\tau,\dots,|v|-\tau\}}[\proxyPinltf(u)-\proxyPinltf(v)]\Big|\leq 3+\frac{2r}{\tau}.
        \]
        Combining the two bounds above yields
        \[
        |\proxyPtf(u)-\proxyPtf(v)|\leq 4+\frac{3r}{\tau}.
        \]\qed

\begin{proof}[Proof of \Cref{cor: lipschitz filter}]
    We prove the corollary for $c=1$. The case of general $c$ follows by rescaling $f$. By the accuracy analysis in the proof of \Cref{thm: small diameter}, whenever $f$ is $1$-Lipschitz we have $2\proxyPtf(x)-3\tau/2=f(x)$. Moreover, by \Cref{lem: lipschitz 1}, the function $\proxyPtf$ is $(4+\frac{3r}{\tau})$-Lipschitz. Setting $\tau=r$ the function $2\proxyPtf(x)-3\tau/2$ is $14$-Lipschitz and can be computed by querying $f$ on the set $\DN_\tau(x)$.
\end{proof}

\section{Double-Monotonization Privacy Wrapper}\label{sec:GenShI-with-nice-noise}
In this section, we present a privacy wrapper with an unbiased noise distribution with exponentially bounded tails and prove the following theorem about its guarantees. 
\begin{theorem}
    \label{thm:GenShi-with-nice-noise}
    There are constants $a,b,c>0$ such that, for every universe $\cU$, privacy parameter $\eps >0$, failure probability $\beta\in(0,1)$, and $r\geq\max(\frac\eps4,\frac c\eps\ln\frac{r+1}{\beta})$, there exists an $(\eps,0)$-privacy wrapper $\cW$ over the universe $\univ$.
    For every  function $f:\dom\to[0,r]$ and dataset $x\in\dom$, with probability at least $1-\beta$, both of the following hold:
    \begin{itemize}
        \item The mechanism $\cW^f$ is $\lambda$-down local, for $\lambda = \frac {a}\eps \ln(\frac{r}{\beta})$.
        \item  If $f$ is Lipschitz, then $\cW$ has  noise distribution $Z_{\eps}$, for a random variable $Z_\eps$ with mean 0 and an exponential tail: for all $k>0$, $\Pr\bparen{|Z_\eps| > \frac{k}{ \eps}} \leq e^{-b k}$. 
    \end{itemize}
\end{theorem}    

\Cref{alg:nice-noise}, used to prove \Cref{thm:GenShi-with-nice-noise}, first constructs a monotonized version of function $f$, then uses it to produce a list of ``offset'' values, and releases an approximation to the median ``offset'' value via the Exponential mechanism. It computes its final output by rescaling and shifting the released value. We start by explaining the monotonization transformation and the construction of the offset functions.

\subsection{Double-monotonization and Offset Functions and Their Properties}\label{sec:double-mono-and-offset}
We use (variants of) the two transformations that monotonize functions, presented in Definitions~\ref{def:monotonization} and \ref{def: cmonfl}. The transformation in \Cref{def:monotonization} monotonizes all functions. In contrast, the transformation in \Cref{def: cmonfl} monotonizes functions under the promise that they are Lipschitz, but it has the advantage of being invertible. To monotonize the black-box function, we consecutively apply both transformations. Given a function $f:\dom\to[0,\infty)$, we redefine $\cmonfl(x)=\tfrac{1}{2}(f(x)+|x|-\ell)$ (Note that this is the same as \Cref{def: cmonfl}, except in this section we do not need to ensure that $\inf(\text{range}(\cmonfl))=\inf(\text{range}(f))$), and that the level-$\ell$ monotonization operator $\monfl$ (\Cref{def:monotonization}) transforms a function $f':\dom\to[-\ell/2,\infty)$ to the function $\mon{f'}{\ell}:\dom\to[-\ell/2,\infty)$ defined by
$\mon{f'}{\ell}(x)=\max\big(\{f'(z): z\subseteq x, |z|\geq \ell\}\cup\{-\ell/2\}\big).$

\begin{definition}[Double-monotonization functions]\label{def:double-monotonization}
Fix a universe $\univ$ and  $\ell\in\N$. The {\em level-$\ell$ double-monotonization of function $f$} is the function 
 $f_\ell=\mon{\cmonfl}{\ell}.$
\end{definition}

The following properties of double-monotonization follow from the properties of the two transformations we use. We rely on them to analyze privacy, accuracy, and query complexity of \Cref{alg:nice-noise}.
\begin{observation}[Properties of double-monotonization]\label{observation:double-monotonization}
 For a level $\ell\in\N$ and a function $f:\dom\to\R$, let $f_\ell$ be the level-$\ell$ double-monotonization of $f$. Then the following properties hold:
 \begin{enumerate}
     \item\label{item1:double-mono} The function $f_\ell$ is monotone.
     \item\label{item2:double-mono} If, for some $x\in\dom$, function $f$ is Lipschitz on $\DN_{|x|-\ell}(x)$ then $f_\ell(x)=\tfrac{1}{2}(f(x)+|x|-\ell)$.
     \item\label{item3:double-mono} The value $f_\ell(x)$ can be computed by querying $f$ on all subsets of $x$ of size at least $\ell.$   \qedhere     
 \end{enumerate}
\end{observation}
\begin{proof}
\Cref{item1:double-mono} follows from the fact that $\mon{f'}{\ell}$ is monotone for all $f'$ (\Cref{item1:mono} of \Cref{observation:monotonization}). If the premise of  \Cref{item2:double-mono} holds, then by the proof of \Cref{lem: cmonfl}, function $\cmonfl$ is monotone and Lipschitz. By \Cref{item2:mono} of \Cref{observation:monotonization}, transformation $\mon{\cdot}{\ell}$ applied to a monotone Lipschitz function does not change the function. This implies \Cref{item2:double-mono}. \Cref{item3:double-mono} follows from \Cref{item3:mono} of \Cref{observation:monotonization}.
\end{proof}

Next, we define the offset functions for a  function $g$. 
\begin{definition}[Offset functions]\label{def:offset-functions}
 For each $j\in\N$, the {\em $j$-th offset} of a function $g:\dom\to\R$ is the function
 $$g_j(x)=\min_{z\in\DN_j(x)}\{g(z)-|z|+|x|-j\}.$$   
\end{definition}

We state and prove two important properties of the offsets of a function $g$. The first (\Cref{lem:offsets-for-Lipschitz}) is used for analyzing accuracy of \Cref{alg:nice-noise} and the second (\Cref{lem:offsets-interleaving}) is used for analyzing its privacy.
\begin{lemma}(Offset property for Lipschitz functions)\label{lem:offsets-for-Lipschitz}
Let $j\in\N, x\in\dom,$ and $g:\dom\to\R$ be a Lipschitz function on domain $\DN_j(x)$. Then
$g_j(x)=g(x)-j.$  
\end{lemma}
\begin{proof}
First, we show that $g_j(x)\leq g(x)-j.$ Note that $x\in\DN_j(x)$. Thus, 
$g_j(x)\leq g(x) - |x| + |x|-j = g(x)-j.$

Now, we show that  $g_j(x)\geq g(x)-j.$ Consider a point $z\in\DN_j(x).$ Since $g$ is Lipschitz on $\DN_j(x)$, we get $g(z)\geq g(x)- (|x|-|z|).$
Consequently, $g(z)-|z|+|x|-j
\geq g(x)-j.$ This inequality holds for all $z\in\DN_j(x)$. Therefore, $g_j(x)\geq g(x)-j.$   
\end{proof}

\begin{lemma}[Interleaving property for monotone functions]\label{lem:offsets-interleaving}
Let $j\in\N$ and $g:\dom\to\R$ be a monotone function. Fix neighbors $x,y\in \dom$ such that  $x\subset y$. Then the offset functions satisfy
\begin{align}\label{eq:interleaving}
 g_{j+1}(y)\leq g_j(x)\leq g_j(y).   
\end{align}    
\end{lemma}
\begin{proof}
To prove the first inequality, let $z_0\in\DN_j(x)$ be the argmin of the expression in \Cref{def:offset-functions} that evaluates to $g_j(x)$, i.e., such that  $g_j(x)=g(z_0)-|z_0|+|x|-j.$
Since $\DN_j(x)\subset \DN_{j+1}(y)$, we get
\begin{align*}
    g_{j+1}(y)&\leq g(z_0) - |z_0|+|y|- (j+1)\\
    &=g(z_0) - |z_0|+|x|- j
    =g_j(x).
\end{align*}
To prove the second inequality, let $z_1\in\DN_j(y)$ be the argmin of the expression in \Cref{def:offset-functions} that evaluates to $g_j(y)$, i.e., such that  $g_j(y)=g(z_1)-|z_1|+|y|-j.$    
If $z_1\in\DN_j(x)$ then 
\begin{align*}
    g_j(x) \leq g(z_1)-|z_1|+|x|-j
     =g_j(y)-1\leq g_j(y).
\end{align*}
Now suppose $z_1\notin\DN_j(x)$. Then $\DN_j(x)$ contains a neighbor $z$ of $z_1$ such that $z\subset z_1$. Then
\begin{align*}
    g_j(x)&\leq g(z) -|z| +|x| -j\\
    &=g(z)- (|z_1|-1) + (|y|-1) - j\\
    &= g(z)-|z|+|y|-j\\
    &\leq g(z_1)-|z|+|y|-j= g_j(y),
\end{align*}
where the last inequality holds because $g$ is monotone.
\end{proof}

\subsection{Proof of \Cref{thm:GenShi-with-nice-noise}}\label{sec:proof-of-nice-noise}

We now turn to analyzing \Cref{alg:nice-noise}. The algorithm uses, as a subroutine, the exponential mechanism MedianExpMech for privately approximating the median of a dataset. Let $\cY$ be a public interval of finite length in which we think the median lies.  MedianExpMech mechanism uses the function $\text{score}(a; y)$ that takes a potential output $a\in \cY$ and a list of real numbers $y \in \R^*$. We define $\text{score}(a; y)$ as  the smallest number of data points in $y$ that need to be changed to get a dataset for which $a$ is the median. When run with privacy parameter $\eps_0$, range $\cY\subseteq \R$, and sensitive input $y$, the mechanism returns $a$ sampled from $\cY$ with probability density proportional to $\exp\left(- \frac{\eps_0}{2}\cdot \text{score}(a; y)\right)$. (This distribution is well defined since $\cY$ is an interval of finite length.)

\begin{algorithm}[H]
	\begin{algorithmic}[1]
		\caption{\label{alg:nice-noise} Double-monotonization Privacy Wrapper}
        \Statex \textbf{Parameters:}privacy parameter $\eps>0,$ failure probability $\beta\in(0,1)$, range parameter $r\geq \frac {16} \eps \ln \frac {4r} \beta$
	    \Statex \textbf{Input:} dataset $x\in\dom$  and query access to $f:\dom\to[0,r]$  
        \Statex \Comment{If the black-box for $f$ returns a value outside the range for some query, replace the answer with the closest value in $[0,r]$}
	    \Statex \textbf{Output:} $a\in \R$ 
        \State Set $\tau\gets \frac {16} \eps \ln \frac {4r} \beta$
        \State {\bf Release} $w\gets |x| + Z$ where $Z\sim\Lap(\frac 2\eps)$   
        \State $\ell\gets \lfloor w -\tau - \frac 2 \eps \ln \frac 2 \beta \rfloor$ 
        \State Let $g=\mon{\cmonfl}{\ell}$ (the level-$\ell$ double-monotonization of $f$). \Comment{See \Cref{def:double-monotonization}}
        \State {\bf for} $j=0$ {\bf to} $\tau$ \ {\bf do}  $y_j=g_j(x)$, where $g_i$ is the $j$-the offset function of $g$ \Comment{See \Cref{def:offset-functions}}
        \State {\bf Release} $a\gets$ MedianExpMech$(y_0,\dots,y_\tau)$ executed with privacy parameter $\eps_0=\frac \eps 2$ and range $[-\frac{3\tau}{2},\frac{r +  5\tau }{2} ]$ 
        \State {\bf Return}  $2a +\tau +\ell-w$ %
	\end{algorithmic}
 \end{algorithm}

\paragraph{Privacy.}
We first analyze privacy of \Cref{alg:nice-noise}. 
Step 2 uses Laplace mechanism to release $|x|$ with privacy parameter $\frac \eps 2$. Since $|x|$ is a Lipschitz function of $x$, \Cref{fact: laplace mechanism} guarantees that this step is $(\frac \eps 2,0)$-DP. 
By \Cref{item1:double-mono} of \Cref{observation:double-monotonization}, function $g=\mon{\cmonfl}{\ell}$ defined in Step 4 is monotone (for all functions $f$ and levels $\ell$).
Therefore, \Cref{lem:offsets-interleaving} guarantees that the offset functions $g_0,\dots, g_\tau$ satisfy the interleaving property (\eqref{eq:interleaving}) for neighboring datasets. 

Let  $x$ and $x'$ be neighboring data sets.  
Since \eqref{eq:interleaving} is satisfied, we have $|\text{score}(a;y)-\text{score}(a;y')|\leq 1$, where $y$ and $y'$ are the inputs to MedianExpMech corresponding to $x$ and $x'$, respectively.\footnote{The fact that the interleaving property suffices for the score to have low sensitivity is also at the heart of the privacy of the shifted inverse mechanism \cite{FangDY22} as well as the generalizations we present in \Cref{sec:generalized-shifted-inverse}.}  
Hence, the step which calls MedianExpMech is $(\eps_0,0)$-DP. By composition (\Cref{fact:composition}), \Cref{alg:nice-noise} is $(\eps,0)$-DP for all functions $f$. 

\paragraph{Locality.}
Now we discuss locality of \Cref{alg:nice-noise}. By \Cref{item3:double-mono} of \Cref{observation:double-monotonization}, to compute the value of $g(z)$ on any input $z$, it suffices to query $f$ on the subsets of $z$ of size at least $\ell$. To compute offset functions, $g$ needs to be evaluated only on subsets of $x$. Therefore, it is sufficient to query $f$ on subsets of $x$ of size at least $\ell.$ By the standard bounds on the Laplace distribution, $|Z|\leq \frac 2 \eps\ln\frac 2 \beta$ with probability at least $1-\frac\beta 2$. It follows that $w\in[|x|-\frac2\eps\ln\frac2\beta,|x|+\frac2\eps\ln\frac2\beta]$, and therefore, $\ell\leq |x|-\tau$ and $\ell\geq |x|-\tau-\frac 4\eps\ln\frac2\beta$. By hypothesis $\tau\geq\frac1\eps\ln\frac2\beta$, and hence $\ell\in[|x|-5\tau,|x|-\tau]$. 

\paragraph{Accuracy.}
Finally, we analyze the accuracy of \Cref{alg:nice-noise} for Lipschitz $f$. When $f$ is Lipschitz, $g=\mon{\cmonf}{\ell}$
is Lipschitz and monotone and $g(z)=\frac12(f(z)+|z|-\ell)$ for all $z$ by \Cref{lem: cmonfl}. Moreover, the $j$-th offset function of $g$, denoted $g_j$, satisfies $g_j(x)=g(x)-j$ for each $j\in\{0,\dots,\tau\}$ by \Cref{lem:offsets-for-Lipschitz}. This structure allows us to analyze the output of MedianExpMech. Notice that its input $g_0(x),\dots,g_\tau(x)$ is a list of evenly spaced points $g(x)-\tau, g(x)-\tau +1 ,..., g(x)$. The median of this list is exactly $g_{\frac \tau 2}(x)=g(x)-\frac\tau2$. We now show that the score function used by MedianExpMech is nicely behaved in the interval $[g(x)-\tau, g(x)]$ with high probability.

Consider the event, which we denote $G$, that $|Z|\leq\frac 2 \eps\ln\frac 2 \beta$. Event $G$ happens with probability at least $1-\beta/2$. 
By the argument above, event $G$ implies that $\ell\in[|x|-5\tau,|x|-\tau]$. We claim that, as a result,
$[g(x)-\tau,g(x)]\subseteq[-\frac{3\tau}2, \frac{r+5\tau}{2}]$. 
To see why this is, recall that $g(z)=\frac12(f(z)+|z|-\ell)$ (for all $z$) and that $f$ is bounded in the interval $[0,r]$. 
Hence, $g(x)\leq \frac12(r+|x|-|x|+5\tau)\leq \frac{r+5\tau}{2}$, and $g(x)\geq \frac12(|x|-\ell)\geq -\frac\tau 2$.

Thus, conditioned on $G$, MedianExpMech is run on an interval $\cY$ that contains $[g(x)-\tau, g(x)]$. 
Because the inputs to MedianExpMech are evenly spaced, the score of each $a\in[g(x)-\tau, g(x)]$ is then exactly $\lfloor|g(x)-\frac\tau2-a|\rfloor$.

Now, let $A \gets\text{MedianExpMech}(g_0(x),\dots,g_\tau(x))$. For all $a\in[g(x)-\tau, g(x)]$, the probability density of $A$, denoted $p_A$, conditioned on $G$, satisfies 
\[
p_A(a|y, G)=  c_y\cdot \exp\paren{-\tfrac{\eps_0}{2}\cdot \text{score}(a; y)} 
=
c_y\cdot \exp\left(-\tfrac{\eps_0}{2}\cdot 
\Big \lfloor |g(x)-\tfrac\tau2-a| \Big \rfloor\right),
\]
where $c_y$ is a normalizing constant. 
Let $F$ be the event that $A \in [g(x)-\tau, g(x)]$. Conditioned on $F$ and $G$, the distribution $p_A$ is symmetric about $g(x)-\frac\tau2$ and has an exponentially decaying tail.  
Since $g(x)=\frac12(f(x)+|x|-\ell)$, the value $2A+\tau+\ell$ is symmetric about $f(x)+|x|$ and has similar tail behavior to $A$. Since $w\sim|x|+\Lap(\frac2\eps)$, the algorithm's final output $2A+\tau+\ell + w$ is symmetric about $f(x)$ and also has an exponentially decaying tail with scale $O\bparen{\frac 1 \eps}$.

To prove \Cref{thm:GenShi-with-nice-noise}, it remains to show that the probability of $F \cap G$ is at least $1-\beta$. We do this by showing that $\Pr[F|G]\geq 1-\frac \beta 2$.

For $a$ outside $[g(x)-\tau, g(x)]$, the score $\text{score}(a;y)$ is $\frac \tau 2$, and thus the probability density of $A$ is $c_y\cdot \exp(-\frac{\eps_0}{2} \cdot \frac{\tau}{2})$. 
MedianExpMech is run on a range of length $\frac {1}{2}r + 4\tau$, which means that there is a region of total length $\frac {1}{2}r + 4\tau - \tau = \frac {1}{2}r + 3\tau $ where the score is $\frac \tau 2$. Hence, 
\begin{align*}
\Pr[F|G]
&=\Pr[\text{score}(a;y) = \tfrac \tau 2|G]
\leq\frac{\Pr[\text{score}(a;y) = \frac \tau 2|G]}{\Pr[\text{score}(a;y) \leq 1|G]}\\
&\leq\frac{c_y (\frac 1 2 r + 3\tau) \cdot \exp(-\frac{\eps_0}{2} \cdot \frac{\tau}{2})}{c_y \cdot 2 \cdot \exp(-\frac{\eps_0}{2} \cdot 1) }
\leq 2r\cdot\exp\bparen{-\frac{\eps_0\cdot\tau}{8}},
\end{align*}
where the second inequality holds because we conditioned on $G$---which means the region with score at most one consists of an interval of length 2 centered at $g(x)-\frac \tau 2$---and the last inequality holds because $\tau\leq r$ and $\tau-2\eps\geq\frac\tau2$. Since $\tau=\frac {16}\eps\ln\frac{4r}{\beta}$ and $\eps_0=\frac\eps 2$, we obtain
$\Pr[\text{score}(a;y)\geq \frac \tau 2|G]
\leq \frac\beta 2.$ 

Finally, by the law of total probability and the fact that $\Pr[\overline G]\leq \frac\beta 2$, we see that $F\cap G$ has probability at least $1-\beta$. Conditioned on this event the output of \Cref{alg:nice-noise} has the desired distributional properties.\qed

\section{Relation to Resilience \cite{SteinhardtCV18}}
\label{sec:resilience}

Steinhardt, Charikar and Valiant introduce the notion of \emph{resilience} \cite[Definition 1]{SteinhardtCV18}, in the context of mean estimation, to the capture the idea of robustness to deletions only (as opposed to insertions).  Generalizing from means to arbitrary real-valued\footnote{Steinhardt et al.\ consider vector-valued functions; we focus on real-valued statistics here.} function $f$, given $\sigma>0$, and $\phi\in[0,1]$, we say an input $x \in \dom$ is \emph{$(\sigma,\phi)$-resilient with respect to function $f$} if the radius of $f\bparen{\DN_{\phi|x|}(x)}$ is at most $\sigma$—that is, there exists a value $\mu$ such that $|f(z) -\mu|\leq \sigma$ for all $z\subseteq x$ of size at least $(1-\phi)|x|$. This concept is closely tied to down sensitivity:  
\begin{quote}
If $x$ is $(\sigma,\phi)$-resilient with respect to the function $f$, then  $DS^f_{\phi |x|}(x) \leq 2\sigma$. Conversely, if $DS^f_{\lambda}(x) \leq \sigma$, then $x$ is $\bparen{\sigma,\frac\lambda{|x|}}$-resilient.     
\end{quote}

We can thus translate one of our results directly into the language of ``resilience": let $\cW= \cW_{\eps,\beta,R}$ be the \ASDalgshort\ privacy wrapper of \Cref{thm:generalized-shifted-inverse-mechanism} and let $\lambda= \lambda(\eps, \beta, R)$ be its down locality. For every function $f:\dom\to [R]$ and  input $x$, if $x$ is $\bparen{\sigma, \frac{\lambda}{n}}$-resilient with respect to $f$, then $\cW^f(x) \in f(x)\pm\sigma$.  

\paragraph{Using the Steinhardt et al.\ transformation.} 

Steinhardt et al.\ observe that, given any function $f$, one can transform it into a new function that is robust to both insertions and deletions whenever its input $x$ is robust to deletion.

We present a quantitatively precise version of their result here, for completeness. Given $f$ and positive parameters $\sigma, \lambda_1,\lambda_2$, with $\lambda_2\geq \lambda_1$, let $f_{\sigma,\lambda_1,\lambda_2}$ denote any function of the following form: on input $x$, if there exist a set $y\subset x$ of size at least $|x|-\lambda_1$ 
 that is $(\sigma,\frac{\lambda_2}{n})$-resilient with respect to $f$, then pick one such $y$ of the largest possible size and return a center\footnote{That is, a value $\mu$ that minimizes $\max{|f(z)-\mu|: z \in \DN_{\lambda_2}(y)}$.} $\mu$ for $f\bparen{\DN_{\lambda_2}(y)}$; if no such $y$ exists, return a special $\bot$ (undefined) value. 

\begin{lemma}[Robustness of $f_{\sigma,\lambda_1,\lambda_2}$]
    For every function $f$, parameters $\sigma,\lambda_1,\lambda_2$ with $\lambda_1\leq \lambda_2$, and inputs $x$ and $y$: 
    \begin{enumerate}
        \item %
        If $f_{\sigma,\lambda_1,\lambda_2}$ is defined (not $\bot$) on both $x$ and $y$, and $\Delta(x,y)\leq \lambda_2-\lambda_1$, then $\big|f_{\sigma,\lambda_1,\lambda_2}(x)-f_{\sigma,\lambda_1,\lambda_2}(y)\big|\leq 2\sigma$. 
        \item %
        If  $x$ is $\bparen{\sigma, \frac {\lambda_1+\lambda_2}{|x|}}$-resilient and $\Delta(x,y)\leq \lambda_1$, then $\big|f_{\sigma,\lambda_1,\lambda_2}(x)-f_{\sigma,\lambda_1,\lambda_2}(y)\big|\leq 4\sigma$. 
    \end{enumerate}
\end{lemma}

\begin{proof}
    \begin{enumerate}
    \item Let $a$ and $b$ be the subsets of $x$ and $y$ such that $f_{\sigma,\lambda_1,\lambda_2}(x)=\mu_a$ and $f_{\sigma,\lambda_1,\lambda_2}(x)=\mu_b$, where $\mu_a$ is a center of $\DN_{\lambda_2}(a)$ and similarly for $\mu_b$. We first argue that $a \cap b$ differs by less than $\lambda_2$ from both  $a$ and $b$. To see why this is, note that $a' = a \cap y = a \cap y \cap x$ satisfies $|a \setminus a'| = |x \setminus (x\cap y)| \leq \Delta(x,y)$. Also, $a' \setminus (a \cap b) = | (a \cap y) \setminus (a \cap y \cap b)| \leq |y \setminus b| \leq \lambda_1$. Thus $|a \setminus (a \cap b)|\leq \Delta(x,y) + \lambda_1 \leq \lambda_2$. 

    A symmetric argument shows that $|b \setminus (a\cap b)|\leq \lambda_2$. By the construction of $f_{\sigma,\lambda_1,\lambda_2}$, the sets $a$ and $b$ are both resilient with respect to $f$. Hence,
    both $\mu_a$ and $\mu_b$ are within $\sigma$ of $f(a\cap b)$, and thus within $2\sigma$ of each other. 

    \item The function $f_{\sigma,\lambda_1,\lambda_2}(x)$ is defined on $x$ since $x$ is resilient; let $\mu_x$ be the corresponding center. To see why $f_{\sigma,\lambda_1,\lambda_2}$ is defined  on $y$, note that the set $x\cap y$ has size at least $|x|-\lambda_1$ and so its $\lambda_2$-down neighborhood is contained in the $(\lambda_1+\lambda_2)$-down neighborhood of $x$. The resiliency of $x$ implies sufficient resiliency for $x\cap y$ to make $f_{\sigma,\lambda_1,\lambda_2}(y)$ defined.

    Let $z$ be the resilient subset of $y$ selected by $f_{\sigma,\lambda_1,\lambda_2}$, and let $\mu_z= f_{\sigma,\lambda_1,\lambda_2}(y)$ be the corresponding center. The set $z\cap x$ is a common subset of both $z$ and $x$ with size at least $|z|  - \Delta(x,y) \geq |z| - \lambda_1$. Therefore, the value $f(z\cap x)$ is within $\sigma$ of $\mu_z$.

    Bounding the distance to $\mu_x$ is a bit delicate, since $\mu_x$ is only known to be a center of $f\bparen{\DN_{\lambda_2}(x)}$, which may not include $z\cap x$. However, by assumption, the radius of the larger set $f\bparen{\DN_{\lambda_1+\lambda_2}(x)}$, which does include $z\cap x$ (which has size at least $|x|-2\lambda_1$), is also bounded by $\sigma$. Thus,
    $|\mu_x-f(x\cap z)|\leq |\mu_x-f(x)| + |f(x) - f(x\cap z)| \leq 3\sigma$, and $\mu_x-\mu_z$ is at most $4\sigma$.
    \end{enumerate}
\end{proof}

\paragraph{Combining with Resilient-to-Robust with Robust-to-Private Transformations}
One could combine this transformation with the robust-to-private transformation of Asi, Ullman, and Zakynthinou~\cite[Theorem 3.1]{AsiUZ23}---who show that any function with bounded local sensitivity (as opposed to down sensitivity) can be made differentially private---to get a nonconstructive result with a similar flavor to that of \Cref{thm:generalized-shifted-inverse-mechanism}. However, the resulting object is weaker in several respects: (a) it requires the sensitivity bound $\sigma$ as an analyst-specified parameter (as opposed to discovering it automatically, and (b)
it is not obviously local (or even computable), since it requires, in principle, considering $f(z)$ for all inputs $z$ at some positive depth $\lambda$ from the input $x$, including all supersets of $x$ of size up to $|x|+\lambda$. It also entails  weaker accuracy guarantees (by a factor of $\lambda$) than those of \Cref{thm:generalized-shifted-inverse-mechanism} in the Lipschitz setting.

\fi %

\end{document}